\documentclass[11pt]{article}
\usepackage[a4paper,bindingoffset=0.2in,%
            left=1.2in,right=1.2in,top=1.5in,bottom=1.5in,%
            footskip=.25in]{geometry}

\usepackage[T1]{fontenc}
\usepackage[utf8]{inputenc}
\usepackage{authblk}
\usepackage[numbers]{natbib}

\usepackage{amsthm,amsmath}
\usepackage{amsfonts}
\newtheorem{lemma}{Lemma}[section]

\newtheorem{thm}{Theorem}[section]
\newtheorem{cor}{Corollary}[section]
\theoremstyle{definition}
\newtheorem{definition}{Definition}[section]

\newcommand{\rhop}{\rho^{(p)}}
\newcommand{\rhou}{\rho^{(u)}}

\usepackage{footnote}
\usepackage{verbdef}
\usepackage{booktabs, makecell, tabularx}
\usepackage{float}

\makeatletter
\newtheorem*{rep@theorem}{\rep@title}
\newcommand{\newreptheorem}[2]{%
\newenvironment{rep#1}[1]{%
 \def\rep@title{#2 \ref{##1}}%
 \begin{rep@theorem}}%
 {\end{rep@theorem}}}
\makeatother

\makeatletter
\newcommand\footnoteref[1]{\protected@xdef\@thefnmark{\ref{#1}}\@footnotemark}
\makeatother

\newreptheorem{theorem}{Theorem}
\newreptheorem{lemma}{Lemma}
\usepackage{graphicx}
\graphicspath{{images/}{../images/}}
\usepackage{booktabs}

	\newcommand{\blind}{0}
	
	\usepackage{amsmath}
	\usepackage{graphicx}
	\usepackage{enumerate}
	\usepackage{url} 
	
\begin{document}
		
		\def\spacingset#1{\renewcommand{\baselinestretch}%
			{#1}\small\normalsize} \spacingset{1}

		\if0\blind
		{
			
			\title{Chasm in Hegemony: Explaining and Reproducing Disparities in Homophilous Networks}	
			
	\author[1]{Yiguang Zhang}
	\author[2]{Jessy Xinyi Han}
	\author[1]{Ilica Mahajan}
	\author[1]{Priyanjana Bengani}
	\author[1]{Augustin Chaintreau}
	\affil[1]{Columbia University}
	\affil[2]{Massachusetts Institute of Technology}

			\date{}
			\maketitle
		} \fi
		
		\if1\blind
		{

            \title{\bf \emph{IISE Transactions} \LaTeX \ Template}
			\author{Author information is purposely removed for double-blind review}
			
\bigskip
			\bigskip
			\bigskip
			\begin{center}
				{\LARGE\bf \emph{IISE Transactions} \LaTeX \ Template}
			\end{center}
			\medskip
		} \fi
		\bigskip

\begin{abstract}

In networks with a minority and a majority community, it is well-studied that minorities are under-represented at the top of the social hierarchy. However, researchers are less clear about the representation of minorities from the lower levels of the hierarchy, where other disadvantages or vulnerabilities may exist. We offer a more complete picture of social disparities at each social level with empirical evidence that the minority representation exhibits two opposite phases: at the higher rungs of the social ladder, the representation of the minority community decreases; but, lower in the ladder, which is more populous, as you ascend, the representation of the minority community improves. We refer to this opposing phenomenon between the upper-level and lower-level as the \emph{chasm effect}. Previous models of network growth with homophily fail to detect and explain the presence of this chasm effect. We analyze the interactions among a few well-observed network-growing mechanisms with a simple model to reveal the sufficient and necessary conditions for both phases in the chasm effect to occur. By generalizing the simple model naturally, we present a complete bi-affiliation bipartite network-growth model that could successfully capture disparities at all social levels and reproduce real social networks. Finally, we illustrate that addressing the chasm effect can create fairer systems with two applications in advertisement and fact-checks, thereby demonstrating the potential impact of the chasm effect on the future research of minority-majority disparities and fair algorithms.
	
	\end{abstract}
\maketitle
\pagestyle{plain}
\thispagestyle{empty}

\section{Introduction}

The "glass-ceiling" effect has multiple real-world applications; it is invoked when describing the invisible barrier that women --- or any minority group --- hit in their career as they approach the upper echelons of management \cite{cotter2001glass}\cite{morgan1998glass}. The top of the hierarchy has been well studied, whereas research on minority representation in the rest of the social hierarchy has received less attention. Having a complete characterization of social disparities at all levels of the hierarchy helps tackle questions including at what point a minority group starts experiencing a systemic disadvantage, and at what rung of the ladder -- if any -- are minorities fairly represented.

We tackle these questions leveraging real-world datasets (QQ, WhatsApp, and Instagram) in an attempt to understand the distribution of minority representation across the entire hierarchy. In order to talk about the advantage or disadvantage of the minorities, we have to agree on a measure of success in a social network. Following the conventional approach that sees network edges as the network's "social capital.", we define successful members in a friendship (unipartite) network to be people with a large number of friends, and define successful groups in a group-member (bipartite) network to be groups with many members.

Our main finding is the surprising but repeated evidence that the ratio of people belonging to a minority group initially increases as one moves up in the lower layers of the hierarchy, before it reaches a plateau and drops. We refer to this effect as a ``chasm'' because people who observe the lower or upper layer of a hierarchy might agree that a systemic bias is present but would hastily claim it is in opposite directions. This is in striking contrast to the monotonic behavior one would expect in all previous systemic models of hegemonic biases. As we prove that previous models cannot explain our observation, we also provide the first generative model that offers a simple explanation and is general enough to apply broadly.

The question we ask in this paper addresses the causes of this chasm effect. What are the mechanisms that interact with each other to create both the glass-ceiling effect and the chasm effect, and in particular, how do social networks play a role in creating these two effects?

Previous studies on the glass-ceiling effect have provided mechanisms that capture the glass-ceiling effect \cite{avin2015homophily}. However, the same mechanisms do not capture the chasm effect we have observed. In this paper, we primarily focus on bi-affiliated bipartite networks, where the network is partitioned into groups and members, and each member and each group has an independent or collective (respectively) viewpoint that favors the minority or majority. We are interested in these bipartite networks for two reasons: (1) the nature of bipartite networks is less understood but more intriguing due to their complexity; (2) many social platforms, such as WhatsApp, are now group-based where members find communities of their interests within the larger network. We analyze the interactions among a few well-observed network-growing mechanisms with a simple model to reveal the sufficient and necessary conditions for both the glass-ceiling effect and the chasm effect to both be present. We further generalize the simple model naturally and present a complete bi-affiliation bipartite network-growth model. We demonstrate our proposed model's effectiveness through both mathematical proofs and data synthesis. Our generative model is the first to capture the chasm effect in social disparities.

This study has important practical applications, especially as it puts a spotlight on structural biases in bipartite networks and hints at ways to address them. More specifically, the new idea of a chasm effect we put forward provides a foundation for allocating resources differently in diverse settings to minimize bias among those people who constitute a large portion of the population that are more disadvantaged and vulnerable. We present two examples taken from different contexts: (1) (gender fairness) we aim to provide recruiters with a better job placement strategy if they want to diversify their pool of candidates; (2) (political fairness) in politics-related group chats where conversations are not accessible outside the immediate community, we aim to show how fake-news can have more of an adverse impact on the minority population in a constrained environment.  

As a summary, our main contributions are:
\begin{itemize}
\item We prove the existence of the chasm effect with empirical evidence from real-world datasets, and characterize the phenomenon in-depth to provide a more complete picture of social parities. That is, we show that the ratio of the minority community does not decrease monotonically as we move up the hierarchy. (Section 3)
\item We analyze the interactions among network-growth mechanisms and derive the necessary mechanisms for both the chasm effect and the glass-ceiling effect to be present in bipartite networks. (Section 4)
\item We propose a complete bipartite bi-affiliation network-growth model that generalizes the necessary mechanisms discussed in Section 4. The generalized model is capable of reproducing real-world social networks. Under the generalized model, we provide proofs to show that both types of entities in the generated networks have power-law degree distributions, and specify the sufficient and necessary conditions mathematically for both the glass-ceiling effect and the chasm effect to present. (Section 5)
\item Finally, we provide two real-world applications of our findings, job advertisement and fact-checking, where the chasm effect could impact the direction of bias, thereby motivating the importance of considering the chasm effect. (Section 6) 
\end{itemize}

Those results together suggest that the chasm effect can be observed, at least frequently in online networks which may exhibit simple selective homophily dynamics, and has consequences. We urge some caution as our results do not, however, prove that the chasm is unavoidable: Some social networks (and, under some conditions, our general model) can exhibit a systemic monotonic bias against minority groups at \emph{all} level of the hierarchy.

\section{Related Work}
Social disparities and the hegemony of the majority community have been widely studied in uni-partite social networks, and it has been well-observed that disadvantages are exerted on the minority community, for example, in the case of the gender gap \cite{cotter2001glass}\cite{morgan1998glass} or rural-urban inequality \cite{qu2017glass}. It has also been shown, through homophilous preferential attachments, that structural bias in uni-partite social networks can create such disadvantages \cite{avin2015homophily} at the top of the hierarchy and the effects can be reinforced when recommendation algorithms are applied \cite{stoica2018algorithmic}. However, no existing model analyzes the structural biases that may exist beyond the top of the hierarchy.

Further, studying hegemony is no longer straightforward in bipartite networks. Often, the bipartite netorks are comprised of different types of entities and it is only meaningful to study homophily within a single entity. \emph{Projection} can convert bipartite networks back to uni-partite networks, but this loses important network information \cite{latapy2008basic}. Therefore, a model that studies hegemony directly on bipartite networks is imperative. Unfortunately, there are not many bipartite network models and even fewer studies on social disparities. Previous analytical literature \cite{borgatti1997network} and \cite{latapy2008basic} provide notations studying bipartite networks and extend several common notations in uni-partite network to bipartite networks, but those do not consider hegemony. Random graphs models like Stochastic Block Model can be used to model homophily, but do not reproduce the large range of degrees \cite{barabasi2002evolution} well-observed in social networks. Configuration models like exponential random graph models \cite{bomiriya2014topics} can be modified to study homophily but are restricted by nature to static graphs with no internal reinforcing dynamics. We hereby introduce the first generative model that can be used to analyze hegemony in bipartite networks.

	One important application of bipartite networks is \emph{fairness in fake news detection in encrypted group-member networks}. In the last decade, researchers have expended tremendous efforts attempting to automatically detect fake news by analyzing texts \cite{castillo2011information}\cite{qazvinian2011rumor}, images \cite{huh2018fighting}, propagation models \cite{liu2018early}, and more \cite{shu2017fake}\cite{feng2013detecting}. Most auto-detection methods apply only to \emph{public social media}\footnote{https://www.facebook.com/facebookmedia/blog/working-to-stop-misinformation-and-false-news} where platforms have access to all content. However, on \emph{private} platforms (such as WhatsApp, which is end-to-end encrypted), platforms are unable to proactively auto-detect misinformation due to the lack of visibility into the content. Instead, one of the ways in which they detect potentially inaccurate political news is through user reports. Due to the diversity of information and the massive volume of queries received, stories reported as fake by a large number of users are often prioritized by fact checkers \cite{babaei2019analyzing}. When there is more than one political party in the network, such detection methodologies may create \emph{unfairness} as the party with more members could have more fake news against them debunked and removed due to their advantages in reporting. To the best of our knowledge, our results are the first to tackle factors affecting fairness of fake news detection in \emph{encrypted} social media.
	
\section{Hegemony in Networks: An Unexpected Chasm}
Here we define a \emph{hegemonic} subset as one that is systematically over-represented among the tail of most popular nodes. It was shown that a majority affiliation among the nodes can become hegemonic under simple rich-get-richer and homophily dynamics \cite{avin2015homophily}. We find, among three large networks with affiliations, including for a bipartite graph, that such hegemonic subsets always exhibit a remarkable paradox: It appears, starting at small degrees, that members of the hegemonic subset are becoming scarcer as degree increases, while the fraction of members from other subset initially increases! This creates a chasm since, assuming one concentrates on a partial local observation of the degree distribution, one may hastily conclude that network growth either disproportionally favors or disfavors those in the hegemonic subset.

\subsection{Gender and political affiliations in the networks}
\textbf{QQ dataset \cite{you2015empirical}}
One of the most popular instant messengers for group chats in China is Tencent QQ, which has over 700 million active users. Users can create new groups or join existing ones. Depending on the account level of the group creator, QQ group sizes are capped at 100, 200, 500, or 1000 members. This dataset contains 274,335,183 users and 58,523,079 groups, of which 273,204,518 users have gender information and 48,676,355 groups have the complete information about the group identifier, member list, and group creation date. Females take up 42.5\% of the users in this dataset; hence, we label groups whose members are less than 42.5\% female as male-dominated groups and those with with more than 42.5\% females are classified as female-dominated \footnote{\label{note1} In the setting of groups with political-leaning, it is often the case that the group creator maintains the group to favor the creator's political affiliation. However, in the gender-based setting, it makes more sense to color groups using the gender ratio within the groups, as it is less obvious to identify the gender-leaning of a topic. We show theoretically in Appendix \ref{gender ratio color} that for the purpose of this paper, assigning groups the same affiliation as the group creators is equivalent to assigning groups affiliations by member ratios.}. 

We observe that the group size distribution of this dataset becomes discontinuous at 100, 200, and 500 members due to the imposed group-size caps. To avoid the impacts of the discontinuity, we focus our analysis on groups of sizes no larger than 100, which account for 99.2\% of all groups in the dataset and have an the average ratio of female membership in a group of 40.9\%. 

\textbf{WhatsApp dataset \cite{garimella2020images}\cite{garimella2018whatsapp}}
WhatsApp is one of the most widely-used messaging apps around the world. The WhatsApp data we use was collected over a period of 9 months from October 2018 to June 2019. It includes 2,092 groups around political conversations and 205,880 unique users. The party affiliation of each group is labeled according to the group title and some of its content by authors in \cite{garimella2018whatsapp}. Based on the ideology and relevant reports of the group's party affiliation, we characterize each group's political leaning as pro-BJP or anti-BJP where BJP stands for Bharatiya Janata Party, the current ruling party in India. To obtain sound and rigorous results, we only consider groups where the political leanings are evident\footnoteref{note1}. Once we identify the political leanings of groups, we label each user as pro-BJP if the ratio of pro-BJP groups the user joined exceeds the overall pro-BJP group ratio in the dataset and vice versa. Overall, we get 1,198 pro-BJP groups and 465 anti-BJP groups, with 62,920 pro-BJP members and 21,625 anti-BJP members sharing 897 images manually labelled as misinformation.

The data are very sparse for groups with more than 165 members in this dataset, so we restrict our analysis to groups of size less than 165. Furthermore, since WhatsApp is an end-to-end encrypted application where members have a reasonable expectation for privacy, we drop all groups with less than 52 members ($20$\% of the maximum group size).

\subsection{Evidence of reinforcing and homophilous growth\label{sec:three_mech}}
\subsubsection{Minority-major affiliation}
An unequal proportion of two affiliations arise in different identity contexts like gender (social identity) and political leaning (political identity). In the QQ dataset which illustrates the social identity aspect, female members make up 42.5\% of the population and 41.0\% of the groups are female-dominated, thus females are considered as the minority and male the majority; in our WhatsApp dataset which exhibits the political identity aspect, 25.6\% of all members and 28.0\% of all groups are anti-BJP, thus anti-BJP is denoted as the minority and pro-BJP the majority. Despite the completely different nature of the majority-minority groups in these two datasets, our later findings will show that they share some similar properties, which is worth further study. 

\subsubsection{Rich-get-richer}
The degree distribution of a network reflects how the resources and power are distributed in society. Previous studies on one-mode social networks demonstrate a ``rich-get-richer'' mechanism \cite{adamic2000power}\cite{barabasi2002evolution}, suggesting that those with more connections have an advantage in building even more connections. In bi-affiliation bipartite networks, we study each affiliation and each type of entities separately. We take the number of members within a group as degree of the group entities and take the number of groups a member joins as degree of member entities. We find a smooth slow decay in small degrees and a fast decay in large degrees for both the group size distributions and member degree distributions, exhibiting a similar ``rich-get-richer'' result as in one-mode networks. Specifically, in the QQ dataset, female-dominated groups follow a power law with power -4.00 and their male counterparts, -3.51; QQ female members follow a power law with power -3.82 and their male counterparts the same; in the WhatsApp dataset, anti-BJP groups follow a power law with power -2.67 and their pro-BJP counterparts -2.48; WhatsApp anti-BJP members and their pro-BJP counterparts follow a power law with almost identical power, -2.29 and -2.23 respectively. This ``rich-get-richer'' result on bi-affiliation bipartite networks illustrates a few basic ideas on member-group interactions: (1) members are more likely to join large groups, likely due to large groups' popularity or their potential to offer more resources; (2) this higher tendency of members to join large groups is more pronounced when joining majority groups; (3) members who are active in joining groups are more likely to join new groups than those who are less active.

\subsubsection{Homophily}

Homophily is a well-observed phenomenon that says that people tend to connect with those who are similar to them\cite{mcpherson2001birds}. To test for homophily, we count the number of minority-majority member pairs. Specifically, two members form a member pair if they are both in the same groups, and they have multiple pairs if they share multiple common groups. We count the number of member pairs in the network such that one end of the pair is a member from the minority affiliation and the other end is a member from the majority. Note that when there is no homophily in the network, the ratio of minority-majority member pairs over all member pairs is $2r(1-r)$, where $r$ is the percentage of minority members in the network. Having the actual minority-majority member pairs be less than the expected number of minority-majority member pairs is therefore an indication for homophily. 

Both the QQ dataset and the WhatsApp dataset show a strong indication for homophily in Fig.  \ref{homophily_test}, as the actual number of minority-majority member pairs (orange line) is significantly smaller than the expected value when assume no homophily (green line). Therefore, we conclude that homophily exists in bipartite networks.

\begin{figure}
\centering
\includegraphics[scale=0.5]{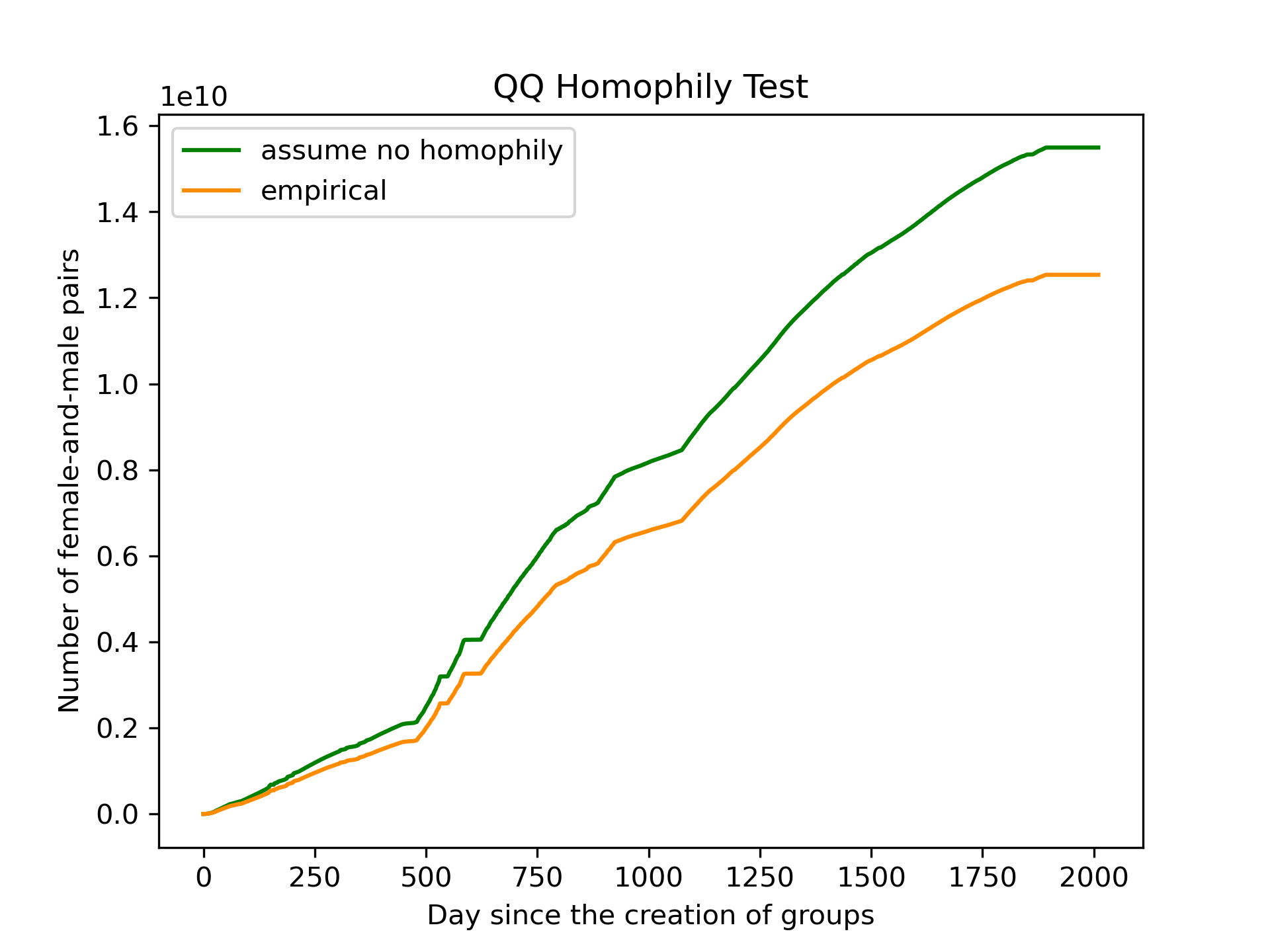}
\includegraphics[scale=0.5]{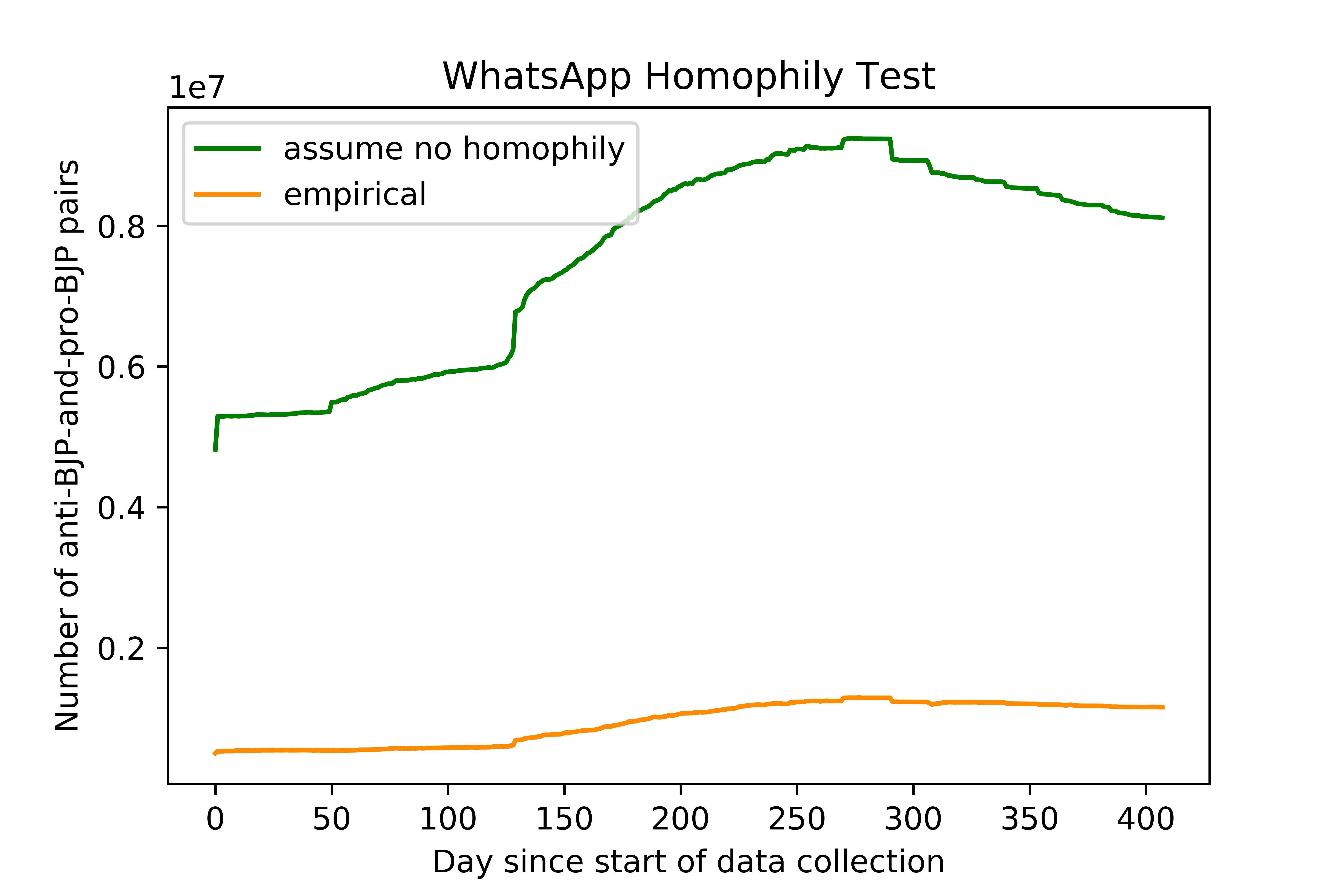}
\caption{\emph{Homophily mechanism}: we observe in both the QQ dataset, where the minority-majority imbalance often arises in the context of gender disparity, and in the WhatsApp dataset, where the imbalance often arises in the context of political parties, the network exhibits homophily. That is, people have the tendency to connect with the ones of their own affiliation.}
\label{homophily_test}
\end{figure}

\vskip 0.15in
As a conclusion, our analysis on the real-world data illustrates the following three mechanisms in bi-affiliated group-member networks:

\begin{enumerate}
\item \emph{Minority-majority affiliation:} the two affiliations have non-negligible size differences.
\item \emph{Rich-get-richer:} new members are more likely to join large groups; members who are active in joining groups are more likely to join new groups than those who are less active. 
\item \emph{Homophily:} members are more likely to join groups of their own affiliation.
\end{enumerate}

\subsection{Disparities before the glass ceiling: Chasm in Hegemony}
 
While the glass-ceiling effect depicts the under-representation of minorities at the higher rungs, we zoom out to study the minority representation across the social hierarchy. We find that at the lowest level, minorities are also under-represented and this under-representation eases as they move up the social ladder but deteriorates closer to the top. This matches the glass-ceiling effect at the higher levels. As the minority representation exhibits opposite trends when we move up in the lower rungs and in the upper rungs, we refer to this phenomenon as the ``chasm effect'' between the lower-level and upper-level.

\begin{figure}[H]
\centering
\includegraphics[scale=0.5]{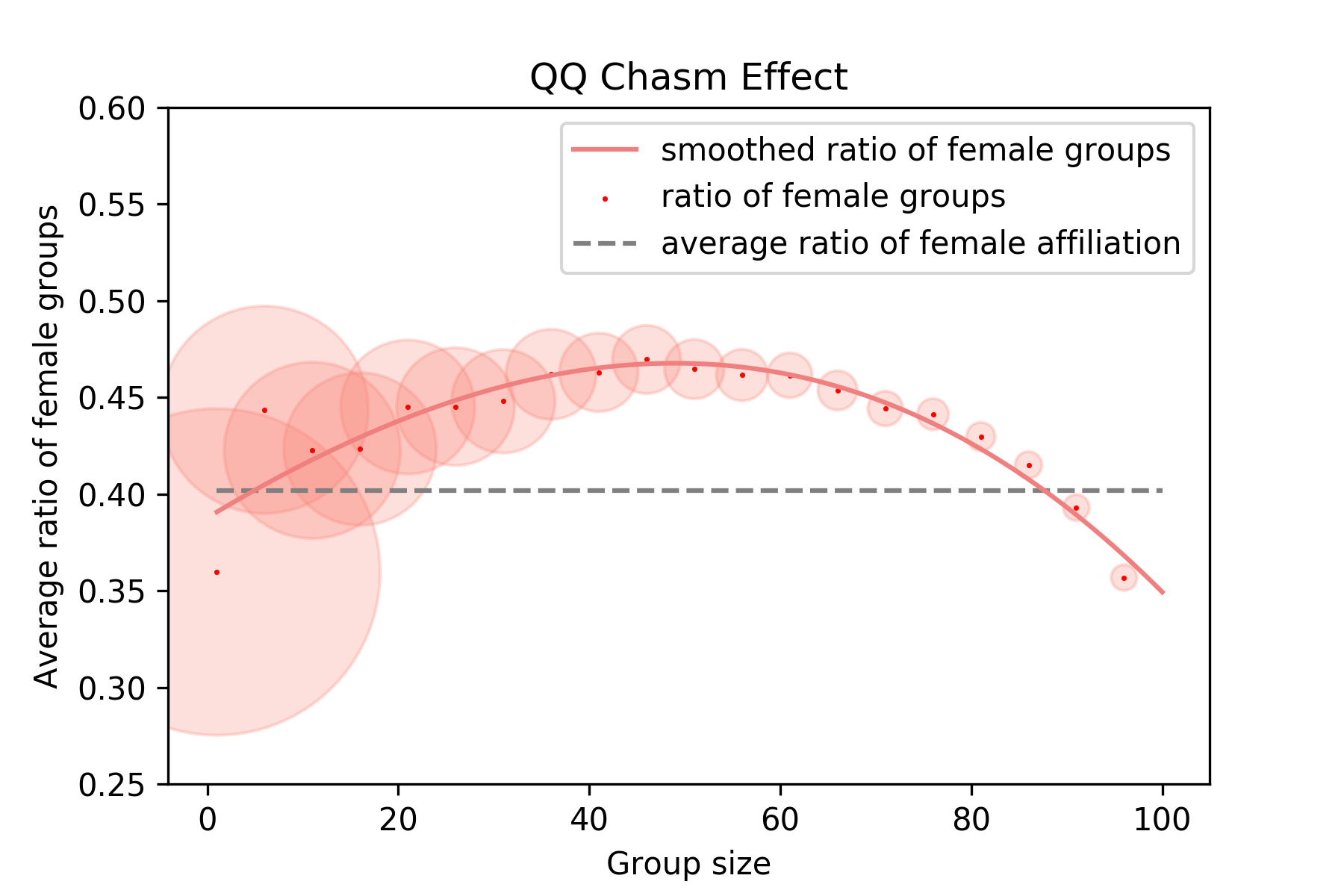}
\includegraphics[scale=0.5]{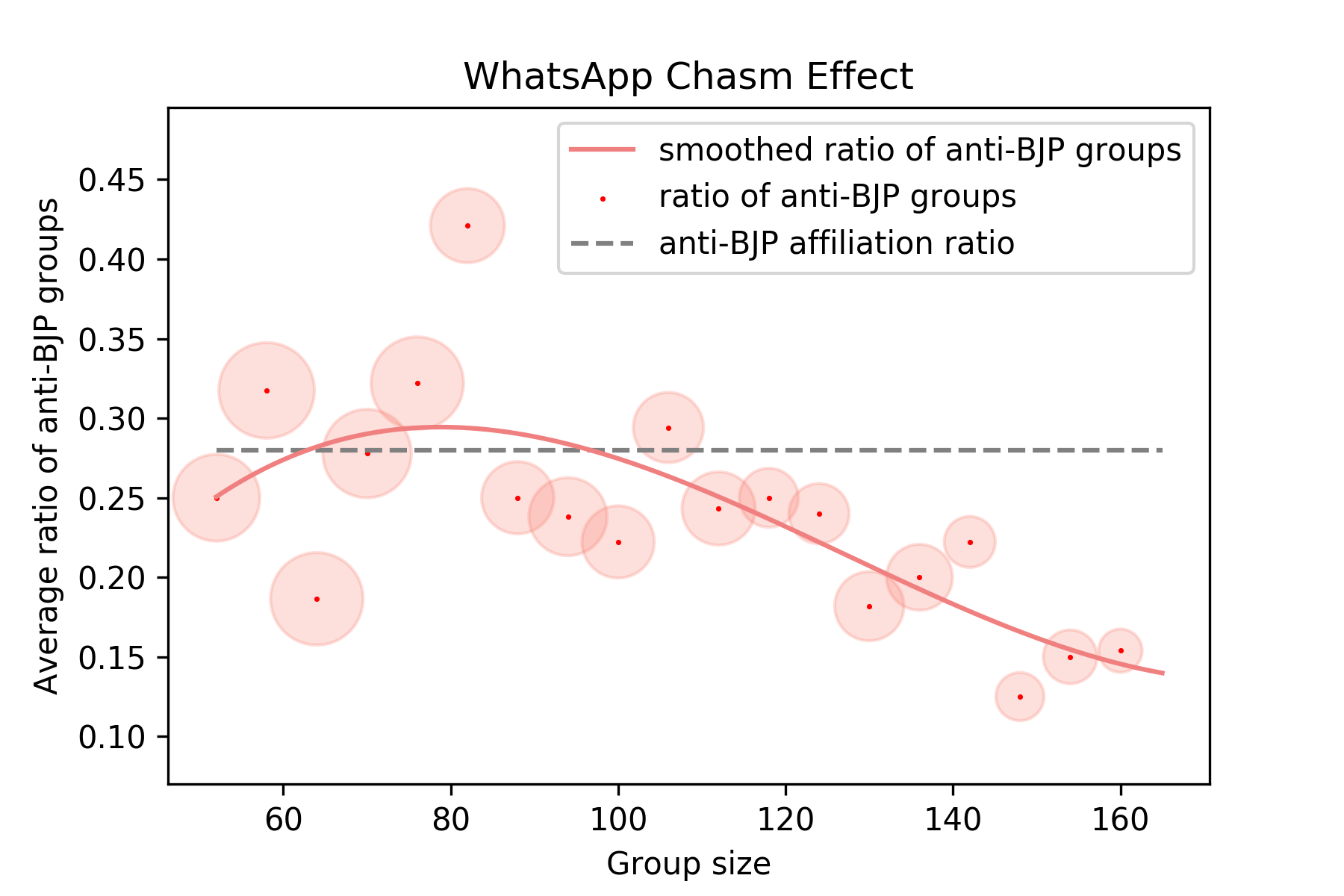}
\caption{Chasm effect on group ratio: we observe that in both datasets, the ratio of minority groups (vs majority groups) is not monotone. As expected from the glass-ceiling effect, the ratio decreases for large group sizes; however, it increases for small groups, which constitutes a larger parts of all groups. In this plot, the radius of light-red circles are proportional to groups counts.}
\label{group_ratio}
\end{figure}

\begin{figure}[H]
\centering
\includegraphics[scale=0.5]{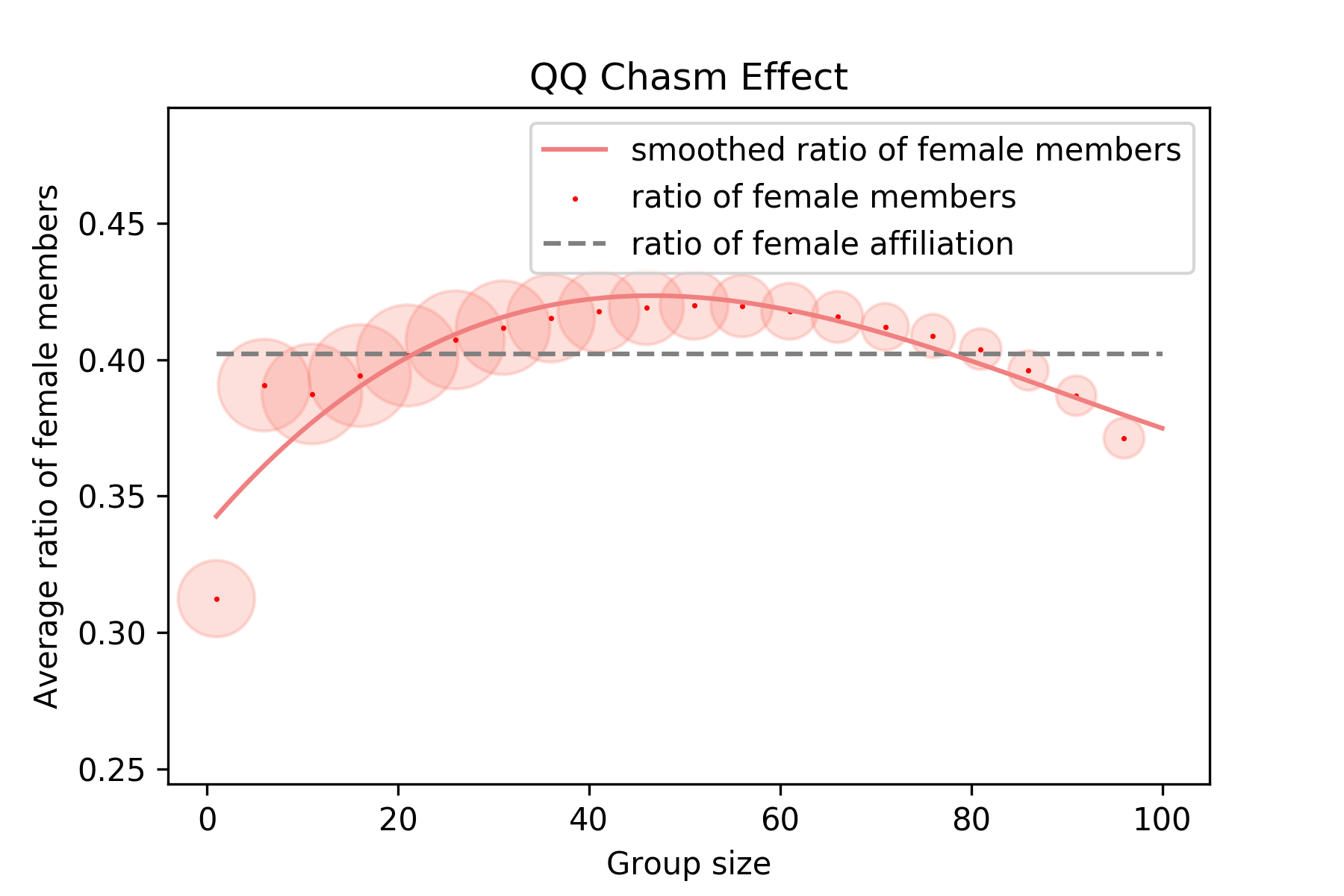}
\includegraphics[scale=0.5]{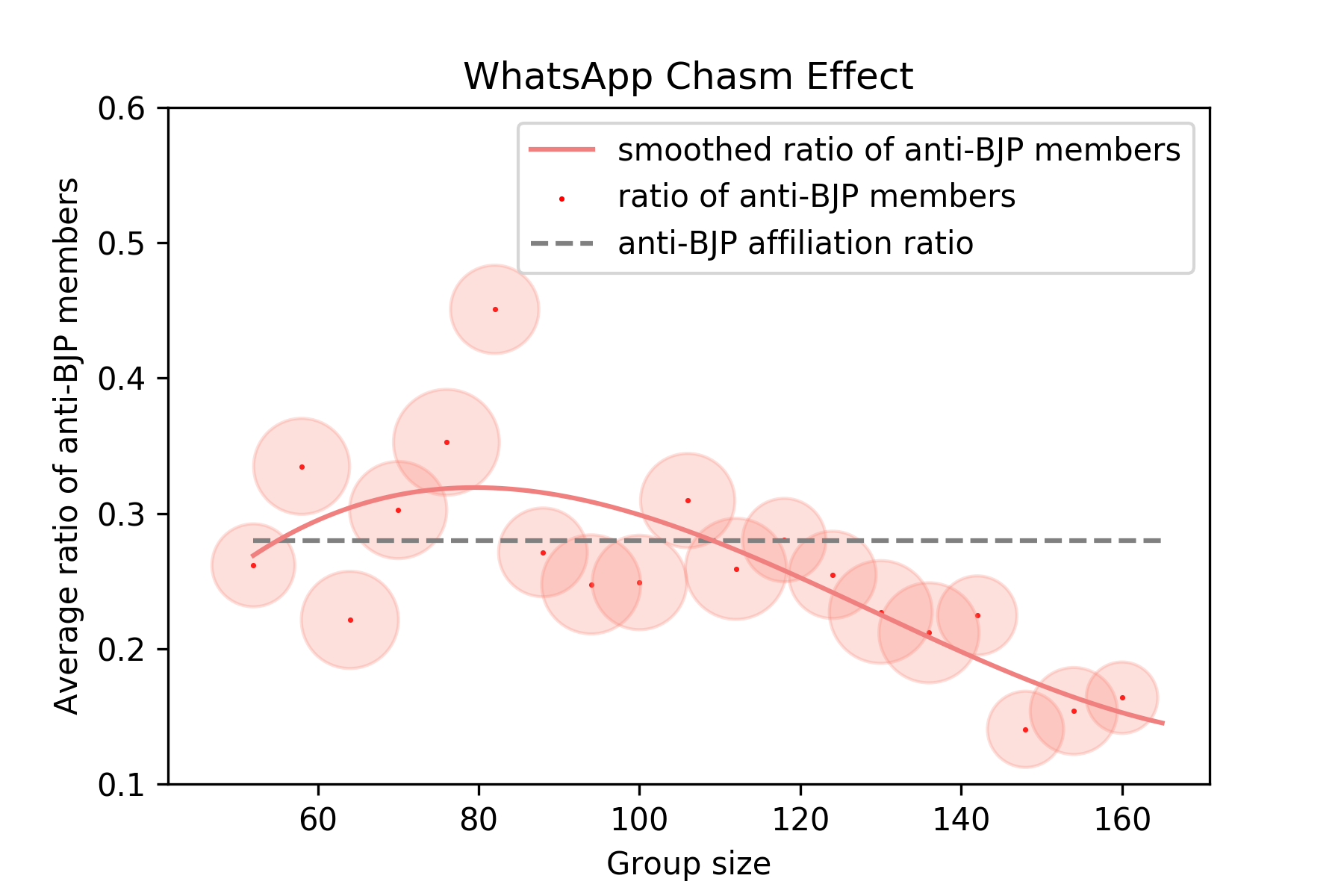}
\caption{Chasm effect on average member ratio: we again observe that the average ratio of minorities within groups of fixed sizes first increases, and then decreases. The radius of light-red circles are proportional to the sum of group sizes.}
\label{member_ratio}
\end{figure}

In our bi-affiliated bipartite networks, we observe this chasm effect for both the group mode and member mode. As shown in Fig. \ref{group_ratio}, we calculate the ratio of minority-dominated groups for each group size bucket and find that the minority group ratio does not monotonically decrease. More specifically, in the QQ dataset, we observe that the ratio of female-dominated groups increases among groups of size 1-55 and decreases afterwards. In the WhatsApp dataset, the ratio of the anti-BJP group increases for groups of size 52-85 and decreases thereafter. In both plots, we see that the very small and very large minority groups are under-represented and the representation improves in medium-sized minority groups. Similarly, in Fig. \ref{member_ratio}, we calculate the average ratio of minority members at each level of group sizes and find a similar non-monotonic trend. In the QQ dataset, the average ratio of female members in a group first increases among groups of size smaller than 55 and decreases afterward; similarly, the average proportion of anti-BJP members in the WhatsApp dataset increases among groups of size less than 82, and decreases thereafter. In both plots, we see that minority members are under-represented in the very small, and the representation gets improved in middle-sized minority groups.

The above observations have not been studied in the existing literature of social networks but they are non-negligible. First, smaller groups constitute a significant portion of all groups in the networks: 40.9\% groups have sizes smaller than 55 in the QQ dataset, and 41.7\% groups have sizes smaller than 82 in the WhatsApp dataset. Furthermore, this observation is not unique to bipartite networks as we find a similar non-monotonic result in unipartite networks (Fig. \ref{member_net}). Due to the space limit, we delay the description of unipartite network datasets, as well as further analysis on unipartite networks to Appendix \ref{unipartite analysis}.

\begin{figure}[H]
\centering
\includegraphics[scale=0.4]{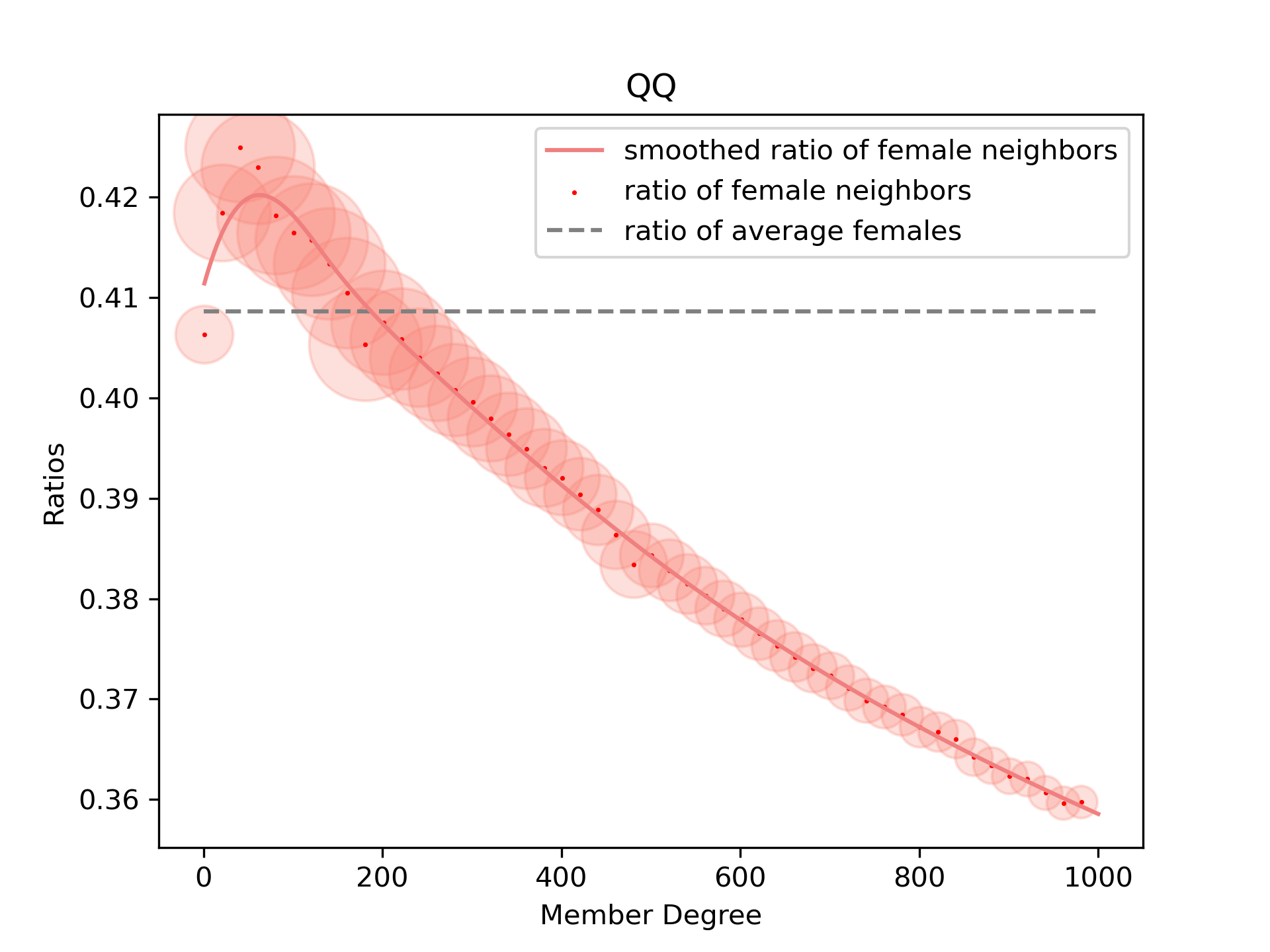}
\includegraphics[scale=0.475]{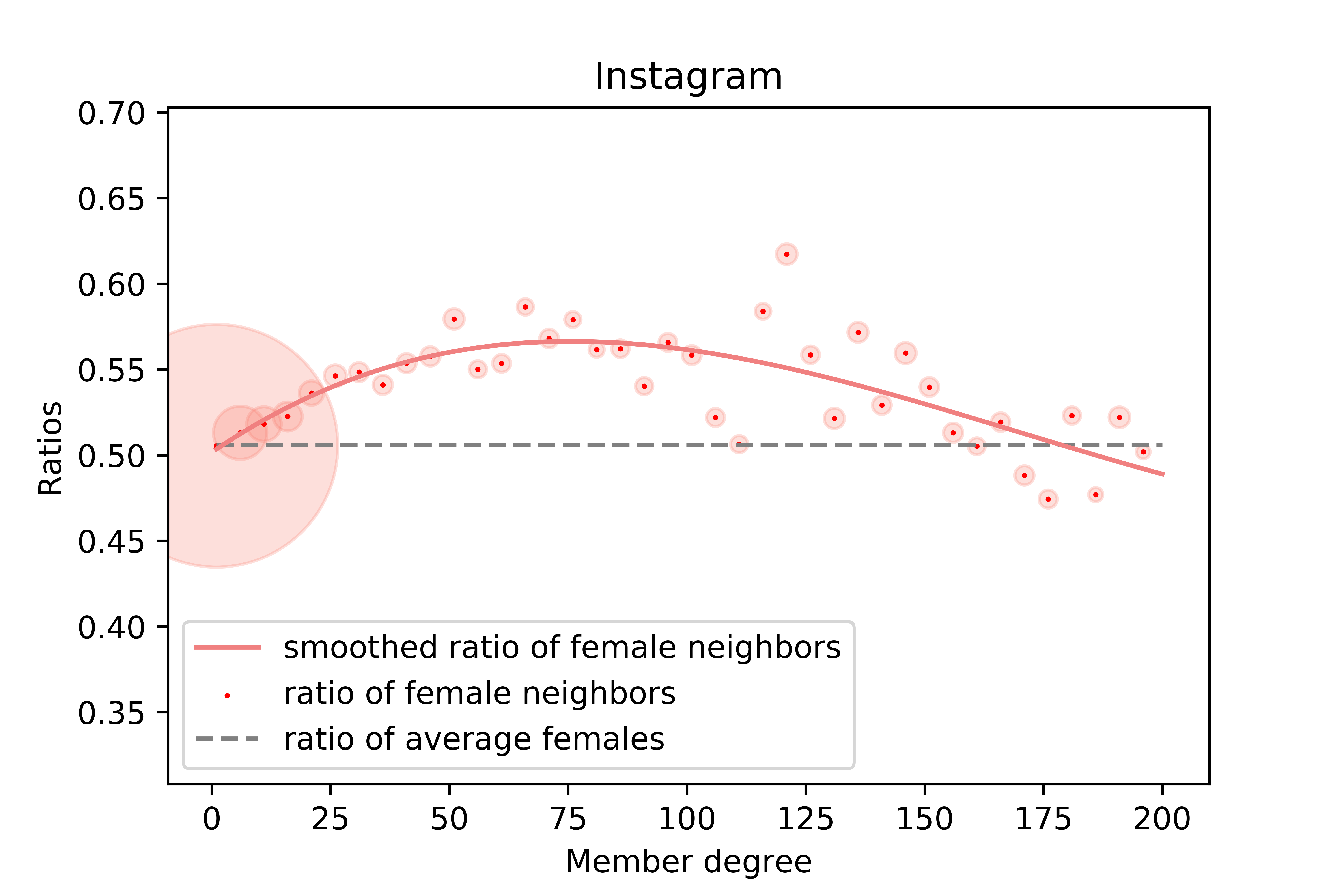}
\caption{Chasm effect on unipartite networks: the chasm effect is not unique to bipartite networks. We observe in the projected QQ membership networks, as well as in the Instagram network, that the ratio of female connections a member has first increase, then decreases. This common pattern shared by networks of different type, as well as networks focusing on different context, indicates that there may be simple structural patterns that are not explained in the existing literature.}
\label{member_net}
\end{figure}

This more complete picture of minority representation in every level of a social hierarchy is especially significant as it can provide insights into minorities at the lower rungs who are far more disadvantaged and vulnerable than those at the higher level. Previous models of network growth with only the three mechanisms discussed in Section \ref{sec:three_mech} are unable to capture or explain this chasm effect (proved in Section 4). This motivates us to propose a new bi-affiliation bipartite network model in the next sections that could reveal the complex interaction among several driving mechanisms of the social disparities. 

\section{Explaining the Chasm using Selective Homophily}
We now examine the roles played by the observed mechanisms, and the way they interact with each other, as well as another well-observed social network mechanism, to create the glass-ceiling and the chasm effect. To better characterize the interactions, we use a simple model to show that the two effects can naturally arise under a specific combination of the network mechanisms. What's more, the mechanisms that constitute this combination are necessary conditions for the two effects to occur at the same time.

\subsection{A model of network growth dynamic}
Formally, we consider a bi-affiliated bipartite network, with one subset of nodes representing members, $M$, and the other groups, $G$. We assume two affiliations in the network and we denote them as red and blue, where the red affiliation represents the minority, and the blue affiliation represents the majority. Every member $m\in M$ belongs to exactly one of the two affiliations. Similarly, every group $g\in G$ belongs to one affiliation. We use $\mathcal{N}(M\cup G, t, \Theta)$ to denote a network generated with a model by $\Theta$ at time step $t$ where $\Theta$ is the set of parameters that is used to generate networks. 

We assume the following well-observed mechanisms:
\begin{enumerate}
\item \emph{rich-get-richer:} current active members are likely to join more groups than current inactive members; large groups are likely to have a higher growth rate than small groups. 
\item \emph{homophily:} members tend to join groups of their own affiliation.
\item \emph{equal-chance:} members may join groups uniformly at random.
\end{enumerate} 

Applying the homophily mechanism to the other two gives rise to three possible homophilous mechanisms. {We test each of them in a \emph{simple homophilous model (SHM)}. Namely, they are \emph{SHM with selective homophily on rich-get-richer}, \emph{SHM with selective homophily on equal-chance}, and \emph{SHM with general homophily}.}

\begin{figure}
\centering
\includegraphics[scale=0.35]{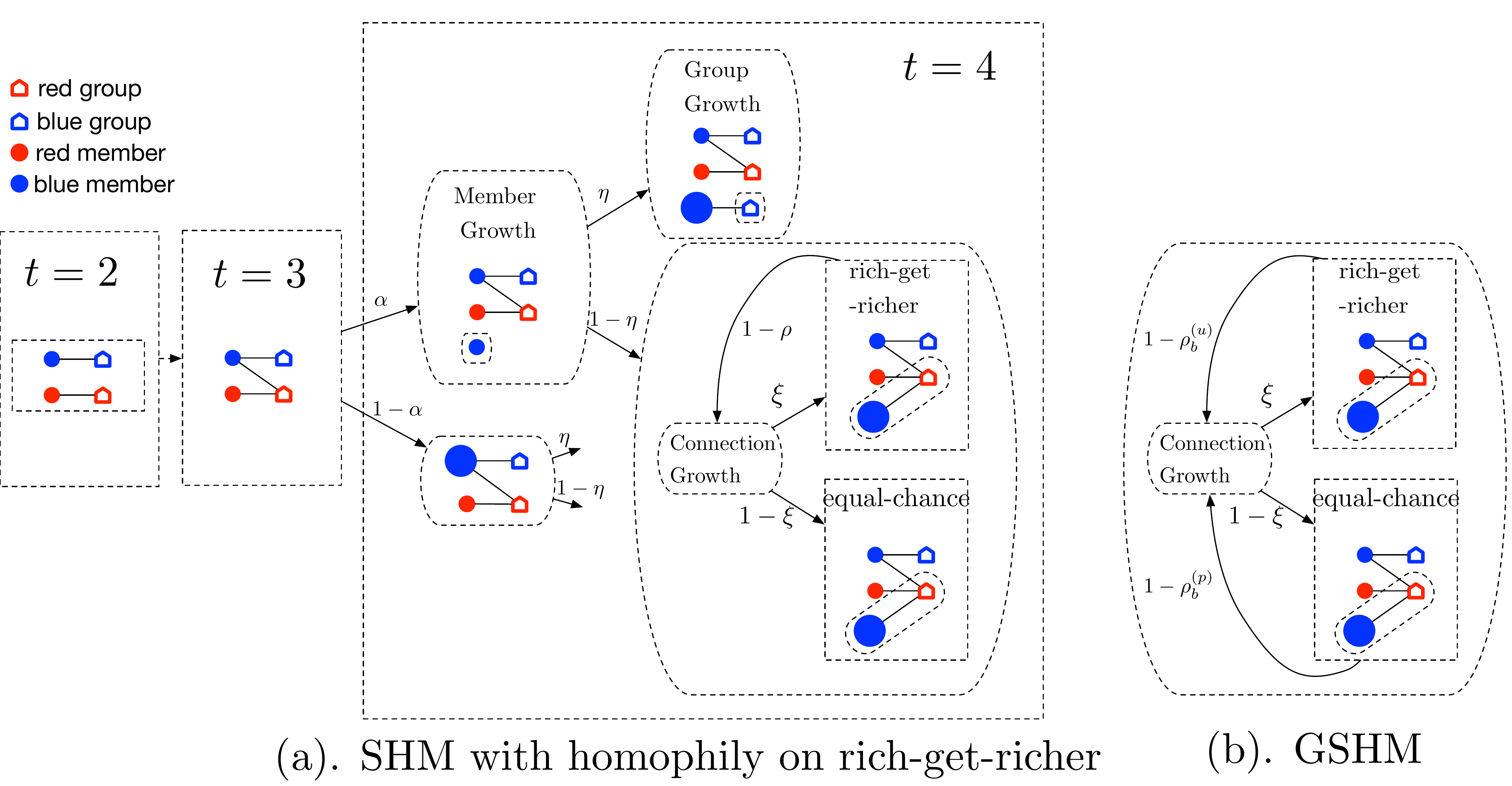}
\caption{SHM (defined in Section 4) and GSHM (defined in Section 5): in SHM, at each time $t$, exactly one connection is built between the set of members and the set of groups. A chosen member can either create a group or join an existing group. If joining an existing group, the member selects a group based on either (1) rich-get-richer mechanism, or (2) equal-chance mechanism. If homophily is applied to the chosen mechanism, then the member may reject the connection and choose a new group until successfully joining a group. The GSHM follows the same step, except that in SHM, the homophily level is the same for different mechanisms and members from different affiliations, while GSHM has differentiated homophily levels.}
\label{model_plots}
\end{figure}

Formally, we have $\Theta = (\alpha, \eta, r, \xi, \rho)$, where $\alpha$ and $\eta$ captures the arrival rate of members and groups, respectively, $r\leq 1/2$ represents the likelihood of a new arrival member being red, $0\leq \xi\leq 1$ captures the level of the rich-get-richer mechanism for groups, and $\rho$ represents the level of homophily in the network.

{We now describe SHM with the three homophilous mechanisms in more details, and demonstrate them in Figure \ref{model_plots} and Figure \ref{corollary_fig}}. At time $t = 2$, we initialize the bipartite network with one red member connecting to a red group, and one blue member connecting to a blue group. At time $t$, the network grows as follows:
\begin{itemize}
\item \textbf{Member Growth}: 
\begin{itemize}
\item (\emph{minority - majority}) with probability $\alpha$ ($0<\alpha <1$), a new member $m^*$ joins the network, and it is colored red with probability $r$ ($0<r\leq 1/2$) and colored blue with probability $1-r$; 
\item (\emph{rich-get-richer}) otherwise, with probability $1-\alpha$, we randomly pick an existing member $m^*$ with a probability proportional to $\text{degree}(m^*)$.
\end{itemize}
\item \textbf{Group Growth}: with probability $\eta$ ($0<\eta <1$), the member creates a group of color $c(m^*)$.
\item \textbf{Connection Growth}: with probability $1-\eta$, the member $m^*$ joins an existing group, according to the following two steps:
    \begin{itemize}
	\item (\emph{rich-get-richer}) with probability $\xi$, $m^*$ picks a group $g^*$ with probability proportional to $\text{degree}(g^*)$.
	\begin{itemize} 
		\item under \emph{selective homophily on rich-get-rich mechanism} or \emph{general mechanism}: if $c(m^*) = c(g^*)$, $m^*$ joins $g^*$ directly; otherwise, $m^*$ accepts the connection with probability $\rho$. If $m^*$ does not accept the connection, $m^*$ restarts from the beginning of the \emph{Connection Growth} until a new connection is built.
		\item under \emph{selective homophily on equal-chance}: $m^*$ joins $g^*$ directly.
	\end{itemize} 
    \item (\emph{equal-chance}) with probability $1-\xi$, $m^*$ uniformly picks a group $g^*$ at random.  
    \begin{itemize} 
		\item under \emph{selective homophily on rich-get-rich mechanism}: $m^*$ joins $g^*$ directly
		\item under \emph{selective homophily on equal-chance} or \emph{general mechanism}: if $c(m^*) = c(g^*)$, $m^*$ joins $g^*$ directly; otherwise, $m^*$ accepts the connection with probability $\rho$. If $m^*$ does not accept the connection, $m^*$ restarts from the beginning of the \emph{Connection Growth} until a new connection is built.
	\end{itemize}   
\end{itemize}
\end{itemize}

\subsection{A sufficient and necessary condition}
We now provide the formal definition of the glass-ceiling effect and the chasm effect. First note that the two subsets of nodes in bipartite networks often represent different entities, and therefore shall be analyzed separately. For the purpose of this paper, we focus our analysis on the group set, and refer to both the tail glass-ceiling effect and the chasm effect as the effects on groups.

{The tail glass-ceiling effect in bipartite networks describes a decreasing fraction of groups of certain affiliation among larger groups, i.e., in the tail of the group-size sequence.} Let $\text{top}_k^{(G)}(R)$ ($\text{top}_k^{(G)}(B)$) be the number of red (blue) groups that have a size at least $k$, as $t$ goes to infinity. 

\begin{definition}(tail glass-ceiling)
	A network sequence $\left\{\mathcal{N}(M\cup G, t, \Theta)\right\}$  exhibits a tail glass-ceiling effect (or glass-ceiling effect for short) against red if there exists an increasing function $k(t)$ such that $\lim_{t\rightarrow \infty} \text{top}^{(G)}_{k(t)}(B) = \infty$, and
	\begin{equation}
		\lim_{t\rightarrow \infty} \frac{\text{top}^{(G)}_{k(t)}(R)}{\text{top}^{(G)}_{k(t)}(B)} = 0.
	\end{equation}
	
\end{definition}

The chasm effect captures the phenomenon that the representation of groups of the minority affiliation first increases and then decreases, as the group size goes up.
\begin{definition}(chasm)
	A network sequence $\left\{\mathcal{N}(M\cup G, t, \Theta)\right\}$ exhibits a chasm effect against red if there exists $K < \infty$ such that as $t$ goes to infinity, the ratio of red groups $r_k^{(G)}$ as a function of $k$ increases for $k<K$ and decreases for $k>K$.
\end{definition}

\begin{figure}
\centering
\includegraphics[scale=0.3]{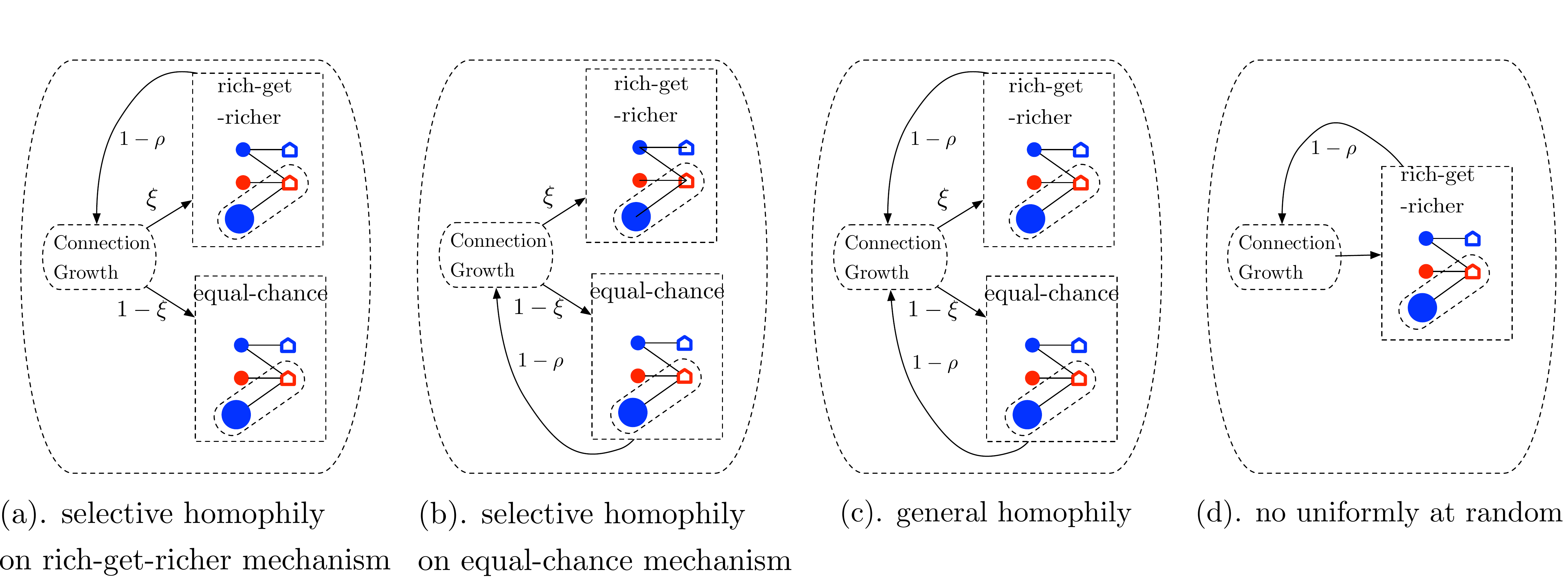}
\caption{Mechanisms in SHM: homophily can apply to either of the rich-get-richer mechanism, or both.}
\label{corollary_fig}
\end{figure}

We first note that the selective homophily on rich-get-richer mechanism can lead to both the tail glass-ceiling effect and the chasm effect. As we will establish all of the following Lemma results later in a generalized model, we defer our proofs to corollaries found in Section 5.
\begin{lemma}\label{lem_select_rich}
	Under some conditions of $\Theta$, a network sequence $\left\{\mathcal{N}(M\cup G, t, \Theta)\right\}$ generated by SHM with the selective homophily on rich-get-richer mechanism exhibits both the tail glass-ceiling effect and the chasm effect as $t$ goes to infinity.
\end{lemma}

\begin{table}
\begin{tabular}{|c|c|c|c|c|}
\hline 
& selective homophily & selective homophily& general homophily & no equal-chance \\ 
 & on rich-get-richer & on equal-chance & &  \\ \hline 
 glass-ceiling & yes [Lemma \ref{lem_select_rich}]&no [Lemma \ref{lem_select_opp}]& yes [citation \cite{avin2015homophily}] & yes [citation \cite{avin2015homophily}]\\ \hline 
chasm &yes [Lemma \ref{lem_select_rich}]& no [no glass-ceiling]&no [Lemma \ref{lem_general}] & no [Lemma \ref{lem_no_homo}]\\ \hline 
\end{tabular}
\caption{Networks generated by SHM cannot have both the glass-ceiling effect and the chasm effect when (1) there is no homophily on rich-get-richer, or (2) there is no equal-chance mechanism.}
\label{table: conditions}
\end{table}
\raggedbottom
Previous works on uni-partite networks imply that the \emph{selective homophily on equal-chance} mechanism cannot lead to the glass-ceiling effects \cite{avin2015homophily}. Indeed, this is also true for bipartite networks.

\begin{lemma}\label{lem_select_opp}
	A network sequence $\left\{\mathcal{N}(M\cup G, t, \Theta)\right\}$ generated by SHM with selective homophily on the equal-chance mechanism do not exhibit tail glass-ceiling effect.
\end{lemma}

Therefore, applying homophily on the rich-get-richer mechanism is a necessary condition for the tail glass-ceiling effect. Then a natural question to ask is: can the general homophily give rise to the tail glass-ceiling effect? The answer is yes, and the reasoning again follows from Theorem Corollary \ref{top_k}. Now, the second question is: can the general homophily give rise to the chasm? The answer here is no.

\begin{lemma}\label{lem_general}
A network sequence $\left\{\mathcal{N}(M\cup G, t, \Theta)\right\}$ generated by SHM with the general mechanism do not exhibit chasm effect.
\end{lemma}

So far, we have shown that having the selective homophily model is necessary for both the tail glass-ceiling effect and the chasm effect. We have also seen that selective homophily on the equal-chance mechanism does not help create the glass-ceiling effect either. Moreover, having the same level of homophily on the rich-get-richer mechanism and the equal-chance mechanism would eliminate the chasm effect. It seems like the equal-chance mechanism is not useful in creating the glass-ceiling and the chasm effects (Figure \ref{corollary_fig}-(d)). However, this is not true. The following corollary shows that although the homophily on equal-chance mechanism is not necessary for either effect to emerge, the equal-chance mechanism itself is needed to have the chasm effect. 
\begin{lemma}\label{lem_no_homo}
A sequence of networks $\mathcal{N}(M\cup G, t, \Theta)$ generated by SHM without the equal-chance mechanism do not exhibit chasm effect.
\end{lemma}
Therefore, the equal-chance mechanism is also a necessary mechanism in creating both effects. We conclude the above findings in Table \ref{table: conditions}.

\begin{thm}
The selective homophily on rich-get-richer mechanism and the equal-chance mechanism are both necessary mechanisms for networks generated through the SHM to exhibit both the tail glass-ceiling effect and the chasm effect.
\end{thm}

Intuitively, the equal-chance mechanism gives small groups chances to grow, and having homophily on rich-get-richer mechanism allows majorities to grow large groups. Under the selective homophily on rich-get-richer mechanism, because there are more majority groups, minorities are less likely to join groups through the rich-get-richer mechanism. Instead, they grow smaller groups. In a long run, there are more minority groups with middle sizes; when there is no equal-chance mechanism, small groups do not have the chance to grow, and therefore the network does not have the chasm effect; under the selective homophily on the equal-chance mechanism, because there is no homophily on rich-get-richer mechanism, majorities do not have the chance to grow large groups, and therefore, there is no glass-ceiling effect; under the general homophily mechanism, small blue groups grow no less than small red groups, and thus do not exhibit the chasm effect. If we allow different homophily levels for the majority and the minority, it is possible for small red groups to grow faster than blue groups. We will see more in the next section. The interaction among the mechanisms in the real world is undoubtedly more complex, but we hope the above intuition could offer a more profound understanding of the driving mechanisms of social disparities.

\section{Hegemony in General Homophilous Networks}
We now extend the SHM with selective homophily on rich-get-richer mechanism to a new model that serves two purposes: first, it can still capture both the glass-ceiling effect and the chasm effect; second, it allows more degrees of freedom, and therefore can reproduce real social networks. In this section, we introduce a generalized model, prove the sufficient and necessary conditions for the two effects to happen, and reproduce real datasets with the generalized model.

For clarity, we list all notations that are used in our theory presentation in Table \ref{notations}.

\subsection{Generalized homophily model}
The previous analysis on SHM implies that the level of homophily plays an important role in large blue groups and small red groups' faster growth rate than the other affiliation. We therefore introduce a new \emph{generalized selective homophily model (GSHM)} with two sets of new parameters: $\rho_r^{(u)}$ ($\rho_b^{(u)}$) captures the level of \emph{red (blue) selective homophily on equal-chance mechanism}; $\rho_r^{(p)}$ ($\rho_b^{(p)}$) captures the level of red \emph{(blue) selective homophily on rich-get-richer mechanism}.

We now present the generalized model in details. At time $t = 2$, we initialize the bipartite network with one red member connecting to a red group, and one blue member connecting to a blue group. At time $t$, the network grows as the following:
\begin{itemize}
\item \textbf{Member Growth}: 
\begin{itemize}
\item (\emph{minority-majority}) with probability $\alpha$ ($0<\alpha <1$), a new member $m^*$ joins the network, and it is colored red with probability $r$ ($0<r\leq 1/2$); 
\item (\emph{rich-get-richer}) otherwise, with probability $1-\alpha$, we randomly pick an existing member $m^*$ with probability proportional to $\text{deg}(m^*)$.
\end{itemize}
\item \textbf{Group Growth}: with probability $\eta$ ($0<\eta <1$), the member creates a group of color $c(m^*)$.
\item \textbf{Connection Growth}: with probability $1-\eta$, the member $m^*$ joins an existing group, according to the following two steps:
    \begin{itemize}
	\item (\emph{rich-get-richer}) with probability $\xi$, $m^*$ picks a group $g^*$ with probability proportional to $\text{deg}(g^*)$. If $c(m^*) = c(g^*)$, $m^*$ joins $g^*$ directly; otherwise, $m^*$ accepts the connection with probability $\rhop_{c(m^*)}$. If $m^*$ does not accept the connection, $m^*$ restarts from the beginning of the \emph{Connection Growth} until a new connection is built.    
    \item (\emph{equal-chance}) with probability $1-\xi$, $m^*$ uniformly picks a group $g^*$ at random. If $c(m^*) = c(g^*)$, $m^*$ joins $g^*$ directly; otherwise, $m^*$ accepts the connection with probability $\rhou_{c(m^*)}$. If $m^*$ does not accept the connection, $m^*$ restarts from the beginning of the \emph{Connection Growth} until a new connection is built. 
\end{itemize}
\end{itemize}

Under GSHM, when a user decides on whether to join a selected group, the probability of accepting depends on both the user's affiliation and the mechanism that the user uses to pick the group. We illustrate this probability specification in Figure \ref{model_plots} - (b). Note that all of the three homophilous mechanisms are special cases of the GSHM.

\begin{table}
    \begin{tabularx}{\textwidth}{p{0.22\textwidth}X}
    \toprule
      { General notations: } \\
 	$c(x)$ & color of node $x \in M\cup G$.\\
 	$deg(x)$ & degree of node $x \in M\cup G$.\\
 
      { Group notations: } \\
 	$G_t(C)$ & number of groups in color $C$ at time $t$.\\
 	$G_{k,t}(C)$ & number of groups in color $C$ with size $k$ at time $t$; $G_k(C):=\lim_{t\rightarrow \infty} \frac{\mathbb{E}\left(G_{k,t}(C)\right)}{t}$.\\
      $r^{(G)}_t(C)$ & group growth rate of color $C$ at time $t$; that is, $r^{(G)}_t(C):= \frac{G_t(C)}{t}$.\\
       { Member notations: } \\  
    $M_t(C)$ &number of $C$ members at time $t$.\\
    $M_{k,t}(C)$ &number of members in color $C$ with degree $k$ at time $t$; $M_{k}(C):=\lim_{t\rightarrow \infty}\frac{\mathbb{E}\left(M_{k,t}(C)\right)}{t}$.\\
    $M^{(k)}_t(C)$ & number of members in color $C$ that are contained in groups of size $k$.\\
    $r^{(M,G)}_{k,t}(C)$ & ratio of expected number of members in color $C$ that are contained in groups of size $k$ at time $t$.\\
    $r^{(M,G)}_{k,t}(C_1, C_2)$ & ratio of expected members of color $C_1$ in groups of size $k$ with color $C_2$; $r^{(M,G)}_{k}(C_1, C_2) = \lim_{t\rightarrow \infty} r^{(M,G)}_{k,t}(C_1, C_2)$.\\
    	{ Edges notations: } \\ 
    	$E^{(G)}_t(C)$ & sum of group sizes in color $C$ at time $t$; $r^{(E, G)}_t(R): = \frac{E^{(G)}(R)}{t}$.\\
    	$E^{(M)}_t(C)$ & sum of member degrees in color $C$ at time $t$; $r^{(E, M)}_t(R): = \frac{E^{(M)}(R)}{t}$.\\
      \bottomrule
     \end{tabularx}
     \caption{Notation}\label{notations}
    \end{table}
\raggedbottom
We now mathematically characterize the degree distributions of the two types of nodes in GSHM and provide the sufficient and necessary conditions for the glass-ceiling effect and the chasm effect to happen.
\subsection{Proof of convergence to limit degree distributions}
\subsubsection{Group-size distributions}
We first investigate the size distributions of the red and blue groups in a bipartite network generated by the GSHM, and show that the number of red groups of size $k$, $G_k(R)$ and the number of blue groups of size $k$, $G_k(B)$ follow power laws under the GSHM model.

\begin{thm}\label{group size power-law}  (proof in appendix \ref{proof of thm 5.1}) Let $\{N(M\cup G, t, \Theta)\}$ be a sequence of networks produced by the GSHM model. Assume that $\rhop_R, \rhop_B > 0$. The red group-size distribution $G_k(R)$ and the blue group-size distribution $G_k(B)$ asymptotically follow the power law distributions; specifically, as $t$ goes to infinity,
\begin{align}
G_k(R)\propto k^{-\beta(R)}, \,\,\,\,\,\,
G_k(B)\propto k^{-\beta(B)},
\end{align}
with $ \beta(R) = 1 + \frac{1}{C_{R,1}}$ and $\beta(B) = 1 + \frac{1}{C_{B,1}}$, where
\begin{align}
    C_{R,1} & := \frac{r (1-\eta) \xi }
    			{1-(1-\rhop_R)\xi (1-\alpha^*)-(1-\rhou_R)(1-\xi) (1-r)}
    			+ \frac{(1-r) (1-\eta)\rhop_B\xi }
    				{1-(1-\rhop_B)\xi \alpha^*-(1-\rhou_B)(1-\xi)r},
\end{align}
\begin{align}
	C_{B,1} & := \frac{(1-r) (1-\eta) \xi }
    			{1-(1-\rhop_B)\xi \alpha^* - (1-\rhou_B)(1-\xi) r}
    			+ \frac{r (1-\eta)\rhop_R\xi }
    				{1-(1-\rhop_R)\xi (1-\alpha^*) - (1-\rhou_R)(1-\xi) (1-r)},
\end{align}
{where $\alpha^*$ denotes the limit of the sum of red group sizes over the sum of all group sizes as $t$ goes to infinity}, and it is the unique solution in $(0,1)$ satisfying
    	\begin{equation}
    		\label{eq_alpha_star}
    	\alpha^* = r \eta + \frac{r (1-\eta) (\xi \alpha^* + (1-\xi)r) }
    			{1-(1-\rhop_R)\xi (1-\alpha^*)-(1-\rhou_R)(1-\xi) (1-r)}
    			+ \frac{(1-r) (1-\eta) (\rhop_B \xi \alpha^* + \rhou_B (1-\xi)r)}
    				{1-(1-\rhop_B)\xi \alpha^*-(1-\rhou_B)(1-\xi)r }.
    	\end{equation}
\end{thm}

\subsubsection{Member degree distribution}
We can use similar strategies to show that the member degrees also follow power-laws with the same power.
\begin{thm}\label{member degree power-law}(proof in appendix \ref{proof of thm 5.2}) Let $\{N(M\cup G, t, \Theta)\}$ be a sequence of networks produced by GSHM. The red member-degree distribution $M_k(R)$ and the blue member-degree distribution $M_k(B)$ asymptotically follow the power law distributions with the same power; specifically, as $t$ goes to infinity,
\begin{align}
M_k(R)\propto k^{-\left(1+\frac{1}{1-\alpha}\right)}, \,\,\,\,\,\,
M_k(B)\propto k^{-\left(1+\frac{1}{1-\alpha}\right)}.
\end{align}
\end{thm}
\subsection{Conditions for glass ceiling and chasm in hegemony}
\subsubsection{Tail glass-ceiling}

The existence of tail glass-ceiling follows directly from Theorem \ref{group size power-law}. 
\begin{cor}\label{top_k}
Let $\{N(M\cup G, t, \Theta)\}$ be a sequence of networks produced by GSHM. Let $\beta(R),\beta(B)$ be as defined in Theorem \ref{group size power-law}. Then
\begin{itemize}
\item when $\beta(R) < \beta(B)$, $\{N(M\cup G, t, \Theta\}$ exhibits tail glass-ceiling effect against the blue groups.
\item when $\beta(R) > \beta(B)$, $\{N(M\cup G, t, \Theta\}$ exhibits tail glass-ceiling effect against the red groups.
\item when $\beta(R) = \beta(B)$ or $\xi =0$, $\{N(M\cup G, t, \Theta\}$ the network does not exhibit tail glass-ceiling effect.
\end{itemize}
\end{cor}
\begin{proof}
Assume $\beta(R) < \beta(B)$. Let $k(n):= n^{\frac{1}{\beta(B)}}$. Then
\begin{equation}
\mathbb{E}[\text{top}_k^G(B)]= n \sum_{k'\geq k} G_{k'}(B) = O\left(n\cdot n^{-\frac{\beta(B)}{\beta(B)}} \right) = O(1);
\end{equation}
\begin{equation}
\textrm{and, we have for an $\epsilon >0$,}\;\; 
\mathbb{E}[\text{top}_k^G(R)]= n \sum_{k'\geq k} G_{k'}(R) = \Omega \left(n\cdot n^{-\frac{\beta(R)}{\beta(B)}} \right) = \Omega\left(n^{1-\frac{\beta(R)}{\beta(B)}}\right) = \Omega(n^\epsilon).
\end{equation}
\end{proof}

\begin{cor}\label{selective homophily on opportunity}
	A network sequence $\left\{\mathcal{N}(M\cup G, t, \Theta)\right\}$ generated by SHM with selective homophily on equal-chance leads no tail glass-ceiling effect for groups.
\end{cor}	
\begin{proof}
	SHM with selective homophily on equal-chance implies $\rhou_R = \rhou_B$ and $\rhop_R = \rhop_B = 1$, which yields $C_{R,1} =  C_{B,1}$.
\end{proof}

\subsubsection{Chasm}
We are now ready to prove the first result on monotonicity of minority ratio change in homophilous networks, from a \emph{novel} analysis of the distribution. Suppose a network produced by the GSHM has tail glass-ceiling effect against red groups, the following theorem provides the necessary and sufficient condition for the chasm effect to happen.
\begin{thm}
    \label{lemma_bump} (proof in Appendix \ref{proof group bump})
	Following the same notations as in Theorem \ref{group size power-law}. Assume $C_{R,1} < C_{B,1}$. then the group ratio sequence $\{G_k(R)/G_k(B), k \geq 1\}$ has the chasm effect against red, if and only if $k^* > 2$, where
	\begin{equation}
		k^* := \frac{(1+C_{R,1})(1+C_{B,2}) - (1+C_{R,2})(1+C_{B,1})}{C_{R,1} - C_{B,1}},
	\end{equation}
	where
	\begin{align}
    C_{R,2} & := \frac{r (1-\eta) (1-\xi) \frac{1}{\eta} }
    			{1-(1-\rhop_R)\xi (1-\alpha^*)-(1-\rhou_R)(1-\xi) (1-r)}
    			+ \frac{(1-r) (1-\eta)\rhou_B(1-\xi) \frac{1}{\eta} }
    				{1-(1-\rhop_B)\xi \alpha^*-(1-\rhou_B)(1-\xi)r};\\
    C_{B,2} & := \frac{(1-r) (1-\eta) (1-\xi) \frac{1}{\eta}  }
    			{1-(1-\rhop_B)\xi \alpha^* - (1-\rhou_B)(1-\xi) r}
    			+ \frac{r (1-\eta)\rhou_R (1-\xi) \frac{1}{\eta}  }
    				{1-(1-\rhop_R)\xi (1-\alpha^*) - (1-\rhou_R)(1-\xi) (1-r)}.
\end{align}
	Moreover, when $k^* > 2$, the monotonicity of $\{G_k(R)/G_k(B), k \geq 1\}$ changes at $[k^*]$, which is the largest integer smaller than $k^*$.
\end{thm}

\begin{cor}
	Following the notation defined in Theorem \ref{group size power-law} and Theorem \ref{lemma_bump}. A network sequence $\left\{\mathcal{N}(M\cup G, t, \Theta)\right\}$ produced by GSHM has a group chasm effect against the red groups if and only if $C_{R,1} < C_{B,1}$ and $k^* >2$.
\end{cor}

\begin{cor}\label{general selective homophily}
	A network sequence $\left\{\mathcal{N}(M\cup G, t, \Theta)\right\}$ generated by SHM with the general homophily mechanism leads to no chasm effect.
\end{cor}
\begin{proof}
	The general selective homophily is equivalent to setting $\rhou_r = \rhop_r=\rhou_b = \rhop_b$ in the GSHM. It is easy to see that, for some positive constant $\gamma > 0$, we have that $C_{R,2} = \gamma C_{R,1}$, $C_{B,2} = \gamma C_{B,1}$. Substituting this relation into the expression for $k^*$, we have that
	\begin{equation}
		k^* = \frac{(1+C_{R,1})(1+\gamma C_{B,1}) - (1+ C_{B,1})(1+\gamma C_{R,1})}{C_{R,1} - C_{B,1}}
			= 1-\gamma < 1,
	\end{equation}
\end{proof}
\begin{cor}
	A network sequence $\left\{\mathcal{N}(M\cup G, t, \Theta)\right\}$ generated by SHM with no equal-chance mechanism in the model leads to no chasm effect.
\end{cor}
\begin{proof}
	Removing the oppotunity mechanism from SHM is equivalent to setting $\xi=1$ in the GSHM. It is easy to check that $C_{R,2} = C_{B,2} = 0$, and thus $k^* = 1$.
\end{proof}

\subsubsection{Non-monotonicity of member-ratios}
So far, our analysis on bipartite networks focuses mainly on groups. We have observed in Section 3.3 that the average member ratio in groups with a fixed size is also non-monotone. The following lemma calculates the average red member ratio among groups of size 1, and that among groups of size going to infinity. When both values are below $r$, we can say that the member ratio is non-monotone.

\begin{lemma}\label{member_ratio_nonmonotone} (proof in appendix \ref{proof member_ratio_nonmonotone})
For the red member ratios within groups with size 1, and within groups with size goes to infinity, we have:
\begin{itemize}
	\item For groups with size 1,
	\begin{align}
		\label{rm_ratio_1}
		\lim_{t \rightarrow \infty} r^{(M, G)}_{1, t} (R)&= \frac{G_1(R)}{G_1(R) + G_1(B)}
		= \frac{1+C_{B,1} + C_{B,2}}{2+C_{R,1} + C_{R,2}+C_{B,1} + C_{B,2}}.\\
	\end{align}
	\item For groups with size goes to infinity, assume $C_{R,1} < C_{B,1}$,
	\begin{equation}
		\label{rm_ratio_infty}
		\lim_{k \rightarrow \infty} \lim_{t \rightarrow \infty} r^{(M, G)}_{k, t} (R)
		= r^{(M, G)}(R),
	\end{equation}
	where $r^{(M, G)}$ is defined as 
	\begin{equation}\label{rm_fit1}
		r^{(M, G)}(R) = \frac{q_{RB}}{q_{RB} + q_{BB}}, 
	\end{equation}
	with
	\begin{align}
			q_{RB} &= r\rhop_R \left(1-(1-\rhop_B)\xi \alpha^* - (1-\rhou_B)(1-\xi) r\right),\\
			q_{BB} &= (1-r)\left(1-(1-\rhop_R)\xi(1-\alpha^*) - (1-\rhou_R)(1-\xi) (1-r)\right).
	\end{align}
\end{itemize}
\end{lemma}

	

\subsection{Fitting the model on real data}
In the previous sections, we have noticed that all the real-data observations we present in Section 3.2 and 3.3 may be present in networks generated by GSHM. In this section, we illustrate its performance in terms of its capability of reproducing the chasm effect and the glass-ceiling effects from real social networks.

{
To do so, we first need to infer parameters from the real dataset. {The minority ratio $r$, the member growth rate $\alpha$, and the group-member ratio $\gamma$ can be directly calculated from the dataset.} We then can optimize over all $\xi, \rhop_R, \rhop_B, \rhou_R, \rhou_B$ to find a set of parameters that gives the that minimizes $\vert\sum_{k=1}^{K}\frac{\hat{G_k}(R)}{\hat{G_k}(B)} - \sum_{k=1}^{K}\frac{{G_k}(R)}{{G_k}(B)}\vert$, where $K$ is the maximam group size, $\hat{G_k}(R)$ and $\hat{G_k}(B)$ are obtained through (\ref{power_law_fit1}) and (\ref{power_law_fit2}), and ${G_k}(R)$ and ${G_k}(B)$ are empirical values observed from the dataset. With the set of $r, \xi, \rho_r^{(p)}, \rho_b^{(p)}, \rho_r^{(u)}, \rho_b^{(u)}$, we can use (\ref{member ratio cor}) to obtain the numerical values for the average ratio of minority members.}

\begin{figure}
\centering
\includegraphics[scale=0.45]{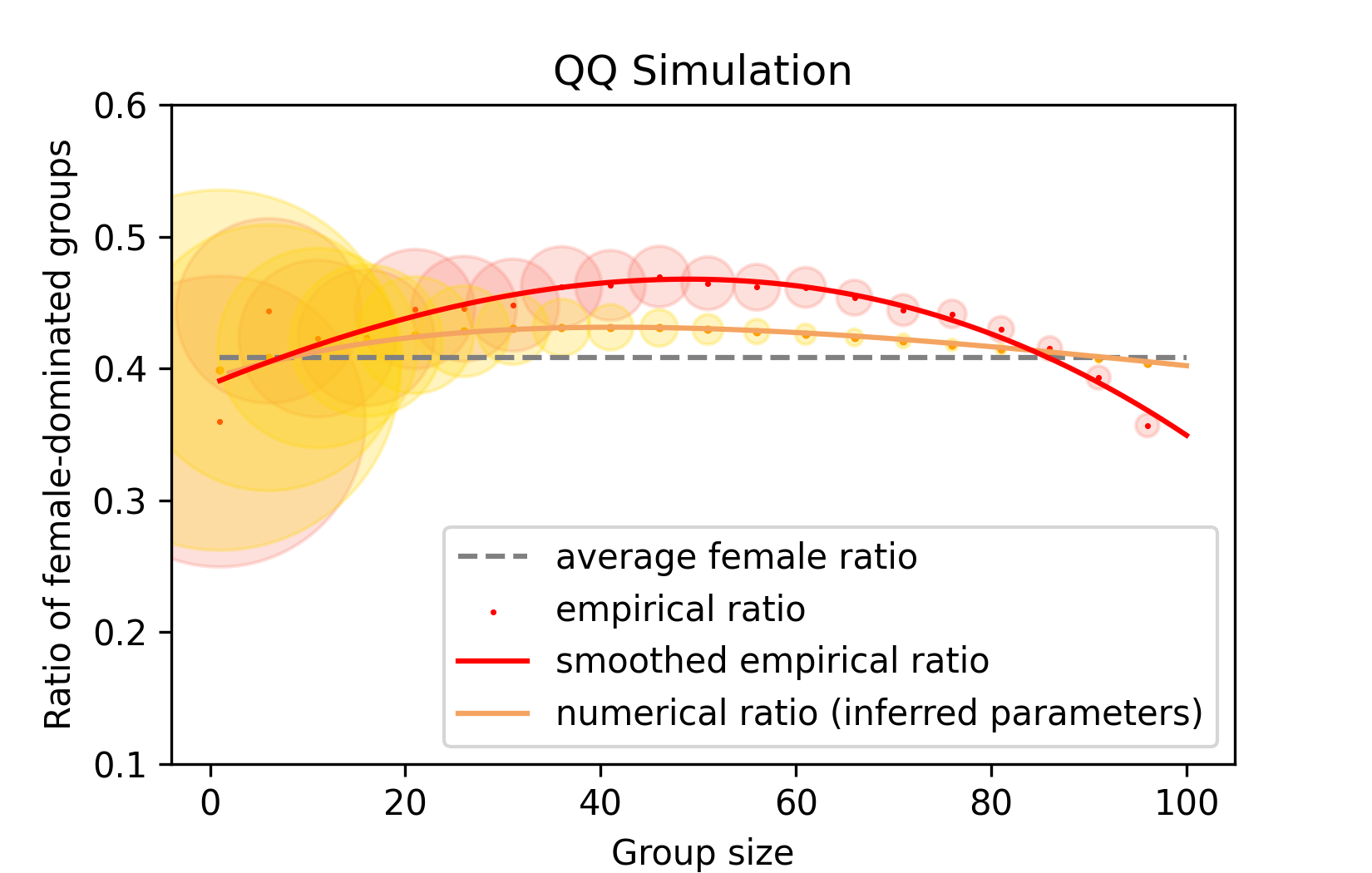}
\includegraphics[scale=0.4]{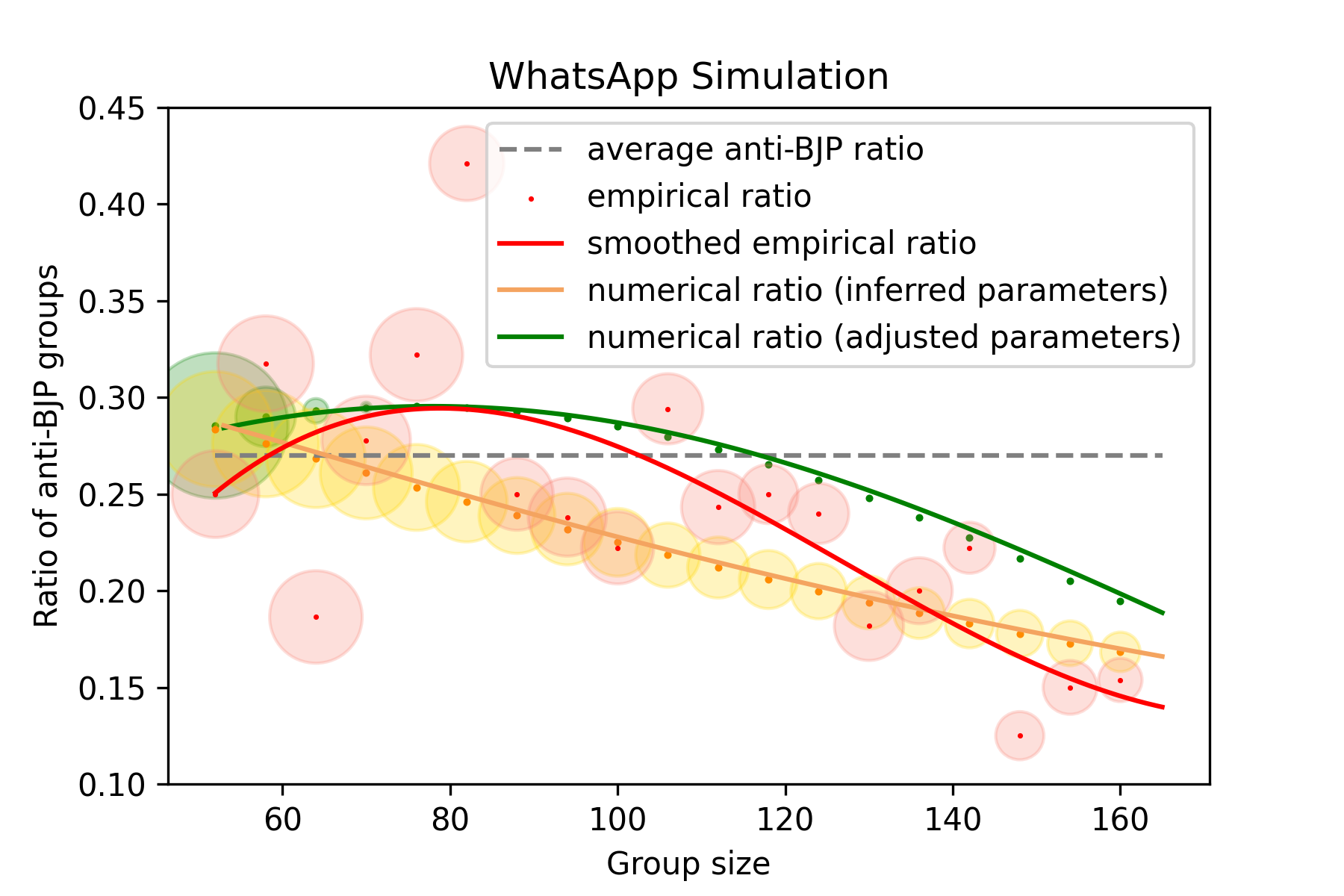}\\
\includegraphics[scale=0.45]{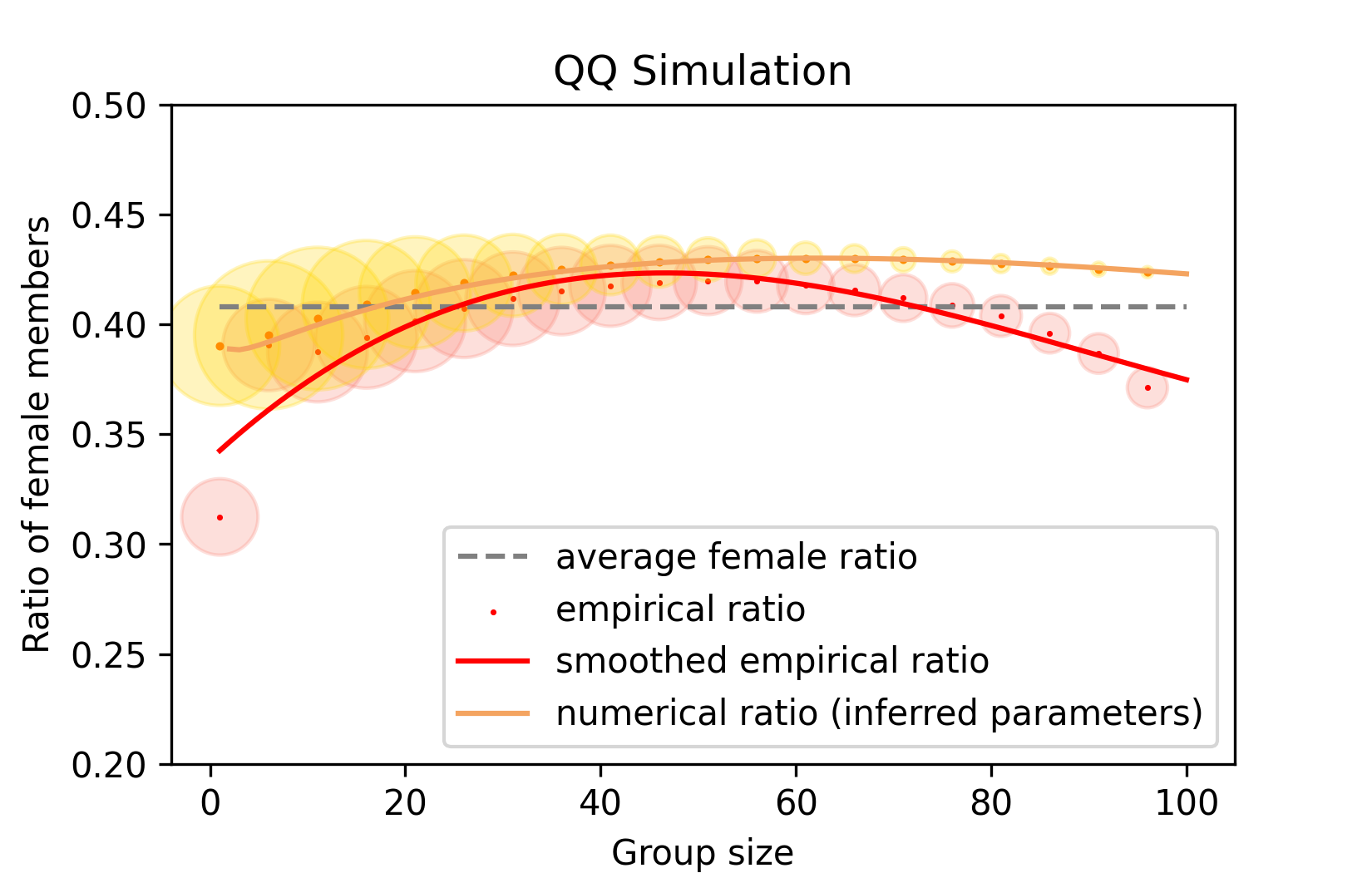}
\includegraphics[scale=0.4]{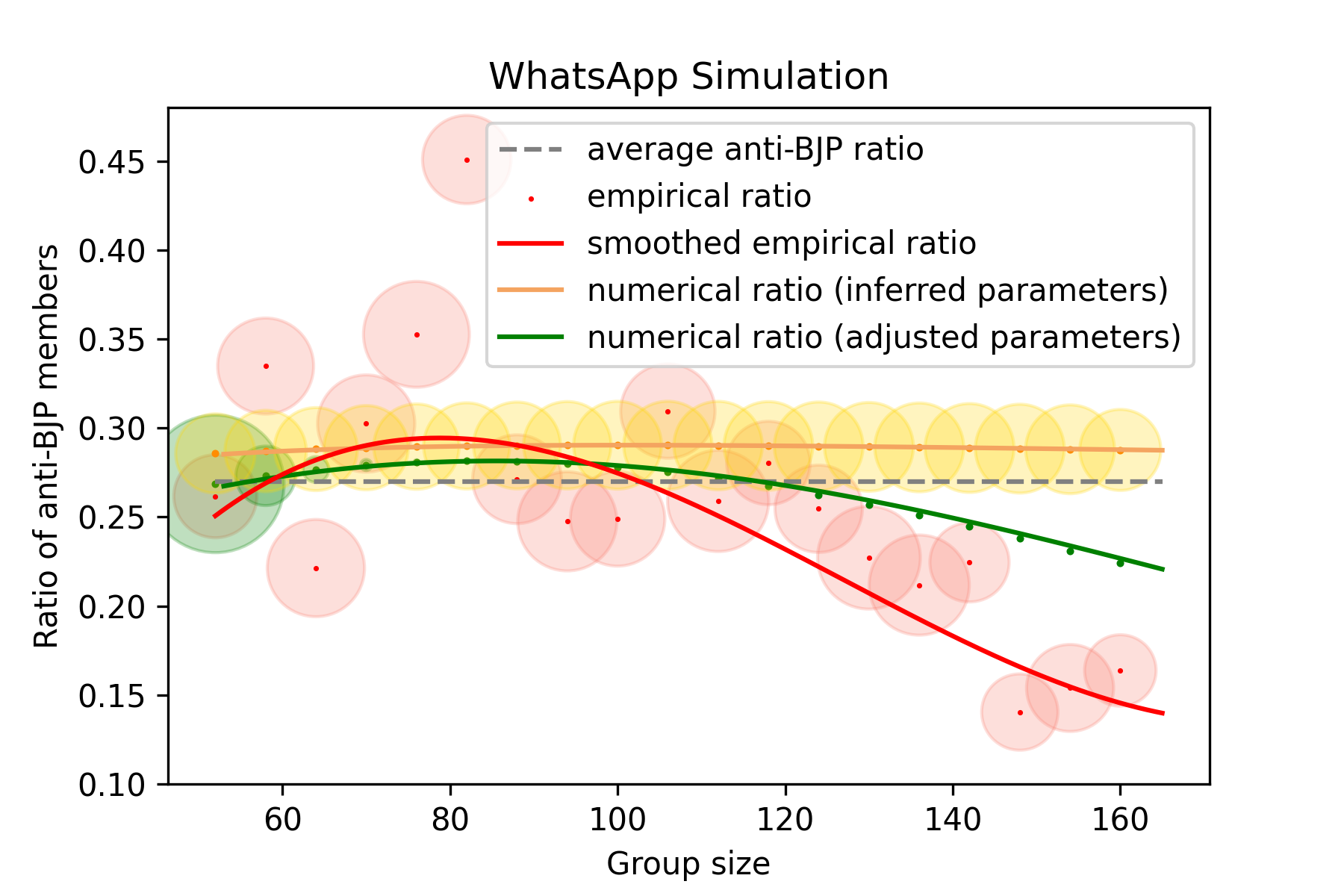}
\caption{Model fits: in the first row, we present the ratios for minority groups. The radii of the red and yellow circles represent the group degree distributions in the empirical datasets and in the numerical simulations respectively. In the second row, we present the ratios for minority members. The radii of the red and yellow circles represent the group member distributions in the empirical datasets and in the numerical simulations respectively. Our numerical simulations with inferred parameters captures the decay of the two distributions in both datasets.}
\label{model_fits}
\end{figure}

In the QQ dataset, we see that for the female-dominated group ratio, our simulation demonstrates both the glass-ceiling effect and the chasm effect. Moreover, our simulation locates the group size where the monotonicity of the ratio changes. This re-confirms our calculation in Theorem \ref{lemma_bump}. For the average female member ratio, we see that it again exhibits both the glass-ceiling effect and the chasm effect. However, it does not locate where the monotonicity changes. We are not surprised by this inaccuracy, as our generalized model only extend SHM by allowing different homophily levels, and we expect real social networks to be more complicated. For a better performance, a more complex model may be needed. 

{
{In the WhatsApp dataset, for the ratio of anti-BJP groups, our simulation with inferred parameters demonstrates the glass-ceiling effect, but no longer the chasm effect, as the yellow line is monotonically decreasing; for the ratio of anti-BJP members, our simulation shows very weak glass-ceiling effect and the chasm effect, as the yellow line first goes up and then goes down, with a minor monotonicity change around group size being 90}. These mismatch can be caused by inaccurate estimations of the parameters, with several factors. For example, the WhatsApp dataset is gathered by collecting member information of 2,092 groups, while members in the collected groups are likely to join other groups in the big WhatsApp network. The missing information of the rest of groups these members join can lead to an under-estimation of $\eta$. Furthermore, the sparsity of the data can also lead to the inaccurate estimation of the other parameters. {To further test the performance of our model, we optimize over  $\eta, \xi, \rhop_R, \rhop_B, \rhou_R, \rhou_B$, instead perviously only $\xi, \rhop_R, \rhop_B, \rhou_R, \rhou_B$, to get a set of adjusted parameters.} We see that with the adjusted parameters, our model well-captures the chasm effects and the glass-ceiling effects, {as the green lines in both the anti-BJP group ratio plot and the anti-BJP member ratio plot clearly show that the numerical ratios first increase and then decrease.} Furthermore, the numerical ratio with adjusted parameters again locates the group size at which the ratio of anti-BJP groups changes the monotonicity.}

\section{Applications to Information Flow}
 
The presence of hegemony in networks have already been linked to important consequences on the fairness of many graph algorithms \cite{stoica2020seeding}. We now present two examples where our identified chasm effect, which contrasts with the tail effect, invites us to shed light on the fairness of targeted advertisement and content moderation.

\subsection{Job advertisements and equal opportunity among genders}

The nature of classified ads went through a seismic shift with the advent of Craigslist. Employers posted job opportunities online, providing an additional advantage to people with access to computers and good internet connections. In the last two decades, recruitment strategies evolved further; prospective employees are targeted on LinkedIn or Facebook based on self-uploaded or inferred profile data, which raises a myriad of issues\footnote{\url{https://www.vox.com/identities/2019/9/25/20883446/facebook-job-ads-discrimination.}}. These approaches can be exclusionary or discriminatory - perhaps inadvertently - and expensive. Nowadays, both hiring companies and recruitment companies post job openings in interest-based groups on social networks or popular job boards as a way to organically reach a larger, more diverse audience without paying a premium for targeted advertising. However, the make-up of groups and job boards may not uniform, and this strategy can impact the diversity of applications. Simply casting a wide net without attempting to understand the breakdown of the people on the platform may increase the gender imbalance. However, the chasm effect shows that there is a group-size threshold that, if adopted, can help ensure a more diverse net is cast with the job posting reaching more women. Acknowledging the existence of this threshold and attempting to determine the optimal threshold could go a long way in reducing the implicit biases in the hiring process. 
 
In detail, consider the advertising strategy that places ads for groups with size greater than or equal to $k_A$. Let $r^{(A)}(k_A)$ be the ratio of red members among all the members seeing the ads, in the limit $t \rightarrow \infty$. We have the following theorem, whose proof is delayed to Appendix \ref{proof_of_ads}.
\begin{thm}\label{ads_theorem}(proof in appendix \ref{proof_of_ads})
	Assume the red member ratios for very small and large groups are smaller than the average red member ratio $r$ in the network. There exist $0 < k_A^{lower} \leq k_A^{upper}$, such that
	\begin{itemize}
		\item For $k_A > k_A^{upper}$, $r^{(A)}(k_A) < r$;
		\item For $k_A < k_A^{lower}$, $r^{(A)}(k_A) > r$.
	\end{itemize}
\end{thm}

We examine this result empirically on the QQ dataset, and we see (in Figure \ref{advertisement}) that we can choose $k_A^{lower}=k_A^{upper}=63$. That is, if the group-size threshold is larger than 63, the advertising strategy favors males; on the other hand, if it is less than or equal to 63, the advertising strategy favors females.
\vspace{-5mm}
\begin{figure}[H]
\centering
\includegraphics[scale=0.6]{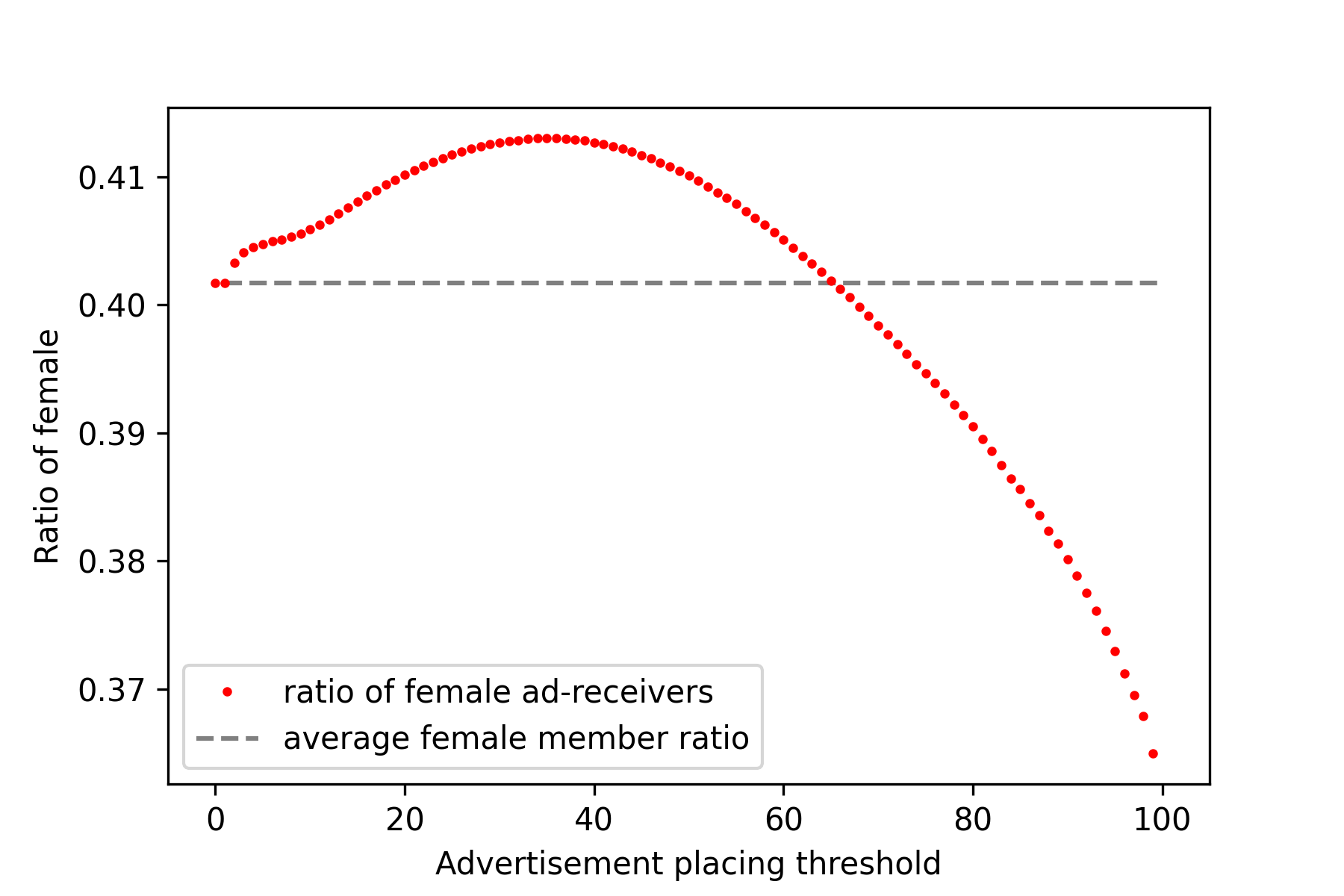}
\caption{Advertisement simulations on QQ: consider the advertising strategy that places ads to groups larger than a threshold. We see as the threshold gets lower, the advertising strategy changes from favoring males to females in terms of member exposures.}
\label{advertisement}
\end{figure}

\subsection{Preventing disproportionate content moderation among political affiliations?}
The conditions of the information landscape have deteriorated significantly over the course of the last decade. Conspiracy theories, flagrant rumours about people and events, and hateful content have all been amplified to the detriment of society at large. Group chats and end-to-end encrypted chats have not escaped this fate-rampant with the same malcontent, they have the additional problem where a lot of conversations are not subject to scrutiny. One of the many approaches to address this ecosystem is to rely on fact-checkers who identify pieces of information to verify and provide in-depth analysis into their veracity. Fact-checking organizations scour different parts of the open web and social platforms to determine what to fact-check\footnote{\url{https://www.boomlive.in/methodology}.}, and many have also set up additional tip-lines as one force to counteract the widespread misinformation\footnote{\url{https://meedan.com/blog/one-of-year-of-running-the-end-end-to-fact-checking-project/}.}. Different fact-checking organizations have different strategies in terms of prioritizing what to fact-check. Typically, it is based on a combination of importance (e.g. elections), relevance (e.g. breaking news events), the number of times an individual piece of content has been flagged, and the number of platforms on which it has been flagged.

In a highly simplified scenario where people have an equal tendency to report fake news when they see it, and the fact-checkers always prioritize to check news with more reports, one could ask the question whether prioritizing based on the number of reports is fair. As news from larger groups is more likely to be checked, the glass-ceiling effect implies that relative to the majority, fake news that originates or spread among minority members might be less likely to be detected and removed; however, the chasm effect shows that this is not necessarily true. 

Assume that the probability of fake news being detected in a group depends on the group size and the likelihood of all pieces of malcontent being detected. For simplicity, let $\theta \in [0,1]$ be the strength of the detector, with $\theta = 1$ indicating all fake news will be detected, and $\theta = 0$ indicating nothing will be flagged for a fact-check. For each group with size $k$, denote $h(k, \theta)$ as the probability that fake news in the group is detected. Equivalently, $h(k, \theta)$ is the expected ratio of detected fake news over all fake news in the group. We assume the function $h(\cdot, \cdot)$ satisfies:

\begin{enumerate}
	\item $h(\cdot, \cdot)$ is monotone increasing in group size: $h(k, \theta) < h(k+1, \theta)$.
	\item $h(\cdot, \cdot)$ is monotone increasing in detecting strength: $h(k, \theta_1) < h(k, \theta_2)$ for $0\leq \theta_1 < \theta_2 \leq 1$;
	\item $h(k, 0) = 0, h(k, 1) = 1$, and
		\begin{equation}
		\label{eq_h_assumption}
			\lim_{\theta \rightarrow 0} \frac{h(k, \theta)}{h(k+1, \theta)} = 0, \,\,\,\,\,\,\,\,\,\,\,\,
			\lim_{\theta \rightarrow 1} \frac{1-h(k, \theta)}{1-h(k+1, \theta)} = \infty.
		\end{equation}
\end{enumerate}

We make Assumption (1) since fake news in larger groups is likely to be reported more times, and therefore has a higher probability to be detected. Assumption (2) makes sense since $\theta$ measures the detector's strength.
The last Assumption (3) is a technical assumption, which means that groups with larger sizes dominate groups with smaller sizes, in the sense that a) as the strength of the detector goes to 0, $h(k, \theta)$ goes to 0 faster than $h(k+1, \theta)$ and b) as the strength of the detector goes to 1, $h(k, \theta)$ goes to 1 slower than $h(k+1, \theta)$.

Regard $h(k, \theta)$ as the protection score of a group with size $k$, and let $r^{(D)}(\theta)$ be the ratio of red groups' scores over total scores, that is,
\begin{align}\label{protection_score}
	r^{(D)}(\theta)=\frac{\sum_{k \geq 1} G_k(R)h(k, \theta)}{\sum_{k \geq 1} (G_k(R)+G_k(B))h(k, \theta)}.
\end{align}

\begin{thm}\label{fake news theorem}(proof in appendix \ref{proof_of_fn}) Assume the red group ratio $G_k(R)/(G_k(R) + G_k(B))$ is less than the overall red group ratio $r$ for very small and large groups in the network as $t\rightarrow \infty$. Then there exist $0 < \theta^{lower} \leq  \theta^{upper} < 1$, such that
	\begin{itemize}
		\item For $\theta > \theta^{upper}$, $r^{(D)}(\theta) > r$;
		\item For $\theta < \theta^{lower}$, $r^{(D)}(\theta) < r$.
	\end{itemize}
\end{thm}

We examine this result empirically in WhatsApp with a simulated fact-checking system. Assume that for fake news to be detected, it first needs to be reported to a fact-checking organization who prioritizes the fact-check. We assume that the number of reports received in a group of size $k$ follows the Poisson distribution with the parameter being $p\cdot k$, where $p$ captures the tendency of reporting fake news in the network. Without further assumptions, we set $p=0.5$. The fact-checker ranks all the reported pieces of content by volume; if two items have the same number of reports, the fact-checker ranks the one from the larger group higher. Finally, the fact-checker sets a percentage threshold $P$ to check items ranked within the top $P\%$ ranked items. We repeat this simulation 100 times, and report our findings in Figure \ref{fake news}. 

Note that the percentage threshold corresponds to the likelihood of all pieces of malcontent being detected. We see, in Figure \ref{fake news} (a) as more fake news is detected, the protection ratio crosses the average anti-BJP ratio. That is, if the fact-checking organizations focus purely on the volume of reports, it favors the majority. If there is an opportunity, however, to apply more resources to the fact-checking initiatives, {the fact-checker starts to protect more minorities, as the protection ratio becomes above the average anti-BJP ratio}. Similar trends are found also for the ratio of number of times red groups are checked, the ratio of the total number of people protected in red, and the ratio of the total number of red members getting protected.

\begin{figure}
\centering
\includegraphics[scale=0.45]{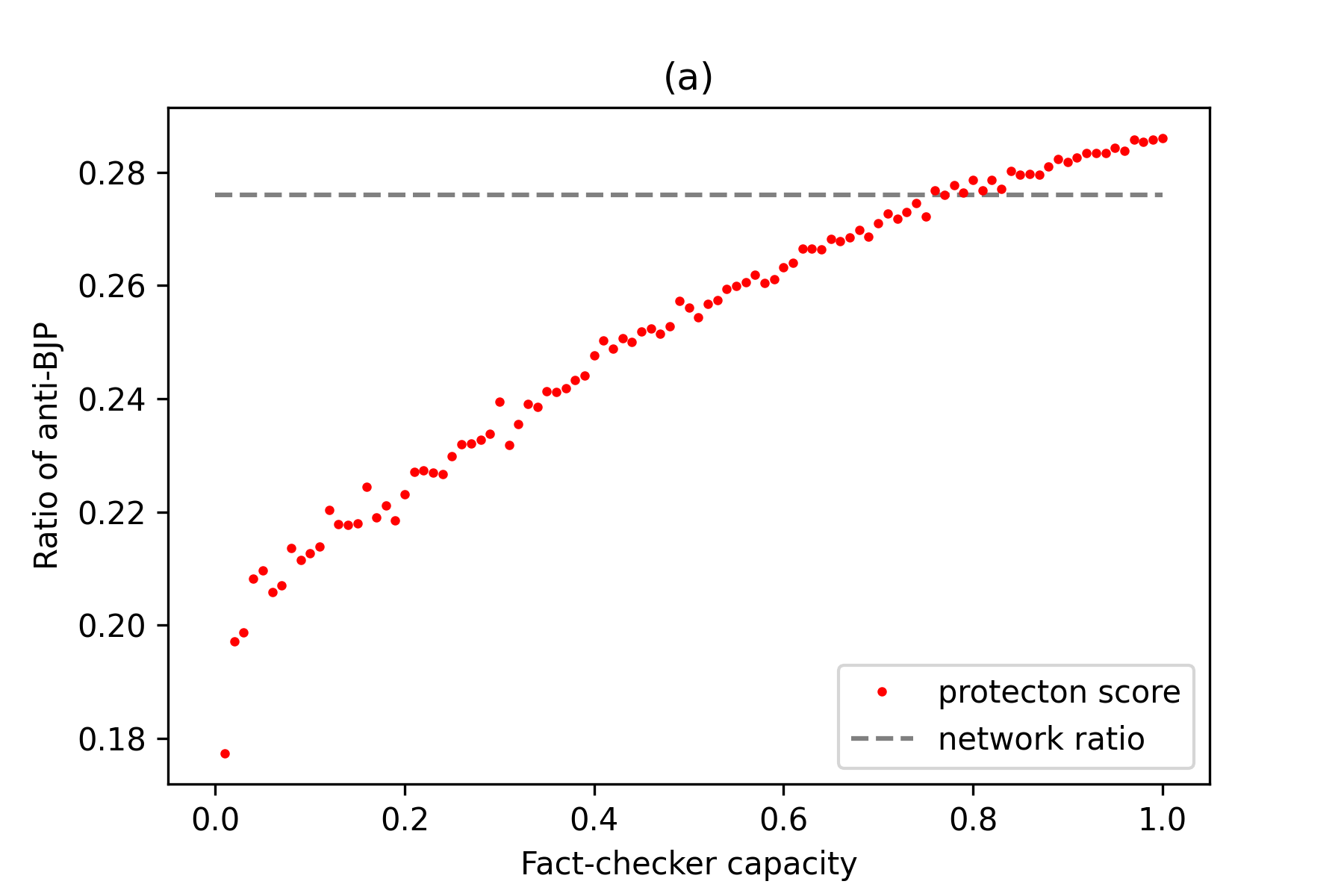}
\includegraphics[scale=0.45]{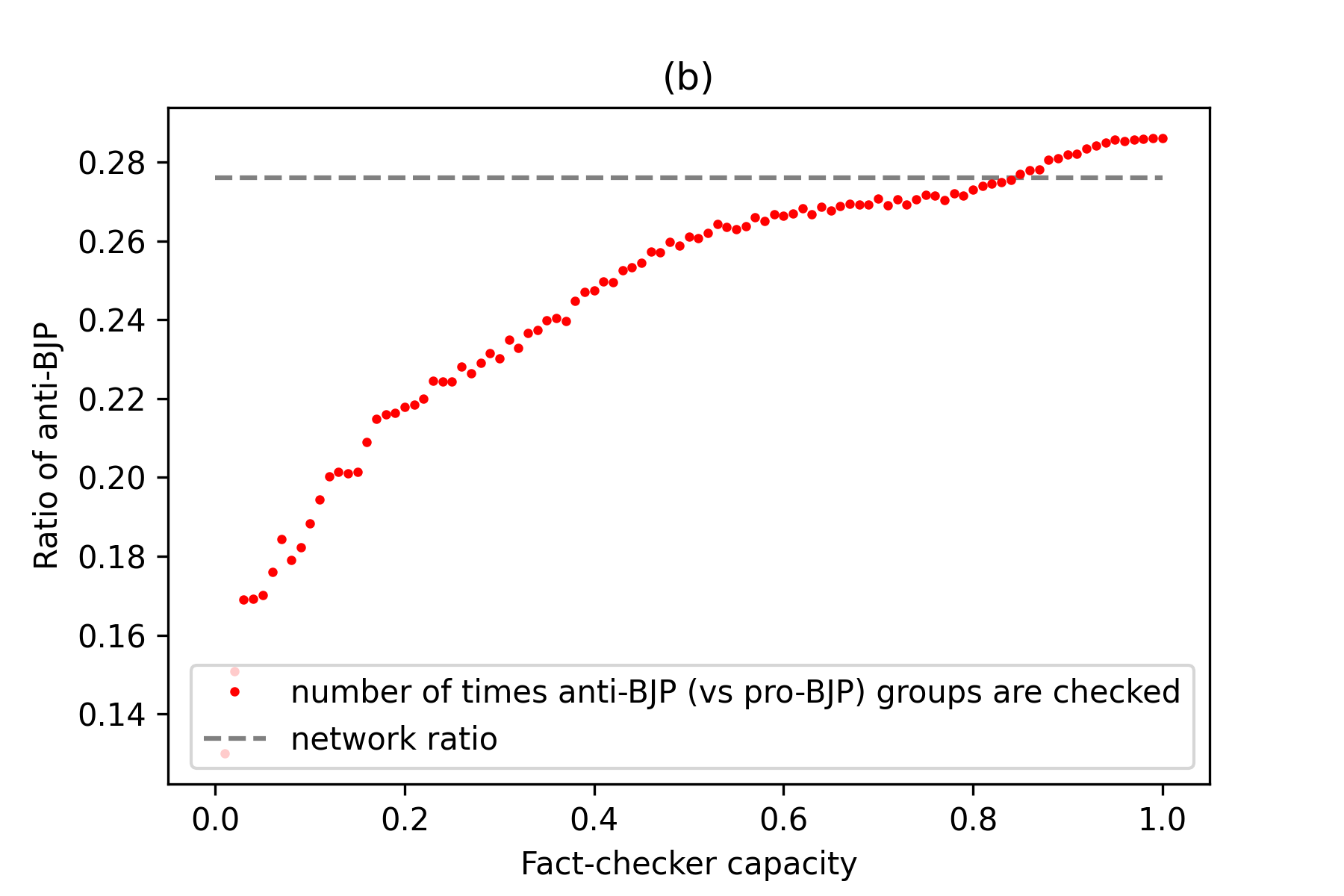}\\
\includegraphics[scale=0.45]{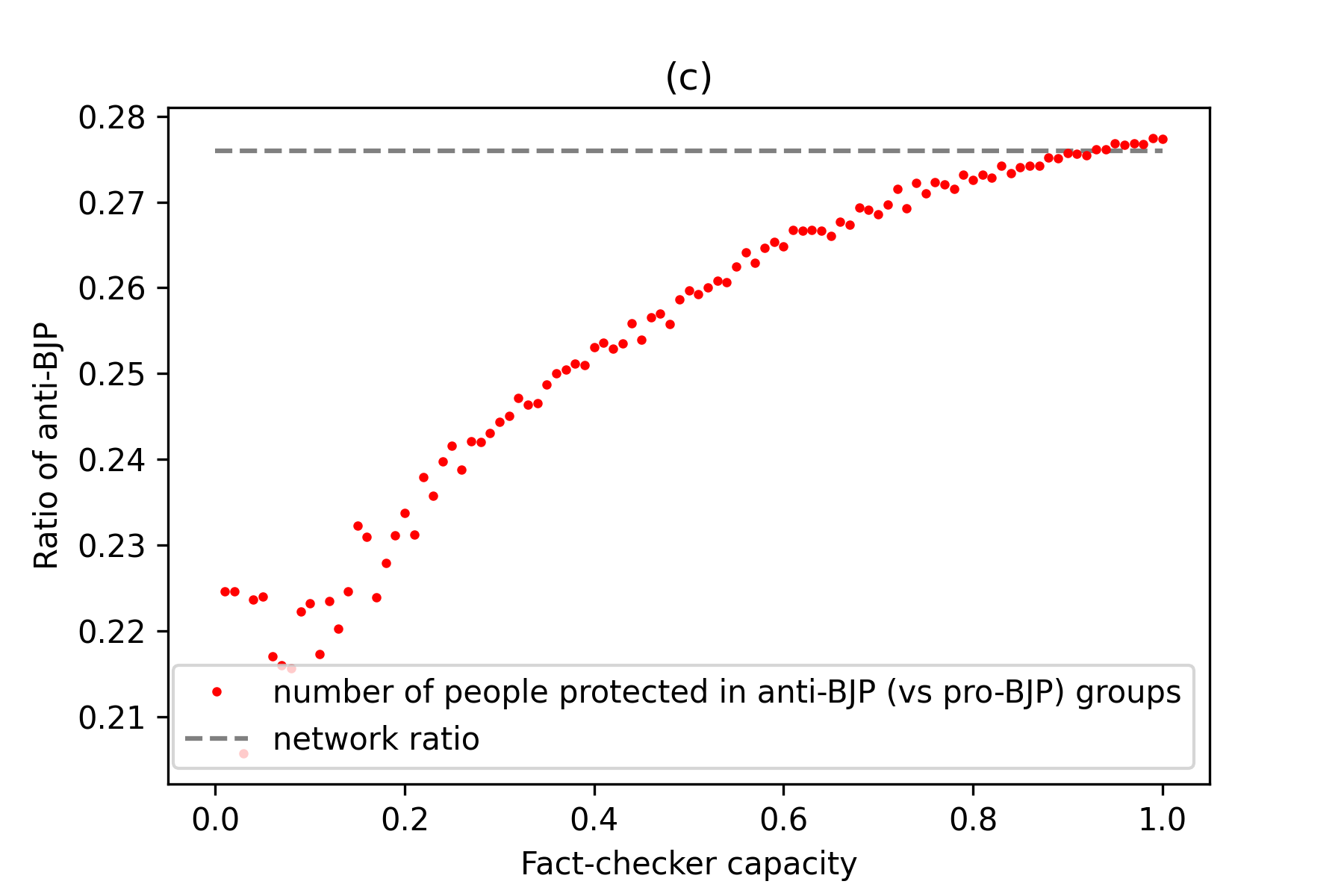}
\includegraphics[scale=0.45]{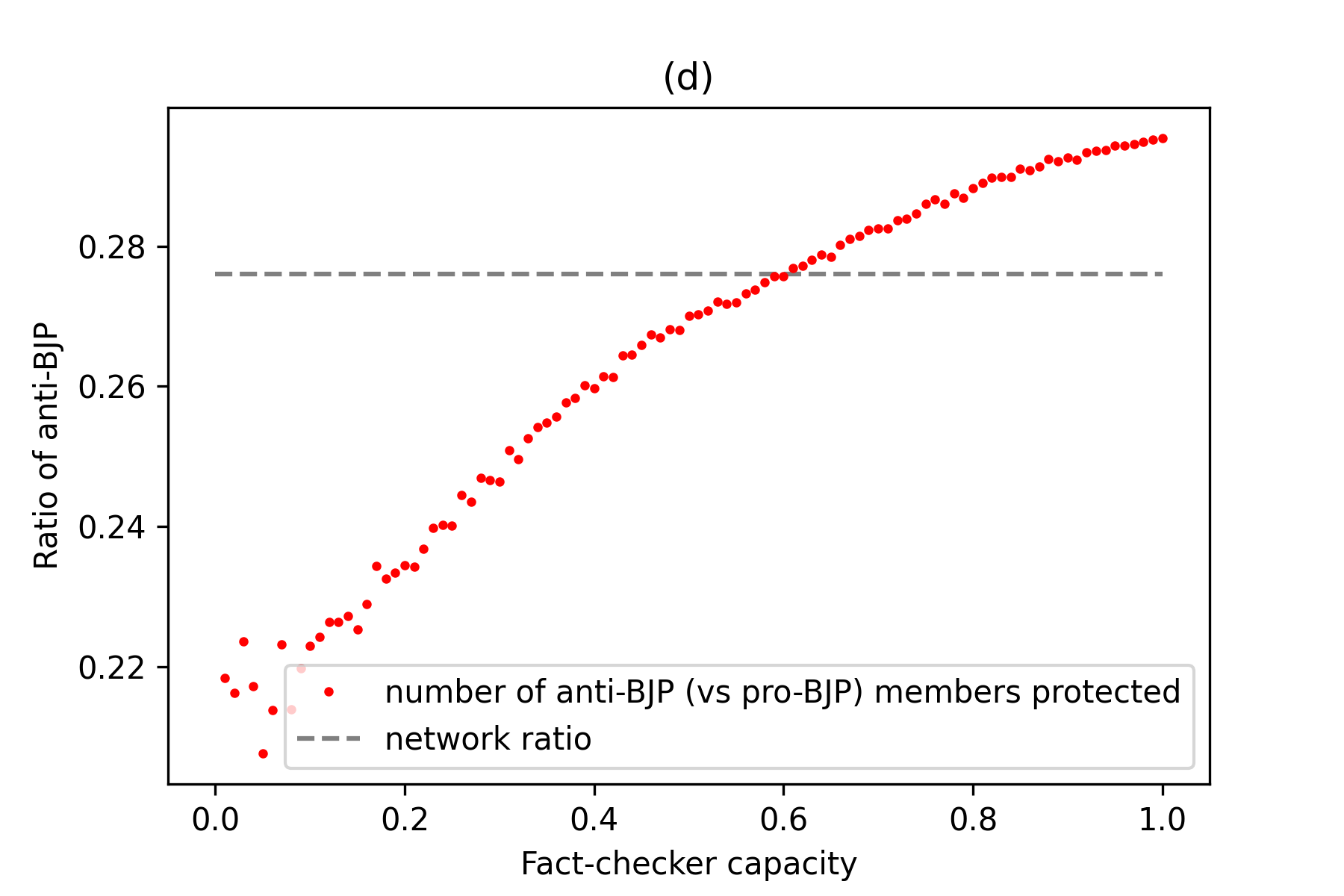}
\caption{Fact-checking simulations on WhatsApp: the glass-ceiling effect indicates that fact-checkers always protect more majority; however, we see that, as the detection strength gets larger, fact-checkers start protecting more minority than majority.}
\label{fake news}
\end{figure}

\section{Conclusions}

The graphs formed among us, as the structures of groups and communities connecting every individual, govern how today's information propagates and gets selectively curated. Bias emerges quickly and interacts with the simplest network primitives as well as complex algorithmic rules. This bias contributes to unequal opportunity among genders or disproportionate effects along political lines. Our results confirm that homophilous and rich get richer dynamics in the graph itself play a critical role in shaping the bias observed among multiple domains, paving the way for finding a common ground to counteract observed disparities. As our theoretical results suggest and empirical results confirm, the bias inside the tail or within the bulk of a popularity distribution can widely vary in orientation. We refer to this as a chasm between seemingly opposing views, but explain that its causes are not always in disparate treatment but may be simple systemic effects of selective homophily. This observation is critical as previous predictions of algorithmic bias on the tail are sometimes diametrically opposed to the case when a similar metric is examined at the lower end, including selecting items for fact-checking or choosing groups for targeted advertisement.

To keep our model generally applicable, we focus on the most commonly found dynamics which spans a range where popularity either plays no role or is entirely responsible for growth. This allows us to identify the necessary and sufficient conditions for the observed chasm to emerge, but that remains a crude unifying model that leaves many domain specific effects aside. We hope that our results encourage a renewed interest in a holistic view of either equitable representation or fairness guarantees for online content moderation. While each of those applications is beyond the scope of this paper, the empirical presence of a chasm and our simulations already suggests that, in order to achieve this goal, a new analysis beyond a narrow focus on tail effects is critical.

\section*{Acknowledgement}
This material is based upon work supported by the National Science Foundation under Grant No. 1761810. We would like to express our appreciation to Dr. Kiran Garimella and Prof. Dean Eckles from Massachusetts Institute of Technology for sharing their collection of the WhatsApp data, to Archis Chowdhury from BOOM for sharing his fact-checking experience with us, and to Ana-Andreea Stoica and Roland Mao for being our first readers. We are also very grateful for the generous support from the Data Science Institute and the Tow Center for Digital Journalism at Columbia University.

\bibliographystyle{ACM-Reference-Format}
\bibliography{draft1018}

\appendix

\section{Chasm effects in unipartite networks}\label{unipartite analysis}
Although we focus on the chasm effects in bipartite networks, the chasm effects and the glass-ceiling effects can also be observed and studied in unipartite social network. With a projected QQ membership network and an Instagram network, we examine the chasm effects in one-mode social networks.

\subsection{Unipartite network datasets}

\textbf{Instagram dataset\cite{stoica2018algorithmic}}
Instagram is a photo- and video-sharing platform where people like and comment on content. The Instagram dataset was collected between 2014 and 2015 and has a total of 553,628 different users whose genders were inferred from their names. Females make up 54.4\% of all users in this dataset. Even though females makeup more than half of the data, they are still considered the disadvantaged for two reasons: (1) other features of this network, like the degree distribution, suggest a bias against the female users; (2) we want to keep it consistent with prior works which have used this dataset.
\\\\
\textbf{Projected QQ membership}
We project the QQ group-member network introduced before to construct a QQ membership network. The projected QQ membership network consists of the node-set, and two nodes are joined by an undirected edge if and only if they are in the same group in the group-member network. If the two members share multiple groups, they are connected by multiple undirected edges.

\subsection{A model for unipartite networks}
Like in the bipartite networks, the observed chasm effects in unipartite networks can be captured in a homophilous model that combines the rich-get-richer and the equal-chance mechanism. We show that an simple extension of the unipartite model introduced in \cite{avin2015homophily} exhibits both chasm effect and glass-ceiling effect.

Specifically, at time $t = 2$, we initialize the unipartite network with one red member connecting to a blue member, and without loss of generalit, we define the collection of red members as the minority of the network.
\begin{itemize}
\item \textbf{Member Growth}: at time $t$, a new member $m^*$ joins the network.
\begin{itemize}
\item (\emph{minority-majority}) with probability $r$ ($0<r\leq 1/2$), the new member $m^*$ is colored red; 
\end{itemize}
\item \textbf{Connection Growth}: $m^*$ connects with an existing member, according to the following two steps:
    \begin{itemize}
	\item (\emph{rich-get-richer}) with probability $\xi$, $m^*$ picks a member $m$ with probability proportional to $\text{deg}(m)$. If $c(m^*) = c(m)$, $m^*$ connects wtih $m$ directly; otherwise, $m^*$ accepts the connection with probability $\rhop_{c(m^*)}$. If $m^*$ does not accept the connection, $m^*$ restarts from the beginning of the \emph{Connection Growth} until a new connection is built.    
    \item (\emph{equal-chance}) with probability $1-\xi$, $m^*$ uniformly picks a member $m$ at random. If $c(m^*) = c(m)$, $m^*$ connects with $m$ directly; otherwise, $m^*$ accepts the connection with probability $\rhou_{c(m^*)}$. If $m^*$ does not accept the connection, $m^*$ restarts from the beginning of the \emph{Connection Growth} until a new connection is built. 
\end{itemize}
\end{itemize}

Let $U_k(C)$ be the number of members with color $C$ and $k$ degrees. Following exact the same analysis approach as in the bipartite GSHM model, we get the degree distribution for red members and blue members as follows:

\begin{equation}\label{red_uni}
    U_1(R) = \frac{r }{1+CU_{R,1} + CU_{R,2}}, \,\,\,\,\,\,
    U_k(R) = U_{k-1}(R)\frac{(k-1)C_{R,1} + C_{R,2}}{1 + k \cdot CU_{R,1} + CU_{R,2}} \,\,\,\,\, \forall k \geq 2;
\end{equation}
\begin{equation}\label{blue_uni}
    U_1(B) = \frac{(1-r)}{1+CU_{B,1} + CU_{B,2}}, \,\,\,\,\,\,
    U_k(B) = U_{k-1}(B)\frac{(k-1)CU_{B,1} + CU_{B,2}}{1 + k\cdot CU_{B,1} + CU_{B,2}} \,\,\,\,\, \forall k \geq 2.
\end{equation}
Here,
\begin{align}
    CU_{R,1} & := \frac{1}{2}\left(\frac{r\xi}
    			{1-(1-\rhop_R)\xi (1-\alpha^*)-(1-\rhou_R)(1-\xi) (1-r)}
    			+ \frac{(1-r)\xi \rho_b^{(p)} }
    				{1-(1-\rhop_B)\xi \alpha^*-(1-\rhou_B)(1-\xi)r}\right),
        \\
    CU_{R,2} & := \frac{r (1-\xi) }
    			{1-(1-\rhop_R)\xi (1-\alpha^*)-(1-\rhou_R)(1-\xi) (1-r)}
    			+ \frac{ (1-r)\rhou_B(1-\xi)}
    				{1-(1-\rhop_B)\xi \alpha^*-(1-\rhou_B)(1-\xi)r};
\end{align}
and
\begin{align}
	CU_{B,1} & :=\frac{1}{2}\left( \frac{ (1-r)\xi }
    			{1-(1-\rhop_B)\xi \alpha^* - (1-\rhou_B)(1-\xi) r}
    			+ \frac{r\rhop_R\xi }
    				{1-(1-\rhop_R)\xi (1-\alpha^*) - (1-\rhou_R)(1-\xi) (1-r)}\right),
        \\
    CU_{B,2} & := \frac{(1-r)(1-\xi)}
    			{1-(1-\rhop_B)\xi \alpha^* - (1-\rhou_B)(1-\xi) r}
    			+ \frac{r\rhou_R (1-\xi) }
    				{1-(1-\rhop_R)\xi (1-\alpha^*) - (1-\rhou_R)(1-\xi) (1-r)},
\end{align}
where $\alpha u^*$ denotes the limit of the sum of degrees of red members over sum of degrees of all members, as $t$ goes to infinity, and $\alpha u^{*}$ is the unique solution in $(0,1)$ that satisfies
\begin{align}
	\alpha u^* &= \frac{1}{2}\left(1 + r  \frac{\xi \cdot\alpha u^* + (1-\xi)r}{1- (1-\rho_r^{(p)})\xi (1-\alpha u^*) - (1-\rho_r^{(u)})(1-\xi)(1-r)}\right.\\
	&-\left.(1-r)\frac{\xi(1-\alpha u^*) + (1-\xi)(1-r)}{1-(1-\rho_b^{(p)})\xi \cdot\alpha u^* -(1-\rho_b^{(u)})(1-\xi)r })
	\right).
\end{align}

Furthermore, let $r^{M,N}_{k,t}(R,C)$ denote the ratio of expected red connections for members of color $C$ with degree $k$ at time $t$. With the same techniques as in Lemma \ref{lemma_binomial_result_general}, we have 
\begin{equation}
	r^{(M, N)}_{k} (R,R) :=\lim_{t \rightarrow \infty} r^{(M, N)}_{k, t} (R,R) =  \frac{\sum_{j=1}^k pu_{RR,j}}{k},
\end{equation}
\begin{equation}
	r^{(M, N)}_{k} (R,B) :=\lim_{t \rightarrow \infty} r^{(M, N)}_{k, t} (R,B) =  \frac{ \sum_{j=1}^k pu_{RB,j}}{k},
\end{equation}
where
\begin{align}
pu_{RR,1} &= \frac{r(\xi \cdot\alpha u^* + (1-\xi)r)}{r(\xi\cdot\alpha u^* + (1-\xi)r) + (1-r)(\xi \cdot\alpha u^* \rho_b^{(p)} + (1-\xi)r\rho_b^{(u)})}\\
pu_{RB,1} &= \frac{r(\xi (1-\alpha u^*)\rho_r^{(p)} + (1-\xi)(1-r)\rho_r^{(u)})}{r(\xi (1-\alpha u^*)\rho_r^{(p)} + (1-\xi)(1-r)\rho_r^{(u)}) + (1-r)(\xi (1-\alpha u^*) + (1-\xi)(1-r))}\\
pu_{RR,j} &= \frac{pu^{(0)}_{RR,j}}{pu^{(0)}_{RR,j} + pu^{(0)}_{BR,j}}, \,\,\,\,\,\,\,\,\,\,
pu_{RB,j} = \frac{pu^{(0)}_{RB,j}}{pu^{(0)}_{RB,j} + pu^{(0)}_{BB,j}}, \,\,\,\,\,\,\,\,\,\,
 j\geq 2
\end{align}
with
\begin{align}
	pu^{(0)}_{RR,j} &= r\frac{\xi j/2 + (1-\xi)/\alpha}{1-(1-\rhop_R)\xi(1-\alpha u^*) - (1-\rhou_R)(1-\xi) (1-r)}, \\
	pu^{(0)}_{BR,j} &= (1-r)\frac{\rhop_B \xi j/2 + \rhou_B(1-\xi)/\alpha}{1-(1-\rhop_B)\xi \cdot\alpha u^* - (1-\rhou_B)(1-\xi) r},\\
	pu^{(0)}_{RB,j} &= r\frac{\rhop_R \xi j/2 + \rhou_R(1-\xi)/\alpha}{1-(1-\rhop_R)\xi(1-\alpha u^*) - (1-\rhou_R)(1-\xi) (1-r)}, \\
	pu^{(0)}_{BB,j} &= (1-r)\frac{\xi j/2 + (1-\xi)/\alpha}{1-(1-\rhop_B)\xi \cdot\alpha u^* - (1-\rhou_B)(1-\xi) r}.
\end{align}
Given $k$, the expected ratio of female connection for members of degree $k$ is 
\begin{align}\label{unipartite member ratio}
r^{(M, N)}_{k} (R):=\frac{r^{(M, N)}_{k} (R,R)\cdot U_k(R) +r^{(M, N)}_{k} (R,B)\cdot U_k(B)}{U_k(R) + U_k(B)}.
\end{align}

\subsection{Fitting the unipartite model on real data}
Note that $r$ can be inferred directly from the dataset. The rest of the parameters $\xi, \rho_r^{(p)}, \rho_b^{(p)}, \rho_r^{(u)}, \rho_b^{(u)}$ can be inferred by finding an optimal solution that minimizes $\vert\sum_{k=1}^{n}\frac{\hat{U_k}(R)}{\hat{U_k}(B)} - \sum_{k=1}^{n}\frac{{U_k}(R)}{{U_k}(B)}\vert$, where $\hat{U_k}(R)$ and $\hat{U_k}(B)$ are obtained through (\ref{red_uni}) and (\ref{blue_uni}), and ${U_k}(R)$ and ${U_k}(B)$ are empirical values observed from the dataset. With the set of $r, \xi, \rho_r^{(p)}, \rho_b^{(p)}, \rho_r^{(u)}, \rho_b^{(u)}$, we can use (\ref{unipartite member ratio}) to obtain the numerical values for the ratio of female connection.

\begin{figure}[H]
\includegraphics[height=0.225\textheight, width=0.5\linewidth]{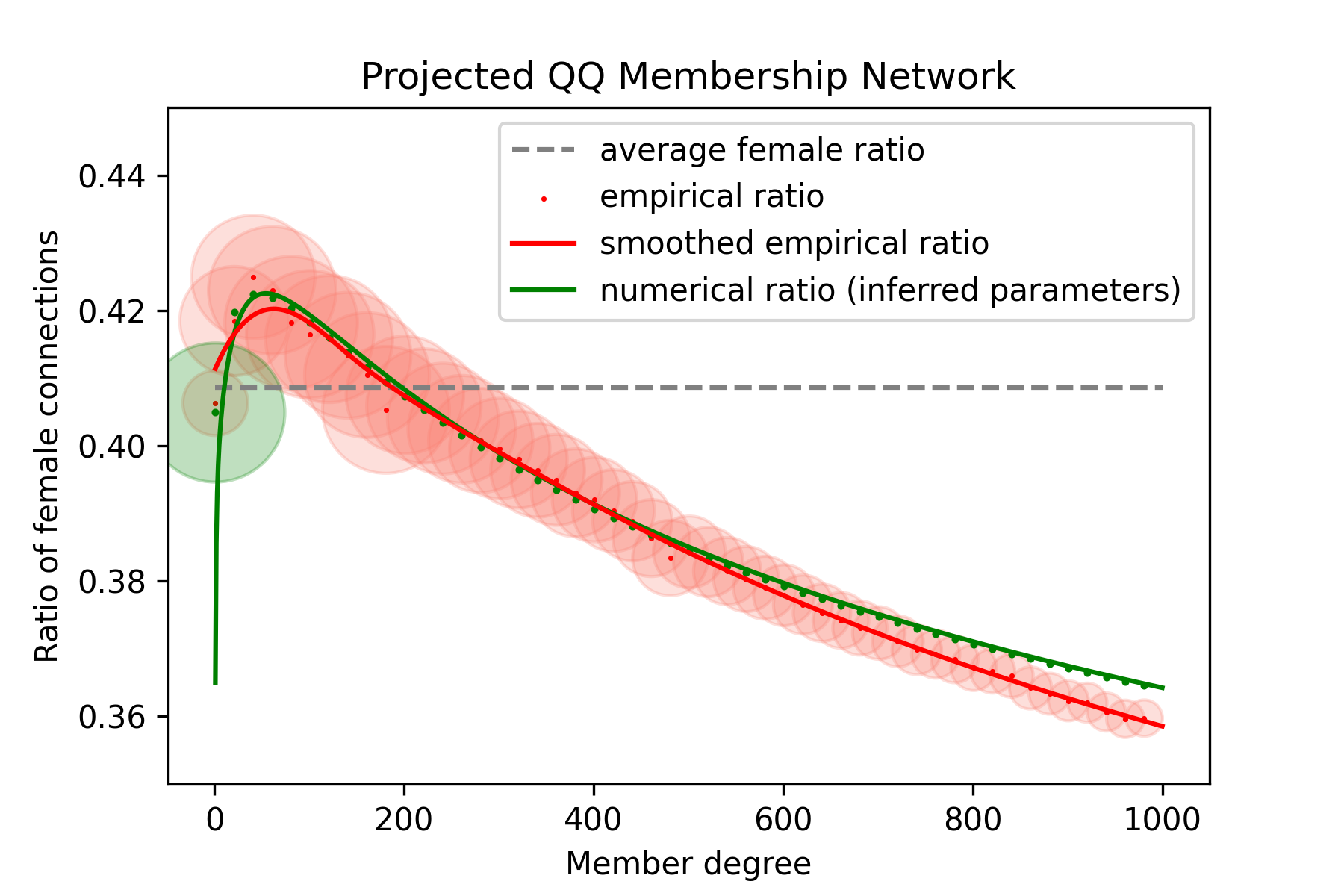}
\includegraphics[scale=0.56]{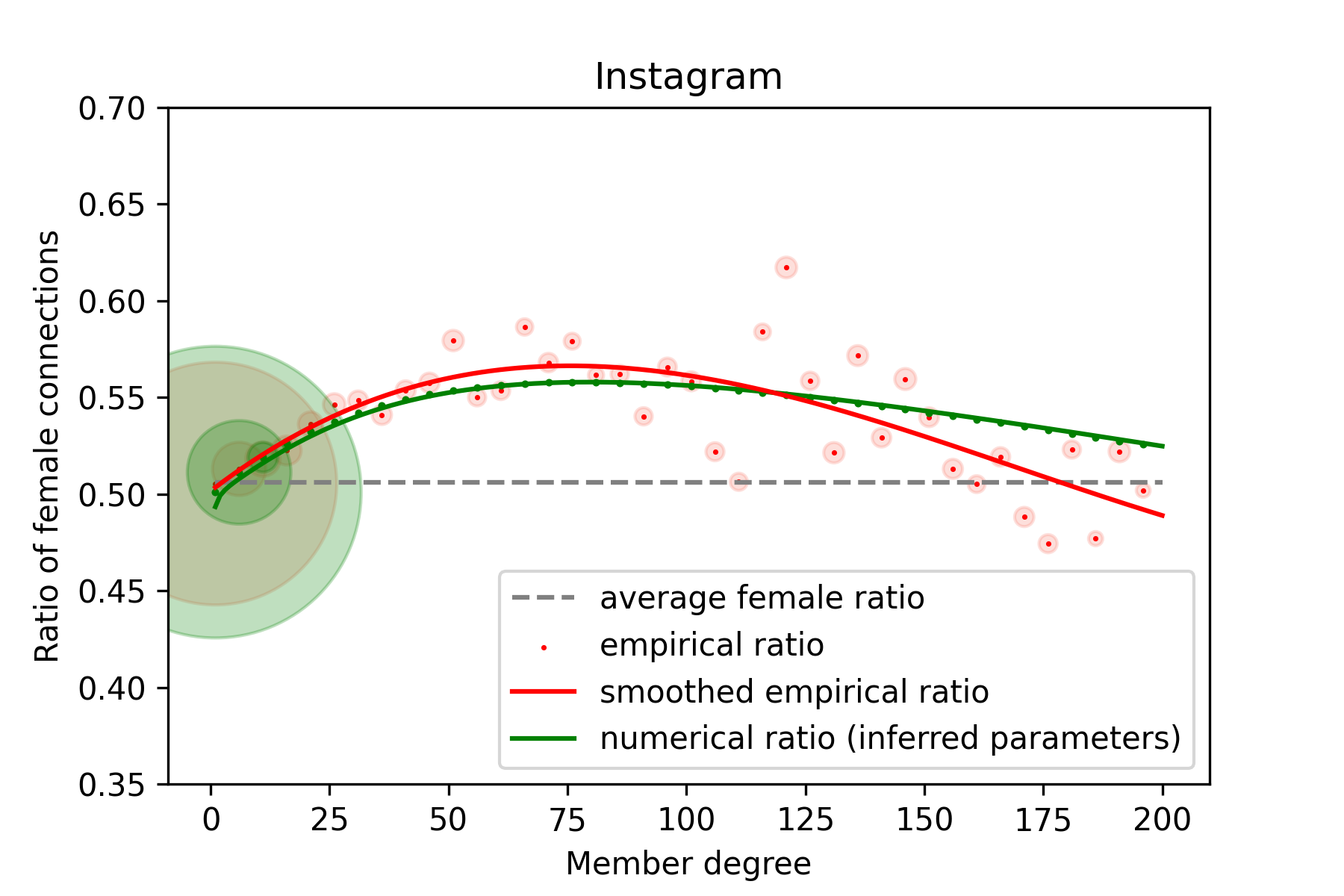}
\caption{The radii of the red and green circles represent the member degree distributions in the empirical datasets and in the numerical simulations respectively. Our unipartite model captures the chasm effect, the glass ceiling effect, as well as the member degree at which the ratio changes its monotonicity in the two unipartite network datasets. The model also captures the distribution of member degree for Instagram.}
\label{unipartite_fit}
\end{figure}

We observe in Figure \ref{unipartite_fit} that our unipartite model clearly captures the chasm effect and the glass-ceiling effect for both the projected QQ membership network and the Instagram network. The model also locates the member degree on which the ratio of female connections changes monotonicity well for both datasets. Furthermore, our unipartite model also captures the member degree distribution in the Instagram dataset, but fails to capture the member degree distribution in the projected QQ membership network. This is not surprising, as the projection can connect strangers that are within the same group in the bipartite network as friends in the projected unipartite network, and gives nodes in the unipartite network relatively high degree.

\section{An adjusted GSHM model}\label{gender ratio color}
In both the SHM and the GSHM models, we assume that the group creator always chooses the topic of the creator's own affiliation, and can maintain the chosen topic as the main focus of the group. This assumption is particularly applicable to the political party setting. However, the assumption can be less obvious in the study of demographic imbalance, as males can also create supportive groups for females. In the later scenario, we can adjust the \emph{Group Growth} step in the GSHM model as the following: 
\begin{itemize}
	\item \textbf{Group Growth}: with probability $\eta$ ($0<\eta <1$), the member creates a group, and the group is colored red with probability $r$ and blue with probability $1-r$.
\end{itemize}
 We refer this adjusted model as \emph{adjusted-GSHM}.

It is easy to check that the adjust-GSHM model is \emph{almost} equivalent to the GSHM model, in the sense that Theorem \ref{group size power-law}, Theorem \ref{member degree power-law}, Theorem \ref{lemma_bump}, and their corresponding corollaries for the GSHM model apply exactly to the adjust-GSHM model. The only difference occurs in Lemma \ref{member_ratio_nonmonotone}, where (\ref{rm_ratio_1}) needs to be changed to 
	\begin{align}
		\lim_{t \rightarrow \infty} r^{(M, G)}_{1, t} (R)&= r.
	\end{align}

Therefore, the adjusted-GSHM model is equivalent to the GSHM model when studying the change of the ratio of minority group over different group sizes, and is equivalent to the GSHM model when studying the change of the ratio of minority group member over groups with large sizes.

\section{Proof of Theorem \ref{group size power-law}}\label{proof of thm 5.1}

\begin{reptheorem}{group size power-law}
Let $\{N(M\cup G, t, \Theta\}$ be a sequence of networks produced by the BGMG model. Assume that $\rhop_R, \rhop_B > 0$. The red group-size distribution $G_k(R)$ and the blue group-size distribution $G_k(B)$ asymptotically follow the power law distributions; specifically, as $t$ goes to infinity,
\begin{align}
G_k(R)\propto k^{-\beta(R)}, \,\,\,\,\,\,
G_k(B)\propto k^{-\beta(B)},
\end{align}
with $ \beta(R) = 1 + \frac{1}{C_{R,1}}$ and $\beta(B) = 1 + \frac{1}{C_{B,1}}$, where
\begin{align}
    C_{R,1} & := \frac{r (1-\eta) \xi }
    			{1-(1-\rhop_R)\xi (1-\alpha^*)-(1-\rhou_R)(1-\xi) (1-r)}
    			+ \frac{(1-r) (1-\eta)\rhop_B\xi }
    				{1-(1-\rhop_B)\xi \alpha^*-(1-\rhou_B)(1-\xi)r},
        \\
	C_{B,1} & := \frac{(1-r) (1-\eta) \xi }
    			{1-(1-\rhop_B)\xi \alpha^* - (1-\rhou_B)(1-\xi) r}
    			+ \frac{r (1-\eta)\rhop_R\xi }
    				{1-(1-\rhop_R)\xi (1-\alpha^*) - (1-\rhou_R)(1-\xi) (1-r)},
\end{align}
and $\alpha^*$ is the unique number $\in (0,1)$ satisfying
    	\begin{equation}
    		\label{eq_alpha_star}
    	\alpha^* = r \eta + \frac{r (1-\eta) (\xi \alpha^* + (1-\xi)r) }
    			{1-(1-\rhop_R)\xi (1-\alpha^*)-(1-\rhou_R)(1-\xi) (1-r)}
    			+ \frac{(1-r) (1-\eta) (\rhop_B \xi \alpha^* + \rhou_B (1-\xi)r)}
    				{1-(1-\rhop_B)\xi \alpha^*-(1-\rhou_B)(1-\xi)r }.
    	\end{equation}

\end{reptheorem}

\begin{proof}
We develop a recurrence for $\mathbb{E}\left(G_{k,t}(R)\right)$. 
First, define
\begin{align}
	p_t^{RR}(k) &:= \mathbb{P}(\text{a red member joins a red group with size $k$ at time t}), \\
	p_t^{BR}(k) &:= \mathbb{P}(\text{a blue member joins a red group with size $k$ at time t}).
\end{align}

By our construction of the model, it is easy to check that,
\begin{align}\label{p_t_RR}
    p_t^{RR}(k) =& (\alpha r + (1-\alpha) r^{(E, M)}_t(R)) (1-\eta) \frac{\xi \frac{k}{t} + (1-\xi) \frac{1}{G_{t}(R)+G_{t}(B)}}
    	{1-(1-\rhop_R)\xi r_t^{(E,G)}(B) - (1-\rhou_R)(1-\xi) \frac{G_{t}(B)}{G_{t}(R)+G_{t}(B)}},\\
    p_t^{BR}(k) =& (\alpha (1-r) + (1-\alpha) (1-r^{(E, M)}_t(R))) (1-\eta) \frac{\rhop_B \xi \frac{k}{t} + \rhou_B (1-\xi) \frac{1}{G_{t}(R)+G_{t}(B)}}{1-(1-\rhop_B)\xi r_t^{(E,G)}(R) - (1-\rhou_B)(1-\xi) \frac{G_{t}(R)}{G_{t}(R)+G_{t}(B)}}. 
\end{align}
Note that a red group of degree $k$ at time $t+1$ could have arisen from three scenarios:
\begin{enumerate}
    \item at time $t$, it was a red group of size $k$, and no new member joins at time $t+1$;
    \item at time $t$, it was a red group of size $k-1$, and a new member joins at time $t+1$;
    \item in the special case of $k=1$, a red group did not exist at time $t$ can appear if a red person creates it.
    \end{enumerate}
Therefore,
\begin{align}
    \mathbb{E}\left(G_{k,t+1}(R)\vert \mathcal{F}_t\right) 
    	= & G_{k,t}(R)\left(1-p_t^{RR}(k) - p_t^{BR}(k) \right) \\
    +& G_{k-1,t}(R)\left(p_t^{RR}(k-1) + p_t^{BR}(k-1) \right),
\end{align}
where $\mathcal{F}_t$ is the $\sigma$-field containing the information of the graph until time $t$. Note that
\begin{align}
    p_t^{RR}(k) + p_t^{BR}(k) = \frac{A_t(R) k + B_t(R)}{t},
\end{align}

\begin{align}\label{A_t}
    A_t(R)&: =   \frac{(\alpha r + (1-\alpha) r^{(E, M)}_t(R)) (1-\eta) \xi }
    				{1-(1-\rhop_R)\xi r_t^{(E,G)}(B) - (1-\rhou_R)(1-\xi) \frac{G_{t}(B)}{G_{t}(R)+G_{t}(B)}}\\
    			&+  \frac{(\alpha (1-r) + (1-\alpha) (1-r^{(E, M)}_t(R))) (1-\eta)\rhop_B \xi}
    			{1-(1-\rhop_B)\xi r_t^{(E,G)}(R) - (1-\rhou_B)(1-\xi) \frac{G_{t}(R)}{G_{t}(R)+G_{t}(B)}},\\
    B_t(R)&: = \frac{(\alpha r + (1-\alpha) r^{(E, M)}_t(R)) (1-\eta) (1-\xi) \frac{t}{G_{t}(R)+G_{t}(B)}}
    				{1-(1-\rhop_R)\xi r_t^{(E,G)}(B) - (1-\rhou_R)(1-\xi) \frac{G_{t}(B)}{G_{t}(R)+G_{t}(B)}}\\
    			&+  \frac{(\alpha (1-r) + (1-\alpha) (1-r^{(E, M)}_t(R))) (1-\eta)\rhou_B (1-\xi) \frac{t}{G_{t}(R)+G_{t}(B)}}
    			{1-(1-\rhop_B)\xi r_t^{(E,G)}(R) - (1-\rhou_B)(1-\xi) \frac{G_{t}(R)}{G_{t}(R)+G_{t}(B)}}.
\end{align}
We then have 
\begin{equation}
\label{eq_power_law_internal_1}
    \mathbb{E}\left(G_{k,t+1}(R)\vert \mathcal{F}_t\right) = G_{k,t} (R)\left(1-\frac{A_t(R) k + B_t(R)}{t}\right) + G_{k-1,t}(R) \frac{A_t(R) (k-1) + B_t(R)}{t}.
\end{equation}

When $k=1$, taking the probability of a red group being created into consideration, we have
\begin{equation}
\label{eq_power_law_internal_2}
\mathbb{E}\left(G_{1,t+1}(R)\vert \mathcal{F}_t\right) = G_{1,t} \left(1 - \frac{A_t(R) + B_t(R)}{t}\right) + \alpha \cdot r  \cdot \eta + (1-\alpha)\cdot r^{(E, M)}_t(R) \cdot \eta.
\end{equation}
By lemma \ref{lemma_pre_convergence}, We can show that
\begin{equation}
    \lim_{t\rightarrow \infty}  A_t(R) = C_{R,1}, \,\, 
    \lim_{t\rightarrow \infty}  B_t(R) = C_{R,2}, \,\, a.s,
\end{equation}
where
\begin{align}
\label{eq_cr}
    C_{R,1} & := \frac{r (1-\eta) \xi }
    			{1-(1-\rhop_R)\xi (1-\alpha^*)-(1-\rhou_R)(1-\xi) (1-r)}
    			+ \frac{(1-r) (1-\eta)\rhop_B\xi }
    				{1-(1-\rhop_B)\xi \alpha^*-(1-\rhou_B)(1-\xi)r},
        \\
    C_{R,2} & := \frac{r (1-\eta) (1-\xi) \frac{1}{\eta} }
    			{1-(1-\rhop_R)\xi (1-\alpha^*)-(1-\rhou_R)(1-\xi) (1-r)}
    			+ \frac{(1-r) (1-\eta)\rhou_B(1-\xi) \frac{1}{\eta} }
    				{1-(1-\rhop_B)\xi \alpha^*-(1-\rhou_B)(1-\xi)r}.
\end{align}
By Lemma \ref{lem3.1},  $G_k(R)$ has the following expressions:
\begin{equation}\label{power_law_fit1}
    G_1(R) = \frac{r \eta}{1+C_{R,1} + C_{R,2}}, \,\,\,\,\,\,
    G_k(R) = G_{k-1}(R)\frac{(k-1)C_{R,1} + C_{R,2}}{1 + k C_{R,1} + C_{R,2}} \,\,\,\,\, \forall k \geq 2.
\end{equation}
This completes the proof for $G_k(R)$, and we can use the same strategy for $G_k(B)$, and show that
\begin{equation}\label{power_law_fit2}
    G_1(B) = \frac{(1-r)\rho_B \eta}{1+C_{B,1} + C_{B,2}}, \,\,\,\,\,\,
    G_k(B) = G_{k-1}(B)\frac{(k-1)C_{B,1} + C_{B,2}}{1 + kC_{B,1} + C_{B,2}} \,\,\,\,\, \forall k \geq 2,
\end{equation}
where 
\begin{align}
\label{eq_cr}
	C_{B,1} & := \frac{(1-r) (1-\eta) \xi }
    			{1-(1-\rhop_B)\xi \alpha^* - (1-\rhou_B)(1-\xi) r}
    			+ \frac{r (1-\eta)\rhop_R\xi }
    				{1-(1-\rhop_R)\xi (1-\alpha^*) - (1-\rhou_R)(1-\xi) (1-r)},
        \\
    C_{B,2} & := \frac{(1-r) (1-\eta) (1-\xi) \frac{1}{\eta}  }
    			{1-(1-\rhop_B)\xi \alpha^* - (1-\rhou_B)(1-\xi) r}
    			+ \frac{r (1-\eta)\rhou_R (1-\xi) \frac{1}{\eta}  }
    				{1-(1-\rhop_R)\xi (1-\alpha^*) - (1-\rhou_R)(1-\xi) (1-r)}.
\end{align}
Using the same argument of the proof of \cite[Theorem 4.12]{avin2015homophily} completes the proof of the power law results.
\end{proof}
\begin{lemma}{\cite[Lemma 3.1]{chung2006complex}}
\label{lem3.1}
Let $(a_t), (b_t), (c_t)$ be three sequences such that $a_{t+1} = \left(1- \frac{b_t}{t}\right)a_t + c_t$, $\lim_{t\rightarrow \infty} b_t = b>0$, and $\lim_{t\rightarrow \infty} c_t = c$. Then $\lim_{t\rightarrow \infty} a_t/t$ exists and its value is
\begin{equation}
    \lim_{t\rightarrow \infty} \frac{a_t}{t} = \frac{c}{1+b}.
\end{equation}
\end{lemma}

\section{Proof of Theorem \ref{member degree power-law}}\label{proof of thm 5.2}
\begin{replemma}{member degree power-law}
 Let $\{N(M\cup G, t, \Theta\}$ be a sequence of networks produced by the BGMG model. The red member-degree distribution $M_k(R)$ and the blue member-degree distribution $M_k(B)$ asymptotically follow the power law distributions with the same power; specifically, as $t$ goes to infinity,
\begin{align}
M_k(R)\propto k^{-\left(1+\frac{1}{1-\alpha}\right)}, \,\,\,\,\,\,
M_k(B)\propto k^{-\left(1+\frac{1}{1-\alpha}\right)}.
\end{align}
\end{replemma}
\begin{proof}
For any $k>1$, a red member of degree $k$ at time $t+1$ could have arisen from two scenarios:
\begin{enumerate}
\item at time $t$, it was a red member of degree $k$, and not chosen at time $t+1$;
\item at time $t$, it has size $k=1$ and chosen.
\end{enumerate}
Thus,
\begin{equation}
\mathbb{E}\left(M_{k,t+1}(R) \vert \mathcal{F}_t \right) = M_{k,t}(R)\left(1-(1-\alpha)  \cdot \frac{k}{t}\right) + M_{k-1,t}(R)\left((1-\alpha) \cdot \frac{k-1}{t}\right).
\end{equation}
When $k=1$, a red member of degree $1$ at time $t+1$ could have arisen from:
\begin{enumerate}
\item at time $t$, it was a red member of degree $1$, and not chosen at time $t+1$;
\item a new member joins the network at time $t$.
\end{enumerate}
\begin{equation}
\mathbb{E}\left(M_{1,t+1}(R) \vert \mathcal{F}_t \right) = M_{1,t}(R)\left(1-(1-\alpha) \cdot \frac{1}{t} \right) + \alpha \cdot r.
\end{equation}
Therefore, $M_k(R)$ has the following expressions:
\begin{equation}
M_1(R) = \frac{\alpha \cdot r}{2-\alpha}, \text{ and }M_k(R) = M_{k-1}(R)\frac{(1-\alpha)(k-1)}{1 + (1-\alpha)\cdot k}.
\end{equation}
Hence, $M_k(R) \propto k^{-\left(1+\frac{1}{1-\alpha}\right)}$. Exactly same argument holds for $M_k(B)$.
\end{proof}

\section{Proof of Theorem \ref{lemma_bump}}\label{proof group bump}
\begin{reptheorem}{lemma_bump}
	Following the same notations as in Theorem \ref{group size power-law}. Assume $C_{R,1} < C_{B,1}$. then the group ratio sequence $\{G_k(R)/G_k(B), k \geq 1\}$ has the chasm effect against red, if and only if $k^* > 2$, where
	\begin{equation}
		k^* := \frac{(1+C_{R,1})(1+C_{B,2}) - (1+C_{R,2})(1+C_{B,1})}{C_{R,1} - C_{B,1}},
	\end{equation}
	where
	\begin{align}
    C_{R,2} & := \frac{r (1-\eta) (1-\xi) \frac{1}{\eta} }
    			{1-(1-\rhop_R)\xi (1-\alpha^*)-(1-\rhou_R)(1-\xi) (1-r)}
    			+ \frac{(1-r) (1-\eta)\rhou_B(1-\xi) \frac{1}{\eta} }
    				{1-(1-\rhop_B)\xi \alpha^*-(1-\rhou_B)(1-\xi)r};\\
    C_{B,2} & := \frac{(1-r) (1-\eta) (1-\xi) \frac{1}{\eta}  }
    			{1-(1-\rhop_B)\xi \alpha^* - (1-\rhou_B)(1-\xi) r}
    			+ \frac{r (1-\eta)\rhou_R (1-\xi) \frac{1}{\eta}  }
    				{1-(1-\rhop_R)\xi (1-\alpha^*) - (1-\rhou_R)(1-\xi) (1-r)}.
\end{align}
	Moreover, when $k^* > 2$, the monotonicity of $\{G_k(R)/G_k(B), k \geq 1\}$ changes at $[k^*]$, which is the largest integer smaller than $k^*$.
\end{reptheorem}
\begin{proof}

	We first define 
		\begin{equation}
		g_{ratio}(k):= \frac{G_k(R)/G_k(B)}{G_{k-1}(R)/G_{k-1}(B)}
					= \frac{1 - \frac{1+C_{R,1}}{1 + k C_{R,1} + C_{R,2}}}{1 - \frac{1+C_{B,1}}{1 + k C_{B,1} + C_{B,2}}}.
	\end{equation}
	To see the monotonicity of $\{G_k(R)/G_k(B)\}$, it is sufficient to compare $g_{ratio}(k)$ with 1. Note that,
	\begin{equation}
	\label{eq_bump1}
		g_{ratio}(k) > 1 \,\,\,\,\,\, \Leftrightarrow \,\,\,\,\,\, 
					\frac{1+C_{R,1}}{1 + k C_{R,1} + C_{R,2}} < \frac{1+C_{B,1}}{1 + k C_{B,1} + C_{B,2}}.
	\end{equation}
	With some algebra, we have that
	\begin{equation}
	\label{eq_bump2}
		\frac{1+C_{R,1}}{1 + k C_{R,1} + C_{R,2}} - \frac{1+C_{B,1}}{1 + k C_{B,1} + C_{B,2}}
		= \frac{(k - k^*)(C_{B,1} - C_{R,1})}{(1 + k C_{R,1} + C_{R,2})(1 + k C_{B,1} + C_{B,2})}.
	\end{equation}
	Since the denominator is positive and $C_{B,1} - C_{R,1} > 0$, we therefore have that $g_{ratio}(k) > 1$ for $k < k^*$, and $g_{ratio}(k) < 1$ for $k > k^*$. When $k^* < 2$, $g_{ratio}(k) > 1$ for all $k>1$, and therefore is monotonically increasing.
	\end{proof}

\section{Proof of Lemma \ref{member_ratio_nonmonotone}}\label{proof member_ratio_nonmonotone}
\begin{lemma}
	\label{lemma_binomial_result_general}
	We have that
\begin{equation}
	\lim_{t \rightarrow \infty} r^{(M, G)}_{k, t} (R,R) = r^{(M, G)}_{k} (R,R) := \frac{1 + \sum_{j=2}^k p_{RR,j}}{k},
\end{equation}
\begin{equation}
	\lim_{t \rightarrow \infty} r^{(M, G)}_{k, t} (R,B) = r^{(M, G)}_{k} (R,B) := \frac{1 + \sum_{j=2}^k p_{RB,j}}{k},
\end{equation}
where
\begin{align}
	p_{RR,j} = \frac{p^{(0)}_{RR,j}}{p^{(0)}_{RR,j} + p^{(0)}_{BR,j}}, \,\,\,\,\,\,\,\,\,\,
	p_{RB,j} = \frac{p^{(0)}_{RB,j}}{p^{(0)}_{RB,j} + p^{(0)}_{BB,j}}
\end{align}
\begin{align}
	p^{(0)}_{RR,j} &= r\frac{\xi j + (1-\xi)/\eta}{1-(1-\rhop_R)\xi(1-\alpha^*) - (1-\rhou_R)(1-\xi) (1-r)}, \\
	p^{(0)}_{BR,j} &= (1-r)\frac{\rhop_B \xi j + \rhou_B(1-\xi)/\eta}{1-(1-\rhop_B)\xi \alpha^* - (1-\rhou_B)(1-\xi) r},\\
	p^{(0)}_{RB,j} &= r\frac{\rhop_R \xi j + \rhou_R(1-\xi)/\eta}{1-(1-\rhop_R)\xi(1-\alpha^*) - (1-\rhou_R)(1-\xi) (1-r)}, \\
	p^{(0)}_{BB,j} &= (1-r)\frac{\xi j + (1-\xi)/\eta}{1-(1-\rhop_B)\xi \alpha^* - (1-\rhou_B)(1-\xi) r}.
\end{align}

\end{lemma}

\begin{proof}
We only prove the result for red groups. For a red group $J_R$ with size $j$ at time $t$, Define the events
\begin{align}
	\Gamma_{t,j} &:= \{\text{At time t, an edge between a member and $J_R$ is added}\}\\
	\Gamma_{t,j,R} &:= \{\text{At time t, an edge between a red member and $J_R$ is added}\}, \\
	\Gamma_{t,j,B} &:= \{\text{At time t, an edge between a blue member and $J_R$ is added}\}.
\end{align}
We then have that, by the definition of our model and Lemma \ref{lemma_pre_convergence},
\begin{align}
	\mathbb{P}(\Gamma_{t,j,R})\cdot t
	=  \frac{(\alpha r + (1-\alpha) r^{(E, M)}_t(R)) (1-\eta)\left(\xi j + (1-\xi) \frac{1}{r^{(G)}_t(R)+r^{(G)}_t(B)}\right)}
    	{1-(1-\rhop_R)\xi (1-r^{(E, G)}_t(R)) - (1-\rhou_R)(1-\xi) \frac{r^{(G)}_t(B)}{r^{(G)}_t(R)+r^{(G)}_t(B)}}
    	\rightarrow  p^{(0)}_{RR,j}, \notag
\end{align}
where the convergence is for $t \rightarrow \infty$. Similarly, we have that
\begin{align}
	\mathbb{P}(\Gamma_{t,j,B})\cdot t
	& \rightarrow p^{(0)}_{BR,j}.
\end{align}
By the Bayes formula, we see that as $t \rightarrow \infty$,
\begin{align}
	&\mathbb{P}(\Gamma_{t,j,R} \mid \Gamma_{t,j}) = \frac{\mathbb{P}(\Gamma_{t,j,R})}{\mathbb{P}(\Gamma_{t,j,R}) + \mathbb{P}(\Gamma_{t,j,B})} \rightarrow p_{RR,j}.
\end{align}

We uniformly choose a red group $J_{k,R}$ at time $t$, among the red groups with size $k$. Define $ t_1 < \ldots < t_k $, such that $t_j$ is the time a new member $M_j$ joins the chosen group $J_{k,R}$. By the construction of our model, we must have that $t_1$ is the time the group is created, and the first member is of color red. For each $j>2$, at $t_j$ this group has size $j-1$. Note that as the graph size $t$ goes to infinity, since $J_{k,R}$ is uniformly chosen, we must have that $t_j \rightarrow \infty$ for each $j$. Therefore we have that
\begin{equation}
	\mathbb{E}[\text{number of red members in } J_{k,R}] \rightarrow 1 + \sum_{j=2}^k p_{RR,j}.
\end{equation}
Recall that $G_{k,t}(R)$ is the number of red groups at time $t$. Since $J_{k,R}$ is uniformly chosen, we have that
\begin{equation}
	\frac{\mathbb{E}[\text{number of red members in red groups with size }k]}
		{G_{k,t}(R)}
	\rightarrow 1 + \sum_{j=2}^k p_{RR,j},
\end{equation}
which finishes the proof with the fact that
\begin{equation}
	r^{(M, G)}_{k, t} (R,R) = \frac{\mathbb{E}[\text{number of red members in red groups with size }k]}{kG_{k,t}(R)}
\end{equation}
\end{proof}

\begin{cor}
	We have that,
	\begin{align}\label{member ratio cor}
		\lim_{t \rightarrow \infty} r^{(M, G)}_{k, t} (R)&= \frac{G_k(R) r^{(M, G)}_{k}(R,R) + G_k(B) r^{(M, G)}_{k}(R,B)}{G_k(R)+ G_k(B)}.
	\end{align}	
\end{cor}

\begin{replemma}{member_ratio_nonmonotone}
For the red member ratios within groups with size 1, and within groups with size goes to infinity, we have:
\begin{itemize}
	\item For groups with size 1,
	\begin{align}
		\lim_{t \rightarrow \infty} r^{(M, G)}_{1, t} (R)&= \frac{G_1(R)}{G_1(R) + G_1(B)}
		= \frac{1+C_{B,1} + C_{B,2}}{2+C_{R,1} + C_{R,2}+C_{B,1} + C_{B,2}}.
	\end{align}
	\item For groups with size goes to infinity, assume $C_{R,1} < C_{B,1}$,
	\begin{equation}
		\label{rm_ratio_infty}
		\lim_{k \rightarrow \infty} \lim_{t \rightarrow \infty} r^{(M, G)}_{k, t} (R)
		= r^{(M, G)}(R),
	\end{equation}
	where $r^{(M, G)}$ is defined as 
	\begin{equation}\label{rm_fit1}
		r^{(M, G)}(R) = \frac{q_{RB}}{q_{RB} + q_{BB}}, 
	\end{equation}
	with
	\begin{align}
			q_{RB} &= r\rhop_R \left(1-(1-\rhop_B)\xi \alpha^* - (1-\rhou_B)(1-\xi) r\right),\\
			q_{BB} &= (1-r)\left(1-(1-\rhop_R)\xi(1-\alpha^*) - (1-\rhou_R)(1-\xi) (1-r)\right).
	\end{align}
\end{itemize}
\end{replemma}

\begin{proof}
	Following Corollary \ref{lemma_binomial_result_general}, we have that for $r^{(M, G)}_{1, t}(R)$, since there is exactly 1 red (blue) member in red (blue) group with size 1, so we have that
	\begin{equation}
		r^{(M, G)}_{1, t}(R) = \frac{G_{1,t}(R)}{G_{1,t}(R) + G_{1,t}(B)}
		\rightarrow \frac{G_1(R)}{G_1(R) + G_1(B)}.
	\end{equation} 
	For the case where $k \rightarrow \infty$, since we assume that there is a glass-ceiling effect against red members, as $k \rightarrow \infty$, we have that $G_k(R)/G_k(B) \rightarrow 0$. That is, we only need focus on blue groups. 

As $j \rightarrow \infty$, it is easy to check that
\begin{align}
	\lim_{j \rightarrow \infty}p_{RB,j} = r^{(M, G)},
\end{align}
and consequently we have that
\begin{align}
	\lim_{k \rightarrow \infty}r^{(M, G)}_{k}(R,B) = r^{(M, G)},
\end{align}
which finishes the proof.
\end{proof}

\section{Proof of Theorem \ref{ads_theorem}}\label{proof_of_ads}

\begin{reptheorem}{ads_theorem}
	Assume the red member ratios for very small and large groups are smaller than the average red member ratio $r$ in the network. There exist $0 < k_A^{lower} \leq k_A^{upper}$, such that
	\begin{itemize}
		\item For $k_A > k_A^{upper}$, $r^{(A)}(k_A) < r$;
		\item For $k_A < k_A^{lower}$, $r^{(A)}(k_A) > r$.
	\end{itemize}
\end{reptheorem}

\begin{proof}
	Under our assumption, there exists some $0 < k_A^{lower} \leq k_A^{upper}$, such that $\lim_{t \rightarrow \infty} r^{(M, G)}_{k, t} < r$ for $k > k_A^{upper}$ and $k < k_A^{upper}$. Therefore, if $k_A > k_A^{upper}$, for all groups where ads are placed, their limiting red member ratios are less than $r$. Consequently, we must have $r^{(A)}(k_A) < r$. On the other hand, if $k_A < k_A^{lower}$, for the groups where ads are not placed, their limiting red member ratios are less than $r$, which means that among all the people not seeing the ads, the red member ratio is less than $r$. It further implies that among all the people seeing the ads, red member ratio is greater than $r$, that is, $r^{(A)}(k_A) > r$.

\end{proof}

\section{Proof of Theorem \ref{fake news theorem}}\label{proof_of_fn}
\begin{reptheorem}{fake news theorem}
	Assume the red member ratios for very small and large groups are smaller than the average red member ratio $r$ in the network. 	There exists $0 < \theta^{lower} <  \theta^{upper} < 1$, such that
	\begin{itemize}
		\item For $\theta > \theta^{upper}$, $r^{(D)}(\theta) > r$;
		\item For $\theta < \theta^{lower}$, $r^{(D)}(\theta) < r$.
	\end{itemize}
\end{reptheorem}
\begin{proof}
	Under our assumption, there exist some $0 < k_F^{lower} < k_F^{upper}$, such that for $k > k_F^{upper}$ and $k < k_F^{upper}$
	$$ \frac{G_k(R)}{G_k(R)+G_k(B)} < r.$$
	As $\theta \rightarrow 0$, by the assumption (\ref{eq_h_assumption}), we see that
	\begin{equation}
		\lim_{\theta \rightarrow 0} \frac{\sum_{k \leq k_F^{upper}} G_k(R)h(k, \theta)}{\sum_{k > k_F^{upper}} G_k(R)h(k, \theta)} = 0, \,\,\,\,\,\,
		\lim_{\theta \rightarrow 0} \frac{\sum_{k \leq k_F^{upper}} G_k(B)h(k, \theta)}{\sum_{k > k_F^{upper}} G_k(B)h(k, \theta)} = 0,
	\end{equation}
	which implies that $\lim_{\theta \rightarrow 0} r^{(D)}(\theta)/r_{k_F^{upper}}(\theta) = 1$, where
	\begin{equation}
		r_{k_F^{upper}}(\theta) : = \frac{\sum_{k > k_F^{upper}} G_k(R)h(k, \theta)}{\sum_{k > k_F^{upper}} (G_k(R)h(k, \theta) + G_k(B)h(k, \theta))} < r.
	\end{equation}
	Hence we see that there exists $\theta^{lower}>0$, such that $r^{(D)}(\theta) < r$ for $\theta < \theta^{lower}$.
	
	As $\theta \rightarrow 1$, by the assumption (\ref{eq_h_assumption}), we see that
	\begin{equation}
		\lim_{\theta \rightarrow 0} \frac{\sum_{k \geq k_F^{lower}} G_k(R)(1-h(k, \theta))}{\sum_{k < k_F^{lower}} G_k(R)(1-h(k, \theta))} = 0, \,\,\,\,\,\,
		\lim_{\theta \rightarrow 0} \frac{\sum_{k \geq k_F^{lower}} G_k(B)(1-h(k, \theta))}{\sum_{k < k_F^{lower}} G_k(B)(1-h(k, \theta))} = 0,
	\end{equation}
	which implies that 
	\begin{equation}
	\label{eq_fk_assist_1}
		\lim_{\theta \rightarrow 1} \frac{\sum_{k \geq 1} G_k(R)(1-h(k, \theta))}
								{\sum_{k \geq 1} (G_k(R)+G_k(B))(1-h(k, \theta))}
		=\lim_{\theta \rightarrow 1} \frac{\sum_{k < k_F^{lower}} G_k(R)(1-h(k, \theta))}
								{\sum_{k < k_F^{lower}} (G_k(R)+G_k(B))(1-h(k, \theta))}
		<r.
	\end{equation}
	Note that
	\begin{equation}
		\frac{\sum_{k \geq 1} G_k(R)}
						{\sum_{k \geq 1} (G_k(R)+G_k(B))} 
		= \frac{\lim_{t\rightarrow \infty} r^{(G)}_t(R)}{\lim_{t\rightarrow \infty} (r^{(G)}_t(R)+r^{(G)}_t(B))}=r,
	\end{equation}
	and thus (\ref{eq_fk_assist_1}) leads to
	\begin{equation}
		\lim_{\theta \rightarrow 1} \frac{\sum_{k \geq 1} G_k(R)h(k, \theta)}
								{\sum_{k \geq 1} (G_k(R)+G_k(B))h(k, \theta)}
		>r.
	\end{equation}
	 Consequently, there exists $0<\theta^{upper}<1$, such that $r^{(D)}(\theta) > r$ for $\theta > \theta^{upper}$.
\end{proof}

\section{Proof of Lemma \ref{lemma_pre_convergence}}
\begin{lemma}
	\label{lemma_pre_convergence}
	Under the assumption that $\rhop_R, \rhop_B > 0$, we have the following convergence results:
	\begin{itemize}
		\item The proportion of edges coming from red members converges; that is
		\begin{equation}
			\label{eq_rt_converge}
    		\lim_{t \rightarrow \infty} r^{(E, M)}_t(R) = r \,\,\,\, a.s.
    	\end{equation}
    	
    	\item The ratio of red group counts over t converges; that is
    	\begin{equation}
    		\label{eq_red_group_converge}
    	\lim_{t\rightarrow \infty} r^{(G)}_t(R) = r \eta \,\,\,\, a.s.
    	\end{equation}
        	
    	\item The proportion of edges coming from red groups converges; that is
    	\begin{equation}
    		\label{eq_red_group_degree_converge}
    	\lim_{t \rightarrow \infty} r^{(E, G)}_t(R) = \alpha^* \,\,\,\, a.s.
    	\end{equation}
    	where $\alpha^*$ is the unique number $\in (0,1)$ satisfying
    	\begin{equation}
    		\label{eq_alpha_star}
    	\alpha^* = r \eta + \frac{r (1-\eta) (\xi \alpha^* + (1-\xi)r) }
    			{1-(1-\rhop_R)\xi (1-\alpha^*)-(1-\rhou_R)(1-\xi) (1-r)}
    			+ \frac{(1-r) (1-\eta) (\rhop_B \xi \alpha^* + \rhou_B (1-\xi)r)}
    				{1-(1-\rhop_B)\xi \alpha^*-(1-\rhou_B)(1-\xi)r }.
    	\end{equation}
	\end{itemize}
\end{lemma}
    We divide the proof into three parts.
    
    \textbf{Part 1. Proof of (\ref{eq_rt_converge})} 
    Note that $E^{(M)}_t(R)$ is the total degree of red nodes at time $t$. By our model, given $r^{(E, M)}_t(R)$, the total degree of red nodes at time $t+1$ could take two values: $E^{(M)}_t(R)$ and $E^{(M)}_t(R)+1$, with probability $1-\alpha r - (1-\alpha) r^{(E, M)}_t(R)$ and $\alpha r + (1-\alpha) r^{(E, M)}_t(R)$ respectively. Recall that $\mathcal{F}_t$ is the $\sigma-$field containing the information of the graph up to time $t$. Therefore we have that
    \begin{equation}
    	\mathbb{E}\left( E^{(M)}_{t+1}(R)\vert \mathcal{F}_t\right) =
    	E^{(M)}_t(R) + \alpha r + (1-\alpha) r^{(E, M)}_t(R),
    \end{equation}
    which gives
    \begin{equation}
    	\mathbb{E}\left( r^{(E, M)}_{t+1}(R) - r \vert \mathcal{F}_t\right) =
    	\frac{t + (1-\alpha)}{t+1} (r^{(E, M)}_t(R)-r).
    \end{equation}
    Recall that our model starts from $t=2$. Therefore
    \begin{equation}
    	\label{eq_r_E_converge}
    	\mathbb{E}\left( r^{(E, M)}_{t+1}(R) - r\right) =
    	\prod_{i=2}^t \frac{i + (1-\alpha)}{i+1} (r^{(E, M)}_2(R)-r)
    	=  O \left(\exp(-\sum_{i=2}^t \frac{\alpha}{i+1}) \right)
    	= O(t^{- \alpha}).
    \end{equation}
    Next we show a concentration inequality for $r^{(E, M)}_t(R)$. For $T > 0$, we define a Doob martingale, that for $0 \leq t \leq T$,
    \begin{equation}
    	W_t := \mathbb{E}\left( r^{(E, M)}_T(R) - r \vert \mathcal{F}_t\right) = \prod_{i=t}^{T-1} \frac{i + (1-\alpha)}{i+1} (r^{(E, M)}_t(R)-r).
    \end{equation}
    It satisfies that $\{W_t, 2 \leq t \leq T\}$ is a martingale, and $W_T = r^{(E, M)}_T(R)-r$, $W_2 = \mathbb{E}[r^{(E, M)}_T(R)-r]$. Next we bound the difference between $W_t$ and $W_{t-1}$. We have that
    \begin{equation}
    	W_t - W_{t-1} = \prod_{i=t}^{T-1} \frac{i + (1-\alpha)}{i+1} \left((r^{(E, M)}_t(R)-r)- \frac{t-\alpha}{t}(r^{(E, M)}_{t-1}(R)-r)\right).
    \end{equation}
    Since $E^{(M)}_{t}(R)$ could just take two values $E^{(M)}_{t-1}(R)$ and $E^{(M)}_{t-1}(R)+1$, we have that $|r^{(E, M)}_t(R) - r^{(E, M)}_{t-1}(R)| = O (1/t)$. And thus
    \begin{equation}
    	W_t - W_{t-1} = \prod_{i=t}^{T-1} \frac{i + (1-\alpha)}{i+1} O(1/t)
    	= O \left(\exp(-\sum_{i=t}^t \frac{\alpha}{i+1}) \right) O(1/t).
    	= O(t^{-1}(T/t)^{- \alpha}).
    \end{equation}
    Applying the Azuma's inequality \cite{alon2004probabilistic} for martingale, we get that there exist constants $c_1, c_2>0$, such that for any $T, x > 0$, 	
    \begin{equation}
    \label{eq_W_T_tail}
	    \mathbb{P}\left(|W_T - W_2| > x \right)
	    \leq \exp \left( - c_1\frac{x^2}{T^{-2 \alpha} \sum_{j = 1}^T t^{-2+2 \alpha}} \right)
	    \leq \exp \left( - c_2 x^2 T^{\min(1,2\alpha)}/\log{T} \right),
	\end{equation}
	where the last step is because
	\begin{equation}
	T^{-2 \alpha} \sum_{j = 1}^T t^{-2+2 \alpha} = 
		\begin{cases}
			O(T^{-2 \alpha}), &\text{if } 2\alpha<1; \\
			O(T^{-1}\log{T}), &\text{if } 2\alpha=1; \\
			O(T^{-1}), &\text{if } 2\alpha>1. 
		\end{cases}
	\end{equation}
	From (\ref{eq_W_T_tail}) we have that, for any $\epsilon > 0$, the tail probability $$\mathbb{P}\left(|r^{(E, M)}_T(R) - \mathbb{E}[r^{(E, M)}_T(R)]| > \epsilon \right) = \mathbb{P}\left(|W_T - W_2| > \epsilon \right)$$ is summable over $T$. By the Borel Cantelli lemma, we see that $r^{(E, M)}_t(R) - \mathbb{E}[r^{(E, M)}_t(R)] \rightarrow 0$ a.s., which gives our desired result with (\ref{eq_r_E_converge}). Moreover, since we already show that $W_2 = O(T^{- \alpha})$ we have that there exist constants $c_3, c_4 > 0$, for any $x>0$,
    \begin{align}
    \label{eq_re_concentration}
	    \mathbb{P}\left(|W_T| > x \right) &\leq \mathbb{P}\left(|W_T-W_2| > x-|W_2| \right) \notag\\
	    &\leq \exp \left( - c_3 \max(x-|W_2|,0)^2 T^{\min(1,2\alpha)}/\log{T} \right)
	    \leq \exp \left( - c_4  x^2 T^{\min(1,2\alpha)}/\log{T} \right).
	\end{align}

	\textbf{Part 2. Proof of (\ref{eq_red_group_converge})} 
	According to our model, at each time $t$, with probability $\alpha r^{(E, M)}_t(R) + (1-\alpha) r$ a red member adds an edge, and with probability $\eta$ the edge is added by creating a new red group. Let's consider the number of red groups in the model conditioned on a given sequence $\{r^{(E, M)}_t(R), t>0\}$. 
	
	For each $t$, there are two cases: (1) case 1, $E^{(M)}_{t+1}(R)=E^{(M)}_{t}(R)+1$, in this case a red member adds an edge at time $t$, and conditioned on $\{r^{(E, M)}_t(R), t>0\}$, the probability that this edge is added by creating a new red group is $\eta$: this is because how this edge is added does not influence the value of $r^{(E, M)}_{t+1}(R)$ and thus does not influence $\{r^{(E, M)}_t(R), t>0\}$, and hence whether we condition on $\{r^{(E, M)}_t(R), t>0\}$ or not does not change the probability that the new edge is added by creating a group; (2) case 2, $E^{(M)}_{t+1}(R)=E^{(M)}_{t}(R)$, in this case a blue member adds an edge at time $t$, and no red group is created. 
	
	We also have that, the events $\{$a red group is created at time $t \}$ over different $t$ are independent conditioned on $\{r^{(E, M)}_t(R), t>0\}$. Intuitively, it is because the probability of $\{$a red group is created at time $t$ $\}$ only depends on the value of $E^{(M)}_{t+1}(R)-E^{(M)}_{t}(R)$. The independence claim could also be verified by writing out the posterior distribution of those events given $\{r^{(E, M)}_t(R), t>0\}$.
	
	Recall that our initial condition is that there is a red (blue) member with an edge to a red (blue) group, in total two members and two groups. Therefore, given $\{r^{(E, M)}_t(R), t>0\}$, the number of red groups $G_t(R)$ satisfies that, $G_t(R)-1$ follows a Binomial distribution $B(E^{(M)}_t(R) - 1, \eta)$. Therefore, by Hoeffding's inequality (\cite{boucheron2013concentration}), we have that for any $x > 0$,
	\begin{equation}
		\mathbb{P}\left(\left\vert \frac{G_t(R)-1}{E_t^{(M)}(R) - 1} - \eta\right\vert > x \mid \{r^{(E, M)}_t(R), t>0\}\right)
		\leq 2 \exp \left( - 2 (E_t^{(M)}(R) - 1) x^2\right),
	\end{equation}
	which further implies that
	\begin{small}
	\begin{equation}
	\label{eq_rG_converge_1}
		\mathbb{P}\left(\left\vert r^{(G)}_t(R) - \eta r^{(E, M)}_t(R) + \frac{\eta-1}{t} \right\vert > x \mid \{r^{(E, M)}_t(R), t>0\}\right)
		\leq 2 \exp \left( - \frac{2x^2 t^2}{E^{(M)}_t(R) - 1}\right)
		\leq 2 \exp \left( - \frac{2x^2 t^2}{t - 1}\right).
	\end{equation}
	\end{small}

	Hence for any $\epsilon > 0$, with probability 1 the tail probability $$\mathbb{P}\left(\left|r^{(G)}_t(R) - \eta r^{(E, M)}_t(R) + \frac{\eta-1}{t}\right| > \epsilon \mid \{r^{(E, M)}_t(R), t>0\} \right) $$ is summable over $t$. By the Borel Cantelli lemma, we see that $r^{(G)}_t(R) - \eta r^{(E, M)}_t(R) + (\eta-1)/t$ goes to 0 a.s., which gives $r^{(G)}_t(R) \rightarrow r \eta$ a.s. with the fact that $r^{(E, M)}_t(R) \rightarrow r$ a.s..
	
	Moreover, since by the triangle inequality
	\begin{equation}
		\left\vert r^{(G)}_t(R) - r \eta \right\vert 
		\leq \left\vert r^{(G)}_t(R) - \eta r^{(E, M)}_t(R) + \frac{\eta-1}{t} \right\vert+\frac{1-\eta}{t} + |\eta r^{(E, M)}_t(R)-\eta r|,
	\end{equation}
	we see that for any $x>0$, 
	\begin{equation}
		\left\{ \left\vert r^{(G)}_t(R) - r \eta \right\vert>x \right\}
		\subset
		\left\{ \left\vert r^{(G)}_t(R) - \eta r^{(E, M)}_t(R) + \frac{\eta-1}{t} \right\vert > \frac{x}{2} - \frac{1-\eta}{t} \right\}
		\cup
		\left\{|\eta r^{(E, M)}_t(R)-\eta r| > \frac{x}{2}\right\}.
	\end{equation}
	Therefore, for the unconditional tail probability of $r^{(G)}_t(R) - r \eta$, we have
	\begin{align}
		\mathbb{P}\left(\left\vert r^{(G)}_t(R) - r \eta \right\vert > x \right)
		&\leq \mathbb{P}\left(\left\vert r^{(G)}_t(R) - \eta r^{(E, M)}_t(R) + \frac{\eta-1}{t} \right\vert > \frac{x}{2} - \frac{1-\eta}{t} \right)
		+ \mathbb{P}\left( |r^{(E, M)}_t(R)-r|\geq \frac{x}{2 \eta}\right) \notag.
	\end{align}
	Note that the unconditional version of (\ref{eq_rG_converge_1}) also holds, since the right hand side does not depend on $\{r^{(E, M)}_t(R), t>0\}$. Together with (\ref{eq_re_concentration}), we have that there exists a constant $c_5 > 0$, such that for any $x>0$,
	\begin{align}
	\label{eq_rg_concentration}
		\mathbb{P}\left(\left\vert r^{(G)}_t(R) - r \eta \right\vert > x \right)
		&\leq 2 \exp \left( - 2 \frac{(x/2 - (1-\eta)/{t})^2 t^2}{t - 1}\right)
		+ \exp \left( - c_4 \left(\frac{x}{2 \eta}\right)^2 t^{\min(1,2\alpha)}/\log{t} \right) \notag \\
		&\leq \exp \left( - c_5  x^2 t^{\min(1,2\alpha)}/\log{t} \right).
	\end{align}
	\textbf{Part 3. Proof of (\ref{eq_red_group_degree_converge}) and (\ref{eq_alpha_star})} 
	Recall that $E^{(G)}_t(R)$ is the total degree of red groups. Similar to part 1, at each time $t+1$, $E^{(G)}_{t+1}(R)$ could take two values: $E^{(G)}_t(R)$ and $E^{(G)}_t(R)+1$. By our definition of the model, one can verify that, the probability that $E^{(G)}_{t+1}(R) = E^{(G)}_t(R)+1$ is a function of $r^{(E, G)}_t(R), r^{(E, M)}_t(R), r^{(G)}_t(R), r^{(G)}_t(B)$, which we denote by $H(r^{(E, G)}_t(R)$, $r^{(E, M)}_t(R)$, $r^{(G)}_t(R)$, $r^{(G)}_t(B))$, and it takes the following expression
		\begin{align}
    		H(x, y, z, w) &:= (\alpha r + (1-\alpha) y)\eta + \frac{(\alpha r + (1-\alpha) y) (1-\eta) (\xi x + (1-\xi)\frac{z}{z+w}) }
    				{1-(1-\rhop_R)\xi (1-x) - (1-\rhou_R)(1-\xi) \frac{w}{w+z}}\\
    			&+  \frac{(\alpha (1-r) + (1-\alpha) (1-y)) (1-\eta)(\rhop_B \xi x + \rhou_B(1-\xi)\frac{z}{z+w})}
    			{1-(1-\rhop_B)\xi x - (1-\rhou_B)(1-\xi) \frac{z}{w+z}}.
    	\end{align}

    	We already see that $r^{(E, M)}_t(R) \rightarrow r$ a.s. and $r^{(G)}_t(R) \rightarrow r \eta$ a.s. Similarly, $r^{(G)}_t(B) \rightarrow (1-r) \eta$ a.s. We denote
    	\begin{align}
    		F(x) &= H(x, r, r \eta, (1-r) \eta) \\
    		&= r \eta + \frac{r (1-\eta) (\xi x + (1-\xi)r) }
    			{1-(1-\rhop_R)\xi (1-x)-(1-\rhou_R)(1-\xi) (1-r)}
    			+ \frac{(1-r) (1-\eta) (\rhop_B \xi x + \rhou_B (1-\xi)r)}
    				{1-(1-\rhop_B)\xi x-(1-\rhou_B)(1-\xi)r }.
    	\end{align}
    We have the following Lemma, whose proof is deferred to Appendix \ref{axillary section}.
    \begin{lemma}
    \label{lem_F_derivative}
	Under the assumption that $\rhop_R, \rhop_B > 0$, $F(x)$ satisfies
	\begin{enumerate}
	\item $F(x)$ has exactly one fixed point, denoted $\alpha^*$, in $[0,1]$;
	\item There exists $\gamma < 1$, such that for any $x \in (0,1)$
		\begin{equation}
			|F(\alpha^*) - x| \leq \gamma |\alpha^* - x|.
		\end{equation}
	\end{enumerate}
	\end{lemma}
    
    Let $\alpha^* \in (0,1)$ be the number satisfying that $F(\alpha^*) = \alpha^*$. Similar to part 1, we can calculate the second moment of $\alpha_{t+1}-\alpha^*$
    \begin{align}
    	\mathbb{E}\left( (\alpha_{t+1}-\alpha^*)^2\vert \mathcal{F}_t\right) =&
    	\left(\frac{tr^{(E, G)}_t(R)}{t+1} - \alpha^*\right)^2 (1-H(r^{(E, G)}_t(R), r^{(E, M)}_t(R), r^{(G)}_t(R), r^{(G)}_t(B)))\nonumber \\
    	&+ \left(\frac{tr^{(E, G)}_t(R)+1}{t+1} - \alpha^*\right)^2 H(r^{(E, G)}_t(R), r^{(E, M)}_t(R), r^{(G)}_t(R), r^{(G)}_t(B)) \\
    	 =& I^{(1)}_t + I^{(2)}_t + I^{(3)}_t \label{eq_decomp_I_123},
    \end{align}
    where
    \begin{align}
    	I^{(1)}_t =& \frac{t^2 (r_t^{(E,G)}(R)-\alpha^*)^2 + 2t(r_t^{(E,G)}(R)-\alpha^*)\left((1-\alpha^*)F(r^{(E, G)}_t(R))\right)
    				-\alpha^*\left(1-F(r^{(E, G)}_t(R))\right)}{(t+1)^2}, \\
    	I^{(2)}_t =& \frac{(\alpha^*)^2\left(1-H(r^{(E, G)}_t(R), r^{(E, M)}_t(R), r^{(G)}_t(R), r^{(G)}_t(B))\right)}{{(t+1)^2}} \nonumber \\
    	&+ \frac{(1-\alpha^*)^2 			H(r^{(E, G)}_t(R), r^{(E, M)}_t(R), r^{(G)}_t(R), r^{(G)}_t(B))}{(t+1)^2}, \\	
    	I^{(3)}_t =& \frac{2t(r_t^{(E,G)}(R)-\alpha^*)\left((1-\alpha^*)\Delta(r^{(E, G)}_t(R), r^{(E, M)}_t(R), r^{(G)}_t(R), r^{(G)}_t(B))\right.}{(t+1)^2} \nonumber \\
    	&-\frac{\left.\alpha^*(1-\Delta(r^{(E, G)}_t(R), r^{(E, M)}_t(R), r^{(G)}_t(R), r^{(G)}_t(B)))\right)}{(t+1)^2},		
    \end{align}
    with
    \begin{align}
    	\Delta(r^{(E, G)}_t(R), r^{(E, M)}_t(R), r^{(G)}_t(R), r^{(G)}_t(B)):=
    		H(r^{(E, G)}_t(R), r^{(E, M)}_t(R), r^{(G)}_t(R), r^{(G)}_t(B))-F(r^{(E, G)}_t(R)).
    \end{align}
	
	We need the following lemmas.

	\begin{lemma}
		\label{lem_H_derivative}
		Under the assumption that $\rhop_R, \rhop_B > 0$, there exists $c_0 > 0$, such that for any $x,y,z,w \in (0,1)$
		\begin{equation}
			|\Delta(x,y,z,w)| < c_0 (|y-r| + |z-r \eta| + |w-(1-r) \eta|).
		\end{equation}
	\end{lemma}
	We ignore the proof of Lemma \ref{lem_H_derivative}, since it could be directly verified by checking that, the first derivatives of $H(\cdot)$ are bounded.
	\begin{lemma}
		\label{lem_delta_summable}
		We have that,
		\begin{equation}
		\label{eq_delta_summable_1}
			\lim_{T \rightarrow \infty} \sum_{t=1}^T \frac{|r^{(E, M)}_t(R)-r| + |r^{(G)}_t(R)-r \eta| + |r^{(G)}_t(B)-(1-r) \eta|}{t}
			< \infty, \,\,\,\, a.s.
		\end{equation}
		and
		\begin{equation}
		\label{eq_delta_summable_2}
			\lim_{T \rightarrow \infty} \sum_{t=1}^T \frac{\mathbb{E} [|r^{(E, M)}_t(R)-r| + |r^{(G)}_t(R)-r \eta| + |r^{(G)}_t(B)-(1-r) \eta| ]}{t}
			< \infty.
		\end{equation}
	\end{lemma}
	The proof of Lemma \ref{lem_delta_summable} is deferred to Appendix \ref{axillary section}.
	
	Next we bound $I^{(1)}_t, I^{(2)}_t, I^{(3)}_t$. For $I^{(1)}_t$, by Lemma \ref{lem_F_derivative} and the fact that $F(\alpha^*) = \alpha^*$, we can have that
	\begin{align}
    	I^{(1)}_t &= \frac{t^2 (r_t^{(E,G)}(R)-\alpha^*)^2 + 2t(r_t^{(E,G)}(R)-\alpha^*)(F(r_t^{(E,G)}(R))-\alpha^*)}{(t+1)^2} \\
    	& \leq (r_t^{(E,G)}(R)-\alpha^*)^2 \left( 1 - \frac{2t (1-\gamma)}{(t+1)^2}\right)
    \end{align}
    For $I^{(2)}_t$, since $H(\cdot)$ is bounded by $1$, obviously for some constant $c_6 > 0$, we have
    \begin{equation}
    	\label{eq_I_2_bound}
    	I^{(2)}_t \leq \frac{c_6}{(t+1)^2}.
    \end{equation}
    With the expression of $I^{(3)}_t$, it is easy to see that for some $c_7 > 0$,
    \begin{equation}
    	|I^{(3)}_t| < c_7 \frac{\Delta(r^{(E, G)}_t(R), r^{(E, M)}_t(R), r^{(G)}_t(R), r^{(G)}_t(B)))}{t}.
    \end{equation}
    Further by Lemma \ref{lem_H_derivative} and Lemma \ref{lem_delta_summable}, we have that
    \begin{equation}
    \label{eq_I_3_bound}
			\lim_{T \rightarrow \infty} \sum_{t=1}^T I^{(3)}_t
			< \infty \,\,\,\, a.s., \,\,\,\,\,\, \text{and} \,\,\,\,\,\,\,\,
			\lim_{T \rightarrow \infty} \sum_{t=1}^T \mathbb{E} [I^{(3)}_t]
			< \infty.
		\end{equation}
	
	We need the following Lemma, whose proof is deferred to Appendix \ref{axillary section}.
	\begin{lemma}
	\label{lem3.2}
	Let $(a_t), (b_t), (c_t)$ be three positive sequences such that $a_{t+1} \leq b_t a_t + c_t$, $b_t < 1$, $\lim_{t\rightarrow \infty} \prod_{i=1}^t b_i = 0$, and $\lim_{t\rightarrow \infty} \sum_{i=1}^t c_i < \infty$. Then $\lim_{t\rightarrow \infty} a_t = 0$. 
	\end{lemma}
	
	Let 
	\begin{equation}
		Z_t = (r_t^{(E,G)}(R)-\alpha^*)^2,\,\, \,\,\,\, a_t = \mathbb{E}(Z_t), \,\,\,\,\,\, b_t = 1-2 (1-\gamma) t/(t+1)^2, \,\,\,\,\,\, c_t = \mathbb{E} [I^{(2)}_t+I^{(3)}_t].
	\end{equation}
	By taking expectation in eq (\ref{eq_decomp_I_123}), we have that $a_{t+1} \leq b_t a_t + c_t$. It is direct to check the conditions $b_t < 1$, $\lim_{t\rightarrow \infty} \prod_{i=1}^t b_i = 0$. By (\ref{eq_I_2_bound}) and (\ref{eq_I_3_bound}), we have $\lim_{t\rightarrow \infty} \sum_{i=1}^t c_i < \infty$. And thus from Lemma \ref{lem3.2} we know that
	\begin{equation}
		\lim_{t\rightarrow \infty}\mathbb{E}(Z_t) = 0.
	\end{equation}
	Since our goal is equivalent to show that $Z_t \rightarrow 0$ a.s., we claim that it is enough to have that, $Z_t$ converges to a limit random variable almost surely as $t \rightarrow \infty$. This is because, assuming that $\lim_{t\rightarrow\infty} Z_t$ exists a.s., since $Z_t$ is bounded, by the bounded convergence theorem, we have $\mathbb{E}(\lim_{t\rightarrow\infty} Z_t)=0$. Since $Z_t \geq 0$, its limit must be nonnegative, and therefore $\lim_{t\rightarrow\infty} Z_t$ must equal 0 a.s., due to the fact that its expectation is 0.
	
	Now we show that $\lim_{t\rightarrow\infty} Z_t$ exists a.s., by checking that $\{Z_t\}$ is an \emph{almost supermartingle}, since by \cite{robbins1971convergence}, every \emph{almost supermartingle} converges to a limit random variable almost surely. By \cite{robbins1971convergence}, to make $\{Z_t\}$ an \emph{almost supermartingle}, we just need to check that $\lim_{T \rightarrow \infty} \sum_{t=1}^T I^{(2)}_t+I^{(3)}_t
			< \infty \,\,\,\, a.s.$, which we have already proved. Therefore the proof is finished.

\section{Proofs of Axillary Lemmas}\label{axillary section}
\subsection{Proof of Lemma \ref{lem3.2}}
\begin{proof}
	It is enough to show that, for any $\epsilon > 0$, there exists $T > 0$, such that $a_t < \epsilon$ for all $t > T$. First, since $c_t$ is summable, we can find $T_1 > 0$, such that $\sum_{t>T_1} c_t < \epsilon / 2$. Also, since $\lim_{t\rightarrow \infty} \prod_{i=1}^t b_i = 0$, we can find a $T_2 > T_1$, such that $\prod_{i=T_1+1}^{t-1} b_i \cdot (a_1 + \sum_{i>0} c_i) < \epsilon / 2$ for all $t > T_2$. We claim that $T_2$ is the desired $T$. Without the loss of generality, in the rest we denote $c_0 = a_1$. By induction, it is not hard to have the following expression for $a_t$
	\begin{align}
		a_t = \prod_{i=1}^{t-1}b_i c_0 + \prod_{i=2}^{t-1}b_i c_1 + \prod_{i=3}^{t-1}b_i c_2 + \ldots + c_{t-1}
			= \sum_{s = 0}^{t-1} \prod_{i=s+1}^{t-1}b_i c_s.
	\end{align}
	We can further decomposition the summation on the right hand side into two parts, according to $s \leq T_1$ and $s > T_1$. Now, for any $t > T_2$, for the first part, by our choice of $T_2$, and the fact that $b_i < 1$, we have that
	\begin{align}
		\sum_{s = 0}^{T_1} \prod_{i=s+1}^{t-1}b_i c_s \leq \sum_{s = 0}^{T_1} \prod_{i=T_1+1}^{t-1}b_i c_s
		= \prod_{i=T_1+1}^{t-1} \sum_{s = 0}^{T_1}c_s < \epsilon / 2.
	\end{align}
	For the second part, by our choice of $T_1$ and the fact that $b_i < 1$, we simply have that
	\begin{align}
		\sum_{s = T_1 + 1}^{t-1} \prod_{i=s+1}^{t-1}b_i c_s 
		\leq \sum_{s = T_1 + 1}^{t-1} c_s < \epsilon / 2.
	\end{align}
	Combine the above two inequalities, with the fact that $\epsilon$ is arbitrary, we finish the proof.
\end{proof}

\subsection{Proof of Lemma \ref{lem_delta_summable}}
\begin{proof}
First, it enough to show (\ref{eq_delta_summable_2}), since if it holds, by the monotone convergence theorem, we have
	\begin{align}
			& \mathbb{E}\left[\lim_{T \rightarrow \infty} \sum_{t=1}^T \frac{|r^{(E, M)}_t(R)-r| + |r^{(G)}_t(R)-r \eta| + |r^{(G)}_t(B)-(1-r) \eta|}{t}\right]\\
			= &\lim_{T \rightarrow \infty} \sum_{t=1}^T \frac{\mathbb{E} [|r^{(E, M)}_t(R)-r| + |r^{(G)}_t(R)-r \eta| + |r^{(G)}_t(B)-(1-r) \eta| ]}{t}
			< \infty,		
	\end{align}
	which directly implies (\ref{eq_delta_summable_1}).

We claim that, for a stochastic process $\{w_t, t>0\}$, in order to show that $\lim_{T \rightarrow \infty} \sum_{t=1}^T \mathbb{E} [|w_t| ]/t < \infty$,
it is enough to have that, for some $\delta, c > 0$, for any $x > 0$
\begin{equation}
\label{eq_general_tail}
		\mathbb{P}\left(\left\vert w_t \right\vert > x \right)
		\leq  \exp \left( - cx^2 t^{\delta}\right).
	\end{equation}
It is because (\ref{eq_general_tail}) implies that $\mathbb{E}[w_t] = O(t^{-\delta/2})$, which makes $\mathbb{E}[w_t]/t$ summable.

By (\ref{eq_rg_concentration}) and (\ref{eq_re_concentration}), we see that $r^{(E, M)}_t(R), r^{(G)}_t(R)$ satisfies the tail bound (\ref{eq_general_tail}). Also $r^{(G)}_t(B)$ satisfies, since it has the same behavior as $r^{(G)}_t(R)$. The proof is finished.

\end{proof}

\subsection{Proof of Lemma \ref{lem_F_derivative}}

\begin{proof}
We define $K(x)$ as 
\begin{equation}
 (F(x)-x) \left(1-(1-\rhop_R)\xi (1-x)-(1-\rhou_R)(1-\xi) (1-r)\right) 
				 \left((1-(1-\rhop_B)\xi x-(1-\rhou_B)(1-\xi)r\right).
\end{equation}
By the definition of $F(x)$, it is easy to see that $K(x)$ is a degree 3 polynomial, with a negative coefficient for $x^3$ term. Therefore, $\lim_{x \rightarrow -\infty}K(x) = -\infty$ and $\lim_{x \rightarrow \infty}K(x)= \infty$. Since a degree 3 polynomial at most have 3 real roots, if we have $K(0)>0$ and $K(1)<0$, then obviously $K(x)$ has exact one root in $(0,1)$. 
Moreover, for $x \in [0,1]$, since $\rhop_R, \rhop_B > 0$
	\begin{align}
	\label{eq_K_F_ratio}
		K(x)/F(x) &> \left(1-(1-\rhop_R)\xi-(1-\rhou_R)(1-\xi) (1-r)\right) 
				 \left((1-(1-\rhop_B)\xi-(1-\rhou_B)(1-\xi)r\right) \\
				 &> \left(1-\xi-(1-\xi) (1-r)\right) 
				 \left((1-\xi-(1-\xi)r\right) \\
				 &= (1-\xi) r (1-\xi) (1-r) \geq 0,
	\end{align}
	which implies that $F(x)-x$ and $K(x)$ share the same sign in $(0,1)$. Hence if $K(0)>0$ and $K(1)<0$, we have that $F(x)-x$ has exact one root $\alpha^*$ in $(0,1)$. Moreover, for $x \in (0,1)$, $F(x)-x < 0$ if $x > \alpha^*$, $F(x)-x < 0$ if $x > \alpha^*$. This implies that $F(x)-F(\alpha^*) < x - \alpha^*$ if $x > \alpha^*$, and $F(x)-F(\alpha^*) > x - \alpha^*$ if $x < \alpha^*$, which leads to the fact that for $x \in [0,1]$
	\begin{equation}
		0 < \left \vert \frac{F(x) - F(\alpha^*)}{ x - \alpha^*} \right\vert < 1.
	\end{equation}
	One can check that $|F'(\alpha^*)|<1$. Taking supreme over $x$ in the above inequality, since $|(F(x) - F(\alpha^*))/(x - \alpha^*)|$ is a continuous function, the supreme is achieved at some point $x_0$. If $x_0 != \alpha^*$, we can set $\gamma = |(F(x_0) - F(\alpha^*))/(x_0 - \alpha^*)| < 1$; if $x_0 = \alpha^*$, we can set $\gamma = |F'(\alpha^*)| < 1$. The proof is finished.
	
	\end{proof}

\end{document}